\documentclass[letterpaper,11pt]{article}

\usepackage{lmodern}
\usepackage{braket}
\usepackage{physics}
\usepackage{amsfonts,amsmath,amssymb,amsthm}
\usepackage{mathtools}

\usepackage{bm}
\usepackage{bbm}
\usepackage{graphicx}
\usepackage{ascmac}
\usepackage{mathrsfs}
\usepackage{float}
\usepackage{url}
\usepackage{qcircuit}
\usepackage[numbers,sort&compress]{natbib}

\usepackage{algorithmic}
\usepackage{algorithm}

\theoremstyle{plain}
\newtheorem{thm}{Theorem}
\newtheorem{lem}[thm]{Lemma}
\newtheorem{cor}[thm]{Corollary}
\newtheorem*{thm*}{Theorem}
\newtheorem*{lem*}{Lemma}
\newtheorem*{cor*}{Corollary}

\theoremstyle{definition}
\newtheorem{dfn}{Definition}

\theoremstyle{remark}
\newtheorem{rem}[thm]{Remark}
\newtheorem*{rem*}{Remark}

\usepackage[margin=0.94in]{geometry}
\usepackage[colorlinks=true, allcolors=blue]{hyperref}

\newcommand{\green}[1]{\textcolor{black}{#1}}

\usepackage{authblk}

\newcommand{\beginsupplement}{%
  \setcounter{section}{0}%
  \renewcommand{\thesection}{S\arabic{section}}%
  \setcounter{thm}{0}%
  \renewcommand{\thethm}{S\arabic{thm}}%
  \setcounter{dfn}{0}%
  \renewcommand{\thedfn}{S\arabic{dfn}}%
  \setcounter{figure}{0}%
  \renewcommand{\thefigure}{S\arabic{figure}}
  \renewcommand{\theHfigure}{S.\arabic{figure}}%
  \setcounter{table}{0}%
  \renewcommand{\thetable}{S\arabic{table}}%
}

\begin{document}

\title{State-to-Hamiltonian conversion with a few copies}

\author[1]{Kaito Wada\thanks{\href{mailto:wkai1013keio840@keio.jp}{wkai1013keio840@keio.jp}}}
\author[1,2,3]{Jumpei Kato}
\author[1]{Hiroyuki Harada}
\author[3,4]{Naoki Yamamoto\thanks{\href{mailto:yamamoto@appi.keio.ac.jp}{yamamoto@appi.keio.ac.jp}}}

\affil[1]{Graduate School of Science and Technology, Keio University}
\affil[2]{Mitsubishi UFJ Financial Group, Inc. and MUFG Bank, Ltd.}
\affil[3]{Quantum Computing Center, Keio University}
\affil[4]{Department of Applied Physics and Physico-Informatics, Keio University}


\date{\today}

\maketitle

\begin{abstract}
Density matrix exponentiation (DME) is a general procedure that converts an unknown quantum state into the Hamiltonian evolution.
This enables state-dependent operations and can reveal nontrivial properties of the state, among other applications, without full tomography. 
However, it has been proven that for any physical process, the DME requires $\Theta(1/\varepsilon)$ state copies in error $\varepsilon$.
In this work, we go beyond the lower bound and propose a procedure called the \textit{virtual} DME that achieves  $\mathcal{O}(\log(1/\varepsilon))$ or $\mathcal{O}(1)$ state copies, by using non-physical processes. 
Using the virtual DME in place of its conventional counterpart realizes a general-purpose quantum algorithm for property estimation, that achieves \textit{exponential} circuit-depth reductions over existing protocols across tasks including quantum principal component analysis, quantum emulator, calculation of nonlinear functions such as entropy, and linear system solver with quantum precomputation.
In such quantum algorithms, the non-physical process for virtual DME can be effectively simulated via simple classical post-processing while retaining a near-unity measurement overhead.
We numerically verify this small constant overhead together with the exponential reduction of copy count in the quantum principal component analysis task. 
The number of state copies used in our algorithm essentially saturates the theoretical lower bound we proved.
\end{abstract}

\bigskip

\addtocontents{toc}{\protect\setcounter{tocdepth}{0}}

\begin{figure}[bt]
 \centering
 \includegraphics[scale=1.05]{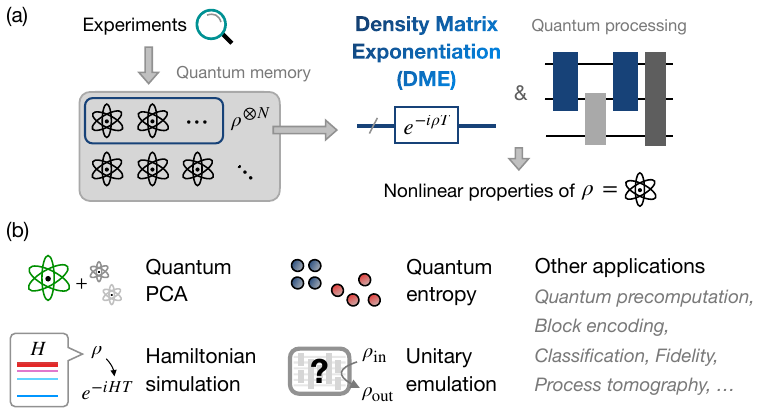}
 \caption{(a) Concept of density matrix exponentiation (DME). 
 The DME converts multiple copies of a state $\rho$ into the Hamiltonian evolution as $e^{-i\rho T}$, which integrates versatile state-dependent quantum operations into quantum data processing. 
 (b) Applications of DME. Our procedure for DME improves the accuracy dependence exponentially over all existing ones, thereby offering the same exponential speedup in various tasks using DME.
 }
 \label{fig:main_result}
\end{figure}

\section{Introduction}

The importance of investigating the properties of quantum states is becoming increasingly evident in various scenarios including condensed-matter physics and quantum information science. 
In particular, with the advent of quantum devices such as analog/digital quantum simulators~\cite{feynman2018simulating,doi:10.1126/science.273.5278.1073,RevModPhys.86.153,bloch2012quantum} and quantum sensors~\cite{RevModPhys.89.035002,giovannetti2011advances}, it has now become possible to reproducibly generate quantum states. 
These states have thus emerged as valuable objects, often referred to as ``quantum data"~\cite{huang2021power,PhysRevLett.126.190505,Biamonte_2017}.

Conventionally, we directly measure each quantum data to reveal the quantum properties,
but this process incurs inevitable information loss due to the quantum mechanical nature. 
On the other hand, recent progress in quantum technology naturally suggests to us to use quantum computers to coherently process quantum data without any measurement. 
This can avoid the information loss caused by measurement, leading to a much more efficient extraction of specific nontrivial properties and quantum data processing~\cite{Lloyd2014-nn,chen2022exponential,huang2022quantum}.
Because this approach is promising, it is important to establish a more general and systematic approach to coherent quantum data processing, 
particularly one that achieves minimal consumption of unclonable quantum data.

By contrast, for Hamiltonians---another central object in quantum physics, there has already been significant development in efficient quantum algorithms such as Hamiltonian
simulation~\cite{berry2014exponential,berry2015hamiltonian,berry2015simulating,low2017qsp}, eigenvalue estimation~\cite{kitaev1995quantum,Dutkiewicz2022heisenberglimited}, and a more general eigenvalue transformation~\cite{low2019qubitization,gilyen2019-qsvt}.
However, comparable algorithms for quantum states remain scarce.
An intuitive way to understand this gap is that Hamiltonian $H$ has a dynamical nature generating the time-evolution $e^{-iHt}$, whereas a state $\rho$ does not due to its static nature. 
The time evolution of $H$ underpins the Hamiltonian-oriented algorithms; hence, an efficient conversion from a state to Hamiltonian and realizing $e^{-i\rho t}$ unlocks those sophisticated algorithms for the purpose of revealing nontrivial properties of $\rho$.
Moreover, this conversion eventually produces versatile quantum operations $\mathcal{G}(\rho)$ that integrate state-dependent functionality into quantum data processing. 
Such $\mathcal{G}(\rho)$ may formally have the Fourier form 
\begin{equation}
\label{eq:Fourier_form}
    \mathcal{G}(\rho) = \int g(t) e^{-i\rho t}dt
\end{equation}
for some function $g(t)$, implying the general necessity of the operation $e^{-i\rho t}$.

Fortunately, there exists a general procedure that converts $\rho$ to $e^{-i\rho T}$ by consuming multiple copies of $\rho$, which is called the {\it density matrix exponentiation (DME)}~\cite{Lloyd2014-nn}.
The DME allows the use of algorithms for Hamiltonian to analyze $\rho$, without very expensive quantum tomography; see Fig.~\ref{fig:main_result}. 
Specifically, the conventional DME uses $N=\mathcal{O}(T^2/\varepsilon)$ state copies and $\mathcal{O}(N\log d)$ elementary quantum operations to approximately implement $e^{-i\rho T}$ in the system dimension $d$ and the error $\varepsilon$~\cite{Lloyd2014-nn,Kimmel2017-hv}. 
This shows a \textit{super-exponential} reduction in the number of copies compared to the polynomial $d$-dependence in any tomographic approaches~\cite{Kimmel2017-hv}. 
Also, an experimental realization of DME has been demonstrated~\cite{Kjaergaard2022-ed}.
However, a critical drawback of the conventional DME is the poor scaling $\mathcal{O}(\varepsilon^{-1})$. 
Notably, it was proven that the scaling $\mathcal{O}(T^2/\varepsilon)$ cannot be overcome, or equivalently the lower bound is $\Omega(T^2/\varepsilon)$, for any use of quantum process alone~\cite{Kimmel2017-hv, Go2024-pb}. 
This is the main reason that has prevented the practical use of DME, despite its importance. 
In contrast, remarkably, recent quantum algorithms, including Hamiltonian simulation and quantum linear system solver, achieve $\mathcal{O}(\log(1/\varepsilon))$ circuit depth~\cite{berry2014exponential,low2017qsp, low2019qubitization, gilyen2019-qsvt, PhysRevLett.129.030503, Chakraborty2024implementingany,PRXQuantum.6.010359,berry2015hamiltonian,berry2015simulating,childs2017quantum,an2022quantum,lin2020optimal}.

In this work, we provide a procedure for DME that breaks the above-mentioned fundamental limit and achieves an \textit{exponential}/\textit{super-exponential} improvement in the number of state copies, by using non-physical quantum processes. 
Specifically, our procedure simulates $e^{-i\rho T}$ with use of $N=\mathcal{O}(T^2\log(T/\varepsilon))$ state copies and $\mathcal{O}(N\log d)$ elementary quantum operations, compared to the previous one with $N=\mathcal{O}(T^2/\varepsilon)$. 
Furthermore, while our procedure works for any unknown quantum states, if the target state $\rho$ is pure, the required number of copies becomes $N=\mathcal{O}(1)$.
That is, surprisingly, we can simulate $e^{-i\rho T}$ of any pure state $\rho$ with no error by using a constant number of state copies. 
The non-physical processes can be effectively simulated via simple (i.e., near-unity measurement overhead) classical post-processing in any quantum algorithm for estimating properties of an (extended) quantum system, although not implementable on a single quantum circuit.
For this reason, we call our DME the \textit{virtual DME}.
Note that the use of non-physical processes does not restrict the applicability of the virtual DME for many tasks. 
Also, we highlight that non-physical processes are widely employed in recent techniques such as 
quantum error mitigation~\cite{PhysRevLett.119.180509,PhysRevX.8.031027,Kim_2023,koczor2021-esd,Huggins2021-vd,PhysRevApplied.23.054021,RevModPhys.95.045005}, quantum circuit knitting~\cite{PhysRevLett.125.150504,Mitarai_2021,10.1109/TIT.2023.3310797,PRXQuantum.5.040308,PRXQuantum.6.010316,Harada2025densitymatrix,yamamoto2024virtualentanglementpurificationnoisy}, virtual cooling~\cite{PhysRevX.9.031013}, resource distillation~\cite{PhysRevLett.132.050203,PhysRevA.109.022403}, quantum broadcasting~\cite{yao2024optimal,PhysRevLett.132.110203,z2pr-zbwl}, nonlinear quantum transformation~\cite{holmes2023nonlinear}, and quantum error correction~\cite{cao2023quantummapscptphptp}.

Our results have a broad impact on quantum technologies, ranging from quantum data analysis to methods that enhance quantum computation itself via DME. 
To demonstrate this, we construct a general quantum algorithm for property estimation, including the signal processing via the Fourier form~\eqref{eq:Fourier_form}, with the use of the virtual DME. 
We prove this algorithm is optimal (for pure state cases) or nearly optimal (for general cases) in terms of the number of state copies.
Then, we apply the algorithm to the quantum principal component analysis~\cite{Lloyd2014-nn}, unknown unitary emulator~\cite{marvian2024universalquantumemulator}, quantum entropy estimation~\cite{acharya2020-entropy, wang2023-entropy,wang2024-new}, and quantum linear system solver with quantum precomputation~\cite{Huggins2024-qf}, beyond the simplest application---Hamiltonian simulation for $e^{-i\rho T}$.
All of these algorithms are exponentially improved compared to the previous one, while keeping the standard scaling of quantum measurement.

\section{Main results}
\subsection{Density matrix exponentiation with a few state copies}

\begin{figure}[bt]
 \centering
 \includegraphics[scale=0.92]{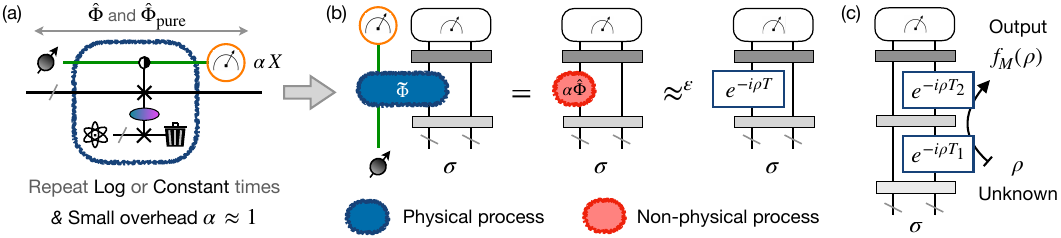}
 \caption{
 Procedure of the virtual DME. 
 (a) Quantum circuit for the virtual DME. 
 For a single run, randomly generated unitary gates are set into the colored ellipse. 
 The number of copies is logarithmic (Theorem~\ref{main_thm:general_result}) or constant (Theorem~\ref{thm_main:purestate}) when $T$ is constant.
 (b) Simulation of non-physical processes that approximate $e^{-i\rho T}$.
 The non-physical process (red) can be simulated by the random physical process (blue) followed by taking expectation on $\alpha X$, the rescaled Pauli X. 
 The explicit construction of $\tilde{\Phi}$ is provided in \green{the Supplementary Materials~\ref{sec:S22}}. 
 (c) Overview of our general quantum algorithm. This algorithm takes unknown state copies as input and outputs the expectation value Eq.~\eqref{eq_main:target_fM}.
 The initial state $\sigma$ may depend on the unknown state $\rho$.}
 \label{fig:algorithm_illustration}
\end{figure}

We now present our main results on the virtual DME, with a particular focus on the required number of state copies.
Let us introduce a random map $\alpha\hat{\Phi}$ to represent the procedure of the virtual DME.
This map can be simulated by a randomized quantum process and classical post-processing in any quantum algorithm for property estimation. 
The randomized quantum process has the common structure with indirect measurement, shown in Fig.~\ref{fig:algorithm_illustration}(a).
The classical post-processing is done by simply multiplying a scalar, represented by $\alpha$, by the measurement outcome.
The following theorem shows that at most $\mathcal{O}(T^2\log(T/\varepsilon))$ copies of $\rho$ are enough to construct an $\alpha\hat{\Phi}$ for approximating the target unitary $e^{-i\rho T}$ within the error $\varepsilon$ in expectation.

\begin{thm}[Virtual DME for general states; informal]
\label{main_thm:general_result}
    Let $\rho$ be an unknown but accessible $d$-dimensional quantum state, and let 
    $\mathcal{U}[e^{-i\rho T}](\sigma)=e^{-i\rho T} \, \sigma \, e^{i\rho T}$ be the unitary channel with a time $T$ ($|T|\geq 1$). 
    We can construct a random modified-quantum map $\hat{\Phi}$ using at most $N$ copies of $\rho$ such that
    \begin{equation}\label{eq_main:thm1_error}
        \left\| \alpha\mathbb{E}[\hat{\Phi}] - \mathcal{U}[e^{-i\rho T}]\right\|_{\diamond}\leq \varepsilon
    \end{equation}
    holds for the diamond norm, given that $N=\mathcal{O}(T^2\log(T/\varepsilon))$ and an error $\varepsilon \in (0, 1/2)$. 
    Here, $\alpha$ is a positive value such that $1<\alpha\leq e$ holds for the Napier's constant $e$. 
\end{thm}

\noindent
When the simulation time $|T|$ is smaller than 1, we can also construct a random modified-quantum map $\hat{\Phi}$ satisfying Eq.~\eqref{eq_main:thm1_error} given that $N=\mathcal{O}(\log(1/\varepsilon))$. 
The factor $\alpha$ is related to a classical post-processing cost, as will be mentioned later. 
As the necessity of classical post-processing implies, the map $\hat{\Phi}$ is not a physical process, and thus we call $\hat{\Phi}$ the random \textit{modified}-quantum map.
Technically, it always preserves Hermiticity; see \green{the Supplementary Materials \ref{sec:def of physical and non-physical} and \ref{supplesec:example_nonphysical}} for detailed characterization and a simple example of non-physical processes, respectively.
Also, the probability distribution of $\hat{\Phi}$ is sharply concentrated. 
Our result exponentially improves the error dependence $1/\varepsilon$ over all existing protocols for DME.
The full proof is provided in \green{the Supplementary Materials~\ref{supple_sec:sim_dme}}.

Theorem~\ref{main_thm:general_result} does not require any prior information about the target unknown state except for its dimension.
Yet, various applications use DME for pure states $\rho=\ketbra{\psi}$ (see Section~\ref{sec:application_main}). 
    When the target unknown state $\rho$ is known to be pure, surprisingly, the accuracy dependence is completely removed as follows.
    
    \begin{thm}[Virtual DME for pure states; informal]
    \label{thm_main:purestate}
        Let $\rho=\ketbra{\psi}$ be an unknown but accessible $d$-dimensional pure state, and let $T$ be a simulation time ($|T|\in (0,2\pi]$).  
    We can construct a random modified-quantum map $\hat{\Phi}_{\rm pure}$ using at most $N\geq 2\times {\max}\{2,2T^2\}$ copies of $\rho$ 
    such that
    \begin{equation}
        \alpha_{\rm pure}\mathbb{E}[\hat{\Phi}_{\rm pure}] = \mathcal{U}[e^{-i\ketbra{\psi}T}].
    \end{equation}
    Here, $\alpha_{\rm pure}$ is a positive value such that $1< \alpha_{\rm pure}\leq e$. 
\end{thm}

\noindent
We remark that if $\rho$ is pure, then $e^{-i\rho(T+2k\pi)}=e^{-i\rho T}$ holds for any integer $k$, and thus we consider $|T|\in (0,2\pi]$ without loss of generality. 
This result shows a super-exponential improvement in the error dependence $1/\varepsilon$ over all existing protocols for DME.
The proof of this theorem is given in \green{the Supplementary Materials~\ref{supplesec_variants}}, where it is essential to use the property $\rho^2=\rho$ for any pure state $\rho=\ketbra{\psi}$.

As for the controlled version of $e^{-i\rho T}$, we can simulate it by simply setting $\ketbra{1}\otimes \rho$ instead of $\rho$ in our procedure.
When the target state $\rho$ is pure, we can also use Theorem~\ref{thm_main:purestate} for simulating the corresponding controlled operation.

\subsection{Applicability and quantum resources in virtual DME}

Our virtual DME itself is a non-physical process, but it can be simulated by standard quantum computing with the help of classical post-processing. 
The procedure is very simple: we replace every operation $e^{-i\rho T}$ with the random circuit of Fig.~\ref{fig:algorithm_illustration}(a) in any quantum algorithm for property estimation, and finally perform classical post-processing.
The overall procedure effectively acts as a map $A\mapsto {\rm tr}_{\rm anc}[(\alpha X)_{\rm anc}A]$ for the single-ancilla system (denoted by anc), which is neither completely positive nor trace-preserving, i.e., a non-physical process; see Fig.~\ref{fig:algorithm_illustration}(b) or \green{the Supplementary Materials~\ref{supple_sec:non_physical_ope}}.
However, the entire classical post-processing cost increases due to the one for the virtual DME.
Specifically, the number of measurements (or the number of the entire circuit runs) increases by a multiplicative factor of $\alpha^M$ when a target quantum algorithm has $M$ uses of $e^{-i\rho T}$.
Here, the \textit{measurement overhead} $\alpha$ per unit is a small positive value determined by our construction of virtual DME.
Importantly, our construction admits the reduction of $\alpha$ to a value at most $e^{1/M}$ by increasing $N$ to $MN$ in Theorems~\ref{main_thm:general_result} and~\ref{thm_main:purestate}, thereby achieving a constant measurement overhead $e=2.718...$ (or a smaller value).
This is one of the most important features of our virtual DME; see the full version of our main theorems in \green{the Supplementary Materials~\ref{sec:Virtual_DME}}.

We here make some remarks on the random quantum processes in $\hat{\Phi}$ and $\hat{\Phi}_{\rm pure}$.
These random processes can be efficiently generated by independent sampling from the classical discrete probability distribution whose support size is $\mathcal{O}(\log(T/\varepsilon))$ for the general case or 2 for the pure state case. 
Their circuit implementation requires at most $N$ controlled (partial) SWAP gates, each of which can be implemented using only $\mathcal{O}(\log d)$ elementary gates~\cite{Lloyd2014-nn, Kimmel2017-hv}, thus $\mathcal{O}(N\log d) $ in total.
The $\log d$-dependence is the same as the conventional DME~\cite{Lloyd2014-nn} whose circuit has a similar structure as in Fig.~\ref{fig:algorithm_illustration}(a) except for the single ancilla qubit.
As a result, the virtual DME uses only a logarithmic number of elementary quantum operations in terms of both $d$ and $1/\varepsilon$.

The non-physical nature clearly separates the virtual DME from the conventional one in terms of the applicability for quantum state preparation.
The conventional DME is comprised of physical processes and thus can produce a quantum state evolved under $e^{-i\rho T}$ on a quantum computer, while the virtual DME cannot.
Nevertheless, in most applications, one needs only properties of an (extended) quantum system such as the expectation value of some observable with respect to final quantum states, rather than directly preparing the states. 
We will provide a general quantum algorithm, combined with the virtual DME, that can be applied to any such observable estimations in the next subsection.

\subsection{Optimal DME-based quantum algorithm for property estimation}
\label{sec:optimalalgorithm}

Theorems~\ref{main_thm:general_result} and~\ref{thm_main:purestate} allow us to replace a unitary gate $e^{-i\rho T}$ or its controlled version with a random modified-quantum map, in any DME-based algorithm for property estimation.
To describe this result systematically, we construct a general quantum algorithm to estimate properties of a quantum system evolved under a quantum process with multiple $e^{-i\rho T}$.
The following problem setup contains a wide range of applications including all instances presented in Section~\ref{sec:application_main}.
Let $f_M$ be a function from a quantum state to the expectation value of an observable $O$ with respect to a quantum state $\sigma$ evolved under $U_M$:
\begin{equation}\label{eq_main:target_fM}
    f_M(\rho):={\rm tr}[OU_M\sigma U_M^\dagger].
\end{equation}
Here, $U_M$ is any unitary circuit that alternatively performs $e^{-i\rho T}$ (to a partial subsystem of $\sigma$) and a known unitary gate $M$ times in total; see Fig.~\ref{fig:algorithm_illustration}(c) for an example of $f_M$ with $M=2$.
Note that $M$ operations in the form of $e^{-i\rho T}$ can have distinct $\rho$ and $T$, while we assume the same $\rho$ in the following for simplicity.
Also, $\sigma$ may be a function of $\rho$ such as $\sigma=\ketbra{0}\otimes \rho$ for an initial state $\ket{0}$.
Without loss of generality, the time of each $e^{-i\rho T}$ is assumed to be $|T|\leq 1$.
We aim to calculate the expectation value $f_M(\rho)$ of an unknown target state $\rho$ by using its copies minimally.

Our quantum algorithm below solves this problem.
We first replace all the $e^{-i\rho T}$ in $U_M$ with the randomly sampled map $\hat{\Phi}$ (or $\hat{\Phi}_{\rm pure}$) with certain $N$. 
Then, we measure the target observable $O$ with the final state.
The final measurement outcome multiplied by $\alpha^M$ and $\pm 1$ from the mid-circuit $X$ measurements serves as an estimate of $f_M(\rho)$. 
The estimator has the bias $\mathcal{O}(\varepsilon)$ by setting the right hand side of Eq.~\eqref{eq_main:thm1_error} to be $\mathcal{O}(\varepsilon/M)$ as the error.
Then, repeating the above procedure $\mathcal{O}(\alpha^{2M}/\varepsilon^2)$ times and taking the average value to reduce the shot noise, we can estimate the target expectation value within $\mathcal{O}(\varepsilon)$ error with a high probability. 
We remark that any quantum circuit appearing in this algorithm can be executed on a standard quantum computer because $\hat{\Phi}$ is a valid physical process except for the classical post-processing.

We summarize here the performance of our general quantum algorithm.
Importantly, the measurement overhead becomes $\alpha^{2M}\leq e$ by using $2MN$ copies of $\rho$ in each $\hat{\Phi}$, as mentioned before. 
Therefore, our algorithm has the standard measurement cost $\mathcal{O}(1/\varepsilon^2)$.
Under this crucial restriction, we prove that each circuit in this algorithm uses at most 
\begin{equation}
\label{eq:num_copies_qalg}
    \mathcal{O}\left(M^2{\log(M/\varepsilon)}\right)
\end{equation}
copies of the unknown quantum state $\rho$, and its gate complexity, except for the interleaving known operations, is given by
\begin{equation}\label{eq:num_gates_qalg}
    \mathcal{O}\left(M^2{\log(M/\varepsilon)}\log (d)\right),
\end{equation}
where we recall that $d$ is the dimension of $\rho$. 
Thus, the circuit depth (except for interleaving operations) of each circuit is logarithmic with respect to both $1/\varepsilon$ and $d$.
If the target density matrices to be exponentiated are all pure, then we can use $\hat{\Phi}_{\rm pure}$, leading to the number of copies $\mathcal{O}(M^2)$ and the gate complexity $\mathcal{O}(M^2\log d)$ in a single circuit.

Finally, in addition to the logarithmic scaling regarding $1/\varepsilon$ and $d$, the proposed quantum algorithm is (nearly) optimal regarding the number $M$, as follows.

\begin{thm}\label{main_thm:lower_bound_M_alg}
    Suppose that $\mathcal{A}$ is any quantum algorithm that, given $k$ copies of an unknown state $\rho$ and a description of $f_M$, outputs an estimate of the expectation value $f_M(\rho)$ within a constant additive error with a high probability. Then, any such algorithm $\mathcal{A}$ requires $k=\Omega(M^2)$ copies of $\rho$.
\end{thm}

\noindent 
The proof is provided in \green{the Supplementary Materials~\ref{supple_sec:alg_lower_bound}}. 
This theorem shows that for a general unknown state $\rho$, our algorithm with state copies at most Eq.~\eqref{eq:num_copies_qalg} is optimal up to a logarithmic factor. 
That is, unlike the case of simulating sparse Hamiltonians~\cite{berry2014exponential,berry2015simulating,berry2015hamiltonian,low2017qsp}, the linear resource scaling on $M$ is not possible.
In the case of pure states, the following remark should be made.

\begin{rem}[Optimality for pure state input]
    Theorem~\ref{main_thm:lower_bound_M_alg} also holds when we restrict the input state of $f_M$ to pure states.
    This means that the required number of copies in our algorithm using $\hat{\Phi}_{\rm pure}$ matches the lower bound $\Omega(M^2)$; that is, our method for pure state input is exactly optimal.
\end{rem}

\subsection{Proof overview of Theorem~\ref{main_thm:general_result}}

The key idea to derive Theorem~\ref{main_thm:general_result} is as follows.
Our construction of $\hat{\Phi}$ is based on the Taylor expansion of $e^{-i\rho T}$.
This expansion enjoys a rapid convergence with respect to the number of terms, leading to the desired $\mathcal{O}(\log(1/\varepsilon))$-scaling. 
However, there is a large gap between each non-unitary term in the expansion and allowable operations on a quantum computer. 
Moreover, even if one can rely on classical post-processing, it is still challenging to make its measurement overhead ($\alpha$ in our case) small. 
Our approach to these challenges is the use of controlled random partial SWAP gates and Pauli X measurement followed by the classical post-processing as shown in Fig.~\ref{fig:algorithm_illustration}(a). 
These operations provide a basis set for a very efficient and exact decomposition of the truncated Taylor expansion, and we eventually obtain $\hat{\Phi}$ with a small measurement overhead from this decomposition. 
More details are provided in \green{the Supplementary Materials~\ref{supple_sec:sim_dme}}.

\section{Applications}\label{sec:application_main}

Due to the ubiquity of DME itself, our algorithm leads to a significant reduction of complexity from $\mathcal{O}({\rm poly}(1/\varepsilon))$ to $\mathcal{O}({\rm polylog}(1/\varepsilon))$ in various important applications, while preserving the scaling of other problem-specific parameters.
Here we show some examples; detailed descriptions
are provided in \green{the Supplementary Materials~\ref{supple_sec:applications} and~\ref{sec:quantum_PCA}}.

\begin{figure}[bt]
  \centering
  \begin{minipage}[b]{0.33\linewidth}
    \centering
    (a) Hamiltonian simulation
    \includegraphics[scale=0.69]{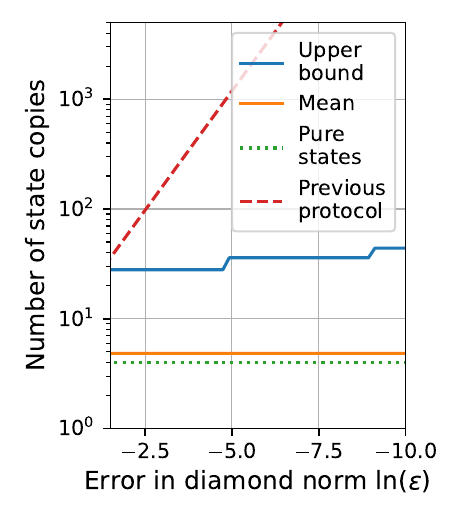}
  \end{minipage}
  \begin{minipage}[b]{0.66\linewidth}
    \centering
    (b) Quantum PCA
    \includegraphics[scale=0.71]{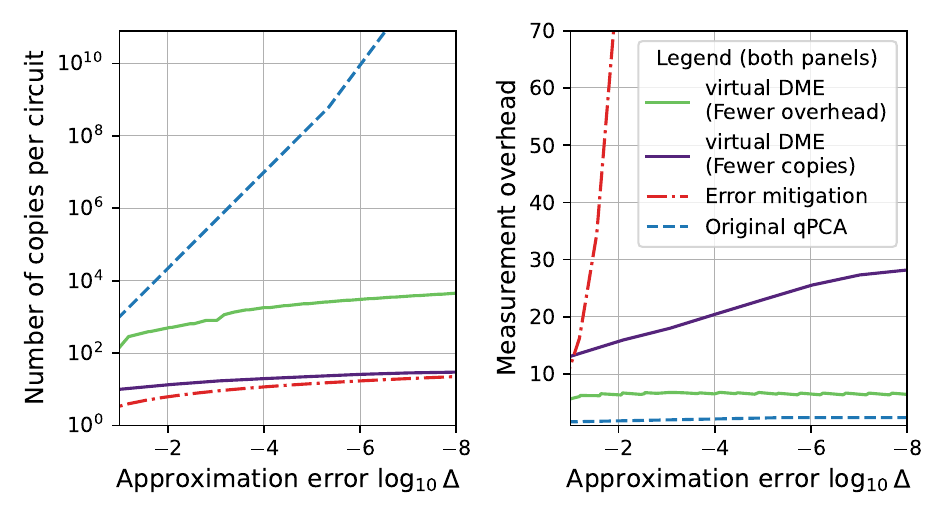}
  \end{minipage}
  \caption{Numerical simulation results. 
  (a) Hamiltonian simulation $e^{-i\rho T}$ with $T=1$, showing the necessary number of state copies against the error. 
  For general $\rho$, the blue and orange lines are the worst and average case in the virtual DME, respectively, indicating exponential reduction of the number of copies needed for the original DME method (red dashed line)~\cite{Lloyd2014-nn,Go2024-pb}. 
  The green dotted line shows the case of pure states $\rho$. 
  (b) Quantum PCA for $\rho=0.7\ketbra{\psi}+0.3\rho_{\rm err}$, where $\rho_{\rm err}$ is an arbitrary unknown pure state orthogonal to $\ketbra{\psi}$.
  We aim to estimate the expectation value of any unit-norm observable on the state $\ket{\psi}$ with the approximation error $\Delta$.
  Two methods derived from the virtual-DME-based algorithm are studied, represented by the green or purple lines.
  The red dash-dot and blue dotted lines represent the result from the quantum error mitigation (more precisely, virtual distillation)~\cite{koczor2021-esd, Huggins2021-vd} and the original qPCA~\cite{Lloyd2014-nn}, respectively.}
  \label{fig:main_numerical_res}
\end{figure}

\subsection{Hamiltonian simulation based on unknown $\rho$}
\label{sec:main-sample-based-HS}

There is a situation where we can access an unknown Hamiltonian $H$ as a black box and aim to do the simulation of $e^{-iHT}$. 
This problem can be formulated as a DME, if $H$ can be encoded into a density matrix in the form $\rho\propto (H+c\bm{1})$ for some constant $c$~\cite{Kimmel2017-hv, Go2024-pb}.
For the problem of estimating the properties of quantum states evolved by $e^{-i\rho T}$, the virtual DME allows the required number of copies to be exponentially smaller than the direct application of the previous method~\cite{Lloyd2014-nn}.
More precisely, we can estimate ${\rm tr}[O e^{-i\rho T} \sigma e^{i\rho T}]$ for any observable $O$ and state $\sigma$ by consuming at most $\mathcal{O}(T^2\log(T/\varepsilon))$ copies of $\rho$ per circuit, while the previous method consumes $\Theta(T^2/\varepsilon)$ copies per circuit. 
To clarify this advantage including the hidden constant in the scaling, we show a numerical simulation result on the number of state copies for this task in Fig.~\ref{fig:main_numerical_res}(a).
Note that the number of state copies for the random modified-quantum map in Theorem~\ref{main_thm:general_result} is randomly determined, and its upper bound is given by Eq.~\eqref{eq:hs_cal_copycount} in \green{the Supplementary Materials~\ref{supple_sec:sim_dme}}.
In Fig.~\ref{fig:main_numerical_res}(a), our virtual DME shows a clear exponential (for the upper bound of general cases) or super-exponential (for pure states) reduction.
Notably, the mean (average) number of the general state copies is almost 5, thereby offering a significant simplification for real device execution.

\subsection{Quantum principal component analysis (qPCA)}

Typically, output states of quantum processes are subjected to unknown noise.
However, we can expect that such a noisy state has a large overlap with a target pure state, when we operate on a near-ideal device~\cite{huang2022quantum,koczor2021dominant}. 
Let $\rho=(1-\lambda)\ketbra{\psi}+\lambda \rho_{\rm err}$ be a state of interest, where $\ket{\psi}$ is a target pure state (i.e., the first principal component) and $\rho_{\rm err}$ is any state orthogonal to $\ket{\psi}$. 
In qPCA tasks~\cite{Lloyd2014-nn,huang2022quantum}, we aim for estimating properties of $\ket\psi$, particularly the expectation value $\bra{\psi}O\ket{\psi}$ for an arbitrary observable $O$. 
Notably, as rigorously established by Ref.~\cite{huang2022quantum}, the original quantum algorithm~\cite{Lloyd2014-nn} for the qPCA task shows   
an exponential advantage compared to any classical algorithm with single-copy measurement outcomes for a fixed $\varepsilon$.
However, for desired errors beyond the fixed value, this algorithm requires at least $\mathcal{O}({\rm poly}(1/\varepsilon))$ state copies per circuit. 
On the other hand, some error mitigation techniques (sometimes called the virtual distillation)~\cite{koczor2021-esd, Huggins2021-vd} construct an estimate using a poly-logarithmic number of copies of $\rho$, yet at the cost of an exponential number of measurement overhead.
In contrast, our virtual DME solves the qPCA task using $\mathcal{O}({\rm polylog}(1/\varepsilon))$ state copies with a small constant measurement overhead and $\mathcal{O}({\rm polylog}(1/\varepsilon,d))$ gates per circuit. 
The key idea is to construct a filter $\mathcal{G}(\rho)\approx\ketbra{\psi}$, which is indeed of the Fourier form \eqref{eq:Fourier_form}, and to apply it to the noisy state $\rho$ for filtering out $\rho_{\rm err}$. 
Note that our algorithm does not require any information regarding noise except for the guarantee on the large overlap. 
Moreover, 
the total measurement cost in our procedure remains $\mathcal{O}(1/\varepsilon^2)$, which is the standard scaling of quantum measurement; thus, the super-exponential reduction of the measurement overhead, compared to the error mitigation~\cite{koczor2021-esd, Huggins2021-vd}, is achieved.

We here demonstrate the clear advantage of the proposed method with the virtual DME, over the original qPCA protocol~\cite{Lloyd2014-nn} and the error mitigation technique~\cite{koczor2021-esd, Huggins2021-vd}.
The detailed setup and in-depth numerical simulations are provided in \green{the Supplementary Materials~\ref{sec:quantum_PCA}}. 
Figure~\ref{fig:main_numerical_res}(b) shows the number of state copies per circuit (left) and the value of measurement overhead (right), to achieve a given approximation error.
It is clear from the figures that the proposed methods (green and purple) yield an estimator using a modest number of copies and measurements, while the existing methods require an exponentially large resource either in the number of copies or measurements.

\subsection{Emulation of an unknown unitary transformation from input-output data}

Universal quantum emulator (UQE)~\cite{marvian2024universalquantumemulator} is one of the powerful quantum machine learning algorithms for analyzing quantum data generated by quantum process~\cite{Biamonte_2017}. 
This algorithm emulates the action of an unknown unitary transformation $U$ given multiple pairs of input-output quantum states for
the unitary, $\{\ket{\phi_i^{\rm in}}, \ket{\phi_i^{\rm out}}\}$, 
without
reconstructing a classical description of $U$. Specifically, UQE coherently constructs a quantum channel which emulates the action of $U$ on a new unseen input state, using the DME-based operations $\{e^{-i \pi \ketbra{\phi_i^{\rm in}}},~e^{-i \pi \ketbra{\phi_i^{\rm out}}} \}$---the reflection for pure states $\ket{\phi}$ which is simply the case of $g(t)=\delta(t-\pi)$ in the Fourier form \eqref{eq:Fourier_form}. 
Notable features of UQE are 
its wide range of applicability~\cite{RevModPhys.91.025001,PhysRevA.78.022304,Arapinis2021quantumphysical}, including secure quantum computation~\cite{10.5555/2011670.2011674}, and its scalability in terms of the system dimension $d$, 
as its run-time (i.e., execution time) grows as $\mathcal{O}(\log d)$, unlike the standard process tomography that needs $\Omega(d)$ 
run-time~\cite{Chuang01111997,PhysRevLett.86.4195,PhysRevLett.78.390,haah2023query}.

Despite the usefulness of UQE, the original algorithm requires sequential consumption of 
$\mathcal{O}(1/\varepsilon)$ state pairs for the error $\varepsilon$ due to the use of the conventional DME~\cite{Lloyd2014-nn}. 
We prove that if the task is restricted to property estimation of the emulated state, our general algorithm with the virtual DME improves the $\varepsilon$-dependence exponentially, while preserving the favorable scaling in other parameters.
As a result, the UQE can be effectively performed using $\mathcal{O}({\rm polylog}(1/\varepsilon))$ state copies and $\mathcal{O}({\rm polylog}(1/\varepsilon,d))$ elementary gates. 
Additional details are provided in \green{the Supplementary Materials~\ref{apsec:uqe}}.

\subsection{Estimation of nonlinear properties such as entropy}

In quantum science, there are several important nonlinear functions of states. 
For instance, quantum entropy such as von Neumann entropy and relative entropy, is of central importance for entanglement quantification~\cite{calabrese2009entanglement} and quantum hypothesis testing~\cite{hiai1991relativy, ogawa2002strong}. 
To compute the entropy $S(\rho)=-\mathrm{tr}[\rho \, \ln\rho]$, which is a nonlinear function of unknown $\rho$, a standard approach is to identify the classical description of $\rho$ through tomographic measurements and then directly compute the entropy  classically~\cite{acharya2020-entropy}, leading to $\mathcal{O}(\mathrm{poly}(d))$ measurement cost. 
DME offers a promising alternative that avoids the costly tomographic process by coherently consuming $\mathcal{O}(\mathrm{poly}(1/\varepsilon))$ state copies per circuit for accuracy $\varepsilon$~\cite{wang2023-entropy}.

We propose a new algorithm for estimating von Neumann entropy and relative entropy using the virtual DME,
which requires $\mathcal{O}(\mathrm{polylog}(1/\varepsilon))$ state copies and $\mathcal{O}(\mathrm{polylog}(1/\varepsilon, d))$ gates per circuit.
The main idea is to construct a Fourier-formed nonlinear operation $\mathcal{G}(\rho) \approx \ln(\rho)$ directly in the circuit, by combining the virtual DME with well-established Hamiltonian signal processing technique~\cite{low2017qsp, low2019qubitization, gilyen2019-qsvt, Dong2022-qetu}.
Virtual DME thus acts as the crucial interface between the unknown state copies and the quantum signal processing, requiring logarithmically fewer copies with respect to $1/\varepsilon$.
Furthermore, this procedure is fully systematic, and thus it will be readily applicable to the task of calculating more general nonlinear properties of multiple quantum states, such as the distance or fidelity.
See \green{the Supplementary Materials~\ref{apsec:q entropy}} for the detail.

\subsection{Linear system solver with quantum precomputation}

The DME can be useful to significantly reduce the ``wall-clock'' time in quantum computing, given reasonably precomputed resources~\cite{Huggins2024-qf}. 
Usually, we execute circuits including $e^{-i\rho T}$ by decomposing it into a set of elementary gates, if the classical description of $\rho$ is available.
This requires $\mathcal{O}({\rm poly}(d))$ gates in the worst case of $d$-dimensional $\rho$.
However, when the state copies are reasonably computed ahead of time, the gate counts for the main circuit decrease to $\mathcal{O}({\rm poly}(1/\varepsilon,\log d))$~\cite{Huggins2024-qf} by the conventional DME and further $\mathcal{O}({\rm polylog}(1/\varepsilon,d))$ by the virtual DME.
Therefore, our algorithm unlocks a new regime for time-sensitive quantum computing when having precomputed resources.
Solving a quantum linear system of equations $A\ket{x}=\ket{b}$~\cite{harrow2009quantum,gilyen2019-qsvt,childs2017quantum,an2022quantum,lin2020optimal} is a typical application of the precomputation strategy~\cite{Huggins2024-qf}. 
As discussed in this reference paper, it is natural to consider the situation where a classical description of $\ket{b}$ is available before specifying the matrix $A$, and thus the preparation of $\ket{b}$ can be done in advance in this case.
Moreover, the state-of-the-art algorithm~\cite{lin2020optimal} with minimal calls to $A$ uses multiple $e^{-i\pi\ketbra{b}}$. 
As a result, we prove that the virtual DME combined with this efficient quantum solver can estimate the properties of the solution $\ket{x}$ using quantum circuits with $\mathcal{O}({\rm polylog}(1/\varepsilon,d))$ gates.
We refer to \green{the Supplementary Materials~\ref{apsec:precomp}} for more details.

\section{Discussion}

We have shown that the DME, the conversion from a state to the corresponding Hamiltonian evolution with a unit time, can be effectively simulated with a few state copies.
Also, the required number of elementary quantum operations is the same as that of state copies up to only a logarithmic factor in the dimension of a target state.
The crucial component for this achievement is non-physical processes, which can be efficiently simulated on standard quantum computers using simple classical post-processing, without directly implementing them on a circuit (we thus call our procedure the virtual DME).
A simple replacement of the conventional DME with the virtual DME offers an exponential improvement over
existing methods in a wide range of property estimation tasks, and we confirm this by deriving the required amounts of quantum resources for Hamiltonian simulation, quantum PCA, quantum emulation, estimation of nonlinear functions such as entropy, and quantum linear system solver with quantum precomputation.

Our results recast DME as a versatile and state-dependent primitive, potentially enabling
much simpler device-level implementations across various applications beyond the examples we presented here, such as 
quantum machine learning~\cite{cong2016quantum,brandao2017quantum,PhysRevLett.113.130503, zhao2019bayesian, lloyd2020quantum,liu2022quantum}, dynamic programming over states~\cite{PhysRevLett.134.180602}, state discrimination~\cite{Lloyd2014-nn,lloyd2020quantum,gilyen2022pgm}, process tomography~\cite{Lloyd2014-nn}, and entanglement-spectrum analysis~\cite{PhysRevX.6.041033}.
Beyond simplification, treating DME as an efficiently implementable primitive expands the design space for quantum algorithms. 
When a dense yet preparable Hamiltonian is encoded as a quantum state, our approach simulates its time evolution with resources governed by the number of state copies, not by sparsity. 
By preparing such quantum states as consumable resources in advance, virtual DME provides building blocks for new forms of quantum advantage under realistic wall-clock constraints, advancing the paradigm of quantum precomputation~\cite{Huggins2024-qf}. 
Moreover, although we focus on a single observable in this paper, we expect that our algorithm can be extended to multiple observable cases by e.g., incorporating shadow tomography~\cite{aaronson2018shadow,Huang_2020,Grier2024sampleoptimal,PRXQuantum.6.010336,grier2024principal}.

Finally, the virtual DME provides a crucial piece for fully quantum approaches to experimental science.
Recent progress in quantum technologies leads us to envision that quantum data is generated by quantum experimental apparatus and then processed by quantum computers. 
Can we have a general and systematic approach to perform such coherent quantum data processing?
Our results suggest a promising solution to this problem.
That is, the virtual DME unlocks the grand unified framework of quantum algorithms~\cite{martyn2021grand} for quantum data processing, as it effectively provides a highly efficient procedure for encoding quantum data (i.e., block encoding of $\rho$) to quantum circuits.

\section*{Acknowledgements}
We thank the fruitful discussions with Kosuke Ito and Suguru Endo.
K.W. was supported by JSPS KAKENHI Grant Number JP24KJ1963. 
This work was supported by MEXT Quantum Leap Flagship Program Grants 
No. JPMXS0118067285 and No. JPMXS0120319794, and JST Grant No. JPMJPF2221.

\clearpage

\appendix

\begin{center}
	\Large
	\textbf{Supplementary Materials for\\
    ``State-to-Hamiltonian conversion with a few copies''}\\[0.5em]
    \large Kaito Wada, Jumpei Kato, Hiroyuki Harada, Naoki Yamamoto
\end{center}

\addtocontents{toc}{\protect\setcounter{tocdepth}{3}}

To help the reader navigate the Supplementary Materials, we provide the following outline. 
Section~\ref{sec:Preliminaries} reviews the theoretical foundations underlying our main results. Section~\ref{sec:Virtual_DME} presents the detailed proofs of the main theorems in the main text. Section~\ref{sec:nearlyopt_qalg} introduces a general algorithmic framework for property estimation tasks together with a proof of its optimality. Section~\ref{supple_sec:non_physical_ope} discusses the non-physical operations responsible for the exponential improvement in our results over the previous DME protocols.
Section~\ref{supple_sec:applications} and Section~\ref{sec:quantum_PCA} summarize theoretical and numerical analysis in five property estimation tasks, and Section~\ref{sec:proofs} provides further theoretical details, including the proofs of the main theorems in Sections~\ref{supple_sec:applications} and~\ref{sec:quantum_PCA}.
\tableofcontents

\beginsupplement

\section{Preliminaries}\label{sec:Preliminaries}

\subsection{Physical and non-physical quantum processes}
\label{sec:def of physical and non-physical}

First, we list the physical and non-physical quantum processes appearing in this work. 
Let us consider finite-dimensional Hilbert spaces $\mathcal{H}_1$ and $\mathcal{H}_2$ for describing quantum systems, where $\mathcal{H}_2$ can be taken as $\mathcal{H}_1=\mathcal{H}_2$.
We write $\mathsf{L}(\mathcal{H})$ as the set of all linear operators on $\mathcal{H}$.
A linear map between the spaces of linear operators is often called a \textit{superoperator}.
\\

\noindent
\textbf{Physical process}
\begin{itemize}
    \item \textbf{Completely positive and trace-preserving (CPTP) map}. A CPTP map $\mathcal{E}$ is a superoperator that satisfies the completely positive and trace-preserving properties.
    This is the most general (deterministic) quantum process.
    It can always be written as the Kraus form: $\mathcal{E}(A)=\sum_i K_i AK_i^\dagger$ for any $A\in \mathsf{L}(\mathcal{H}_1)$, where $K_i$ is a linear operator from $\mathcal{H}_1$ to $\mathcal{H}_2$ such that $\sum_i K_i^\dagger K_i=\bm{1}$ for the identity $\bm{1}$. 
    Any CPTP map can be implemented on a quantum computer, while this may require a large number of ancilla qubits and elementary gates.
    
\end{itemize}

\noindent
\textbf{Non-physical process}
\begin{itemize}
    \item \textbf{Hermitian-preserving (HP) map}. A HP map $\mathcal{E}_{\rm h}$ is a superoperator that satisfies $\mathcal{E}_{\rm h}(A)^\dagger = \mathcal{E}_{\rm h}(A^\dagger)$ for any $A\in \mathsf{L}(\mathcal{H}_1)$. 
    It always returns a Hermitian operator for an arbitrary Hermitian input.
    Any CP map is automatically an HP map, but the inverse direction does not hold in general.
    An important example is $\mathcal{E}_{\rm h}(A)={\rm tr}_{\rm anc}[X_{\rm anc}A]$, where $X_{\rm anc}$ is the Pauli X operator acting only on a 1-qubit subsystem on $\mathcal{H}_1$, and ${\rm tr}_{\rm anc}$ is the partial trace on the subsystem.
    
    \item \textbf{Asymmetric map} $U\bullet V^\dagger$. 
    This is a superoperator that maps $A\mapsto UAV^\dagger$ for any $A\in\mathsf{L}(\mathcal{H}_1)$, where $U$ and $V$ are distinct linear operators from $\mathcal{H}_1$ to $\mathcal{H}_2$.
    This map is neither CP nor HP.
\end{itemize}

\noindent
Although non-physical processes cannot be directly implemented on a quantum computer, one may effectively simulate them with the help of classical post-processing.
We will discuss this with a clear example in Section~\ref{supplesec:example_nonphysical}.

\subsection{SWAP operation}
\label{sec:swap-operation}

Here, we list several properties of the SWAP operation.
The SWAP operator $S_{12}$ is defined as
\begin{equation}
    S_{12}:=\sum_{i,j}\ket{e_i}\bra{e_j}_{1}\otimes \ket{e_j}\bra{e_i}_{2},
\end{equation}
where $\{\ket{e_j}\}$ is an orthonormal basis set.
When it is clear from the context, we omit the subscript of $S_{12}$ for simplicity. 
\\

\noindent
\textbf{SWAP trick}:
For arbitrary and same-dimensional operators $A$ and $B$, the products $AB$ and $BA$ are realized as
\begin{equation}\label{eq:swap_trick}
    {\rm tr}_{2}\left[S_{12}(A_1\otimes B_2)\right]=BA,~~~
    {\rm tr}_{2}\left[(A_1\otimes B_2)S_{12}\right]=AB, 
\end{equation}
where ${\rm tr}_{2}$ denotes the partial trace over the second register. 
This relation can be generalized as follows.
For any operators $A,B$ and unitary $e^{-i\delta S_{12}}$, we have
\begin{align}
    {\rm tr}_2\left(e^{-i\delta S_{12} }A_1\otimes B_2\right)&={\rm tr}_2\left(\cos(\delta) A_1\otimes B_2 -i\sin(\delta) S_{12} A_1\otimes B_2\right)\notag\\
    &={\rm tr}[B]\cos(\delta)A-i\sin(\delta)BA\notag\\
    &={\rm tr}[B]\left(\cos(\delta)\bm{1}-i\sin(\delta)\frac{B}{{\rm tr}[B]}\right)A,
\end{align}
where $\bm{1}$ denotes an identity operator. 
Similarly,
\begin{align}
    {\rm tr}_2\left(A_1\otimes B_2 e^{i\delta S_{12} }\right)&={\rm tr}[B]A\left(\cos(\delta)\bm{1}+i\sin(\delta)\frac{B}{{\rm tr}[B]}\right).
\end{align}
In particular, if $B$ is a density matrix $\rho$, we obtain the following useful relations that we often use in this work:
\begin{equation}\label{eq:generalized_swap_tric_1}
    {\rm tr}_2\left(e^{-i\delta S_{12}}A_1\otimes \rho_2\right)=\left(\cos(\delta)\bm{1}-i\sin(\delta)\rho\right)A,
\end{equation}
\begin{equation}\label{eq:generalized_swap_tric_2}
    {\rm tr}_2\left(A_1\otimes \rho_2e^{i\delta S_{12} }\right)=A\left(\cos(\delta)\bm{1}+i\sin(\delta)\rho\right).
\end{equation}

\subsection{LMR protocol for density matrix exponentiation}
\label{sec:LMR}

Here, we review the conventional procedure proposed by Lloyd, Mohseni, and Rebentrost \cite{Lloyd2014-nn}, which we call the LMR protocol, for approximately implementing the density matrix exponential $e^{-i\rho T}$ from multiple copies of an unknown state $\rho$; that is, with the notation 
\begin{equation}
     \mathcal{U}[V](\bullet):=V\bullet V^\dagger, 
\end{equation}
the target is $\mathcal{U}[e^{-i\rho T}]$. 
The original approach uses a CPTP map $\Lambda_{\delta T}$ defined as
\begin{equation}
    \Lambda_{\delta T} (\bullet)={\rm tr}_{2}\left[e^{-i\delta T S}(
    \bullet\otimes \rho)e^{i\delta T S}\right],
\end{equation}
where again $S$ is the SWAP operator.
We note that this operation uses a unitary channel (thus, a physical operation) between a target system and a single copy of $\rho$. 
Then, for any $\sigma$,
\begin{equation}
    \Lambda_{\delta T}(\sigma)=\sigma-i\delta T[\rho,\sigma]+\mathcal{O}((\delta T)^2)=e^{-i\rho \delta T}\sigma e^{i\rho \delta T}+\mathcal{O}((\delta T)^2)
\end{equation}
holds. 
Therefore, sequential $N$ applications of $\Lambda_{\delta T}$ implement the unitary channel $\mathcal{U}[e^{-i\rho T}]$ ($T=N\delta T$) up to $\mathcal{O}(T^2/N)$ error.
In other words, we can implement the unitary $e^{-i\rho T}$ up to $\varepsilon$ error, using $N=\mathcal{O}(T^2/\varepsilon)$ copies of an unknown state $\rho$.
The gate complexity of this implementation is $\mathcal{O}(N\log d)$ for the dimension $d$ of $\rho$, so the logarithmic dependence on $d$ is achieved. 
The number of copies matches the lower bound $\Omega(T^2/\varepsilon)$, which was proven in Refs.~\cite{Kimmel2017-hv,Go2024-pb} for the case where we can use any physical quantum process (i.e., any CPTP map) between a target system and copies of $\rho$. 
The above procedure is optimal in this sense.

\subsection{Illustrative example for a non-physical process and its overhead}
\label{supplesec:example_nonphysical}

In the algorithm proposed in this paper, the non-physical process and resulting measurement overhead play a crucial role.
Here, we give a simple example of a non-physical process for the case of simulating a single-qubit gate. 
Let us consider a single-qubit gate $e^{-i\theta X}$ with Pauli $X$ and $\theta=\pi/4$.
We can decompose the unitary channel of 
$e^{-i\pi X/4}=(I-iX)/\sqrt{2}$ as
\begin{equation}\label{supple 1.3 example}
    e^{-i\pi X/4}A e^{i\pi X/4}=2\left[\frac{1}{4}A+\frac{1}{4}(-iX)A+\frac{1}{4}A(iX)+\frac{1}{4}XAX\right]
    \equiv 2\sum_{i=1}^4 \frac{1}{4} U_i A V_i^\dagger,
\end{equation}
where $U_i$ and $V_i$ are the unitary operators in the middle expression.
This equality shows that we can take a probabilistic (non-physical) approach to simulate the (physical) unitary channel of $e^{-i\pi X/4}$; that is, we apply the map $U_i\bullet V_i^\dagger$ to $A$ with the probability $1/4$, take their mean, and finally multiply the constant factor $2$. 
Among the four maps $U_i \bullet V_i^{\dagger}$, only the identity map and the conjugation $A \mapsto XAX$ are CPTP, and thus they can be physically implemented. 
However, the other two maps $A \mapsto (-iX)A$ and $A \mapsto A(iX)$ are not even Hermitian preserving. 
Therefore, they cannot be directly implemented on a quantum circuit in general.

If we only need expectation values at the end, such a non-physical process can be effectively simulated with the help of classical post-processing.
To see this, we revisit the \textit{Hadamard test} that can be used to obtain ${\rm Re}\left({\rm tr}[OU \sigma V^\dagger]\right)$ for distinct unitaries $U,V$, a general observable $O$, and any input state $\sigma$. 
From the circuit construction of the Hadamard test, we find that the expectation value ${\rm tr}[Oe^{-i\pi X/4}\sigma e^{i\pi X/4}]$ can be expanded as 
\begin{align}
    {\rm tr}\left[Oe^{-i\pi X/4}\sigma e^{i\pi X/4}\right]
    &=2\sum_{i=1}^4 \frac{1}{4} {\rm tr}[OU_i \sigma V_i^\dagger]=2\sum_{i=1}^4 \frac{1}{4} {\rm Re}\left({\rm tr}[OU_i \sigma V_i^\dagger]\right)\notag\\
    &=2\sum_{i=1}^4 \frac{1}{4} {\rm tr}\left[(X_{\rm anc}\otimes O)\cdot \mathcal{U}\left[\begin{pmatrix}
        U_i&0\\
        0&V_i
    \end{pmatrix}\right]
    (\ketbra{+}_{\rm anc}\otimes \sigma)\right]\notag\\
    &={\rm tr}\left(O\sum_{i=1}^4 \frac{1}{4} {\rm tr}_{\rm anc}\left[(2X_{\rm anc}\otimes \bm{1})\cdot \mathcal{U}\left[\begin{pmatrix}
        U_i&0\\
        0&V_i
    \end{pmatrix}\right]
    (\ketbra{+}_{\rm anc}\otimes \sigma)\right]\right), 
\end{align}
where $X$ is the Pauli X matrix and $\ket{+}$ is the $+1$ eigenstate of $X$. 
Also, ${\rm tr}_{\rm anc}$ means the trace out of the single ancilla system. 
Therefore, we have
\begin{align}
\label{supple 1.3 example final expression}
    e^{-i\pi X/4}\bullet e^{i\pi X/4}&=\sum_{i=1}^4 \frac{1}{4} {\rm tr}_{\rm anc}\left[(2X_{\rm anc}\otimes \bm{1})\cdot \mathcal{U}\left[\begin{pmatrix}
        U_i&0\\
        0&V_i
    \end{pmatrix}\right]
    (\ketbra{+}_{\rm anc}\otimes \bullet)\right]\notag\\
    &=\sum_{s=\pm} (2s) {\rm tr}_{\rm anc}\left[(\ketbra{s}\otimes \bm{1})\cdot \sum_{i=1}^4 \frac{1}{4} \mathcal{U}\left[\begin{pmatrix}
        U_i&0\\
        0&V_i
    \end{pmatrix}\right]
    (\ketbra{+}_{\rm anc}\otimes \bullet)\right].
\end{align}
The final expression has the following clear interpretation: we can simulate the target unitary channel with a random quantum process and Pauli X measurement followed by classical post-processing that multiplies $2$ with the outcome $\pm 1$.
In addition, the multiplication of the factor $2$ amplifies the estimation variance of the Hadamard test; we call this factor the {\it measurement overhead}. 
That is, the factor $2$ becomes the measurement overhead in simulating the target unitary channel via the expansion with non-physical processes in Eq.~\eqref{supple 1.3 example}.

In the next subsection, we present a general framework, called linear combination of superoperators, that extends the above discussion and allows for a systematic analysis.

\subsection{Linear combination of superoperators (LCS)}
\label{sec:theory_of_LCS}

We review a theoretical framework called linear combination of superoperators (LCS) introduced by~\cite{kato2024exponentially,yu2024exponentially}, which uses superoperators such as $U\bullet V^\dagger$ for some unitaries $U, V$~\cite{sun2022perturbative, sun2025lcs}. 
The framework of LCS significantly simplifies the proofs of our main results.

Let $\mathcal{W}_{\nu}(\bullet)$ be a sequence of superoperators $\{{\Phi}_{\nu,k}(\bullet)\}_{k=1}^{\#_\nu}$ ($\#_\nu$ means the number of superoperators involved in $\mathcal{W}_{\nu}$): 
\begin{equation}
    \mathcal{W}_{\nu}:={\Phi}_{\nu,\#_{\nu}}\circ \cdots {\Phi}_{\nu,2}\circ {\Phi}_{\nu,1}.
\end{equation}
Each superoperator $\Phi_{\nu,k}(\bullet)$ is a CPTP
map, or it may have the form of $\sum_i q_i U_i\bullet V_i^\dagger$ with unitary operators $U_i, V_i$ and a probability distribution $\{q_i\}$, like the case of example \eqref{supple 1.3 example}. 
The theory of LCS systematically provides a procedure to simulate any convex combination of superoperators $\{\mathcal{W}_\nu\}$ using quantum circuits with only a single additional qubit. 
Specifically, for a given probability distribution $\{p_{\nu}\}$, 
if $\sum_\nu p_{\nu}\mathcal{W}_{\nu}$ is a Hermitian-preserving map, then we can systematically find CPTP 
maps $\{\mathcal{\widetilde{W}}_{\nu}\}$ satisfying 
\begin{equation}\label{eq:LCS_formula}
    \sum_\nu p_\nu \mathcal{W}_{\nu}(\bullet)=\sum_{\nu} p_{\nu}{\rm tr}_{\rm anc}\left[(X_{\rm anc}\otimes \bm{1})\mathcal{\widetilde{W}}_{\nu}\left(\ketbra{+}_{\rm anc}\otimes \bullet\right)\right].
\end{equation}
Here, $\mathcal{\widetilde{W}}_{\nu}$ is obtained from $\mathcal{W}_{\nu}$ as
\begin{equation}
\mathcal{\widetilde{W}}_{\nu}:=\widetilde{\Phi}_{\nu,\#_\nu}\circ \cdots \circ \widetilde{\Phi}_{\nu,2}\circ \widetilde{\Phi}_{\nu,1}.
\end{equation}
Each $\widetilde{\Phi}$ is an extended superoperator of $\Phi$ defined as follows; if $\Phi$ is a CPTP
map,
\begin{equation}\label{eq:typeA_translated}
    \widetilde{\Phi}:=\mathcal{I}_{\rm anc}\otimes \Phi
\end{equation}
for the identity channel $\mathcal{I}_{\rm anc}$, and if $\Phi$ is a convex combination of asymmetric forms $\sum_i q_i U_i\bullet V_i^\dagger$ ,
\begin{align}
\label{eq:typeB_translated}
    \widetilde{\Phi}(\bullet)
    &= \sum_{i}q_i \Big(\ketbra{0}_{\rm anc}\otimes U_i+\ketbra{1}_{\rm anc}\otimes V_i\Big)\bullet \left(\ketbra{0}_{\rm anc}\otimes U^\dagger_i+\ketbra{1}_{\rm anc}\otimes V^\dagger_i\right) \notag\\
    &= \sum_i q_i \, \mathcal{U}\Big[\ketbra{0}_{\rm anc}\otimes U_i+\ketbra{1}_{\rm anc}\otimes V_i\Big](\bullet).
\end{align}
These translation rules are summarized in Table~\ref{tab:translationlist}.

\begin{table}
        \centering
        \begin{tabular}{c|c}
                superoperator $\Phi$ & circuit diagram for $\widetilde{\Phi}$\\ 
                \hline
                \hline
                \begin{tabular}{c}
                     CPTP map $\mathcal{E}$
                \end{tabular}
                ~~~&
                \begin{tabular}{c}
                     \Qcircuit @C=1em @R=1.2em {
                     \\
                     & \qw & \qw & \qw\\
                     &{/} \qw & \gate{\mathcal{E}} & \qw
                     \\
                     \\
                     }
                \end{tabular}
                \\
                \hline
                \begin{tabular}{c}
                     Convex combination of asymmetric forms\\
                     $\sum_i q_i U_i \bullet V^\dagger_i $
                \end{tabular} &~~$\sum_i q_i$
                \begin{tabular}{c}
                     \Qcircuit @C=1em @R=1.2em {
                     \\
                      &\qw& \ctrlo{1} & \ctrl{1} &\qw\\
                       &{/}\qw& \gate{U_i}  & \gate{V_i} & \qw
                     \\
                     \\
                     }
                \end{tabular}
                \\
                \hline
        \end{tabular}
    \caption{Translation rule from superoperators $\Phi$ into quantum circuits for $\widetilde{\Phi}$. The top (bottom) line represents the ancilla (system) qubit(s).
    }
    \label{tab:translationlist}
\end{table}

The LCS formula Eq.~\eqref{eq:LCS_formula}, which is a generalization of Eq.~\eqref{supple 1.3 example final expression}, can be proved as follows.
We first introduce the following notation: for any operator $B$ acting on the ancilla 1-qubit and the other (multiple) qubits, we denote $B$ as in the $2\times 2$ matrix form
\begin{equation}
    \begin{pmatrix}
        B_{00}&B_{01}\\
        B_{10}&B_{11}
    \end{pmatrix},
\end{equation}
where $B_{ij}:={(\bra{i}_{\rm anc}\otimes \bm{1} )} \, B \, {(\ket{j}_{\rm anc}\otimes \bm{1})}$ for the computational basis $\{\ket{i}_{\rm anc}\}$ on the ancilla system.
Using this notation, for any input operator $\ketbra{+}_{\rm anc}\otimes A$, we write it as
\begin{equation}
    \ketbra{+}_{\rm anc}\otimes A=\frac{1}{2}\begin{pmatrix}
        *&A\\
        A&*
    \end{pmatrix}.
\end{equation}
Because we are not interested in the diagonal elements, we use $*$ for simplicity.
Also, we define the anti-linear map $\mathcal{J}:A\mapsto A^\dagger$ for any operator $A$. 
In the following, we often use the following relation: for any Hermitian-preserving linear map $\mathcal{E}$, 
\begin{equation}
    \mathcal{J}\circ \mathcal{E}\circ \mathcal{J}=\mathcal{E}.
\end{equation}

When $\Phi$ is a CPTP map, the action of $\widetilde{\Phi}$ (defined by Eq.~\eqref{eq:typeA_translated}) can be written as
\begin{equation}
\label{CPTP expression}
    \widetilde{\Phi}\left(\begin{pmatrix}
        A_{00}&A_{01}\\
        A_{10}&A_{11}
    \end{pmatrix}\right)=\begin{pmatrix}
        *&\Phi(A_{01})\\
        \Phi(A_{10})&*
    \end{pmatrix}=\begin{pmatrix}
        *&\Phi(A_{01})\\
        \mathcal{J}\circ\Phi\circ\mathcal{J}(A_{10})&*
    \end{pmatrix},
\end{equation}
where the second equality holds due to the Hermitian-preserving property of CPTP maps.
On the other hand, if $\Phi=\sum_i q_iU_i\bullet V_i^\dagger$, the action of $\widetilde{\Phi}$ (defined by Eq.~\eqref{eq:typeB_translated}) takes the same form as Eq.~\eqref{CPTP expression}:
\begin{equation}
    \widetilde{\Phi}\left(\begin{pmatrix}
        A_{00}&A_{01}\\
        A_{10}&A_{11}
    \end{pmatrix}\right)=\begin{pmatrix}
        *&\Phi(A_{01})\\
        \mathcal{J}\circ \Phi\circ \mathcal{J}(A_{10})&*
    \end{pmatrix}.
\end{equation}
Therefore, in both cases of $\Phi$, we have
\begin{align}
    &\sum_\nu p_{\nu}\mathcal{\widetilde{W}}_{\nu}\left(\ketbra{+}_{\rm anc}\otimes A\right)\notag\\
    &=\sum_\nu p_{\nu}\frac{1}{2}\begin{pmatrix}
        *&({\Phi}_{\nu,\#_\nu}\circ \cdots\circ {\Phi}_{\nu,2}\circ {\Phi}_{\nu,1})(A)\\
        (\mathcal{J}\circ {\Phi}_{\nu,\#_\nu}\circ \cdots\circ {\Phi}_{\nu,2}\circ {\Phi}_{\nu,1}\circ \mathcal{J})(A)&*
    \end{pmatrix}\notag\\
    &=\frac{1}{2}\begin{pmatrix}
        *&\sum_\nu p_{\nu}\mathcal{{W}}_{\nu}(A)\\
        \sum_\nu p_{\nu}\mathcal{{W}}_{\nu}(A)&*
    \end{pmatrix},
\end{align}
where we used $\mathcal{J}^2=\mathcal{I}$ in the first equality. 
In the final equality, we used the Hermitian-preserving assumption on $\sum_\nu p_{\nu}\mathcal{W}_{\nu}$, leading to
\begin{equation}
    \sum_\nu p_\nu \mathcal{J}\circ \mathcal{W}_{\nu}\circ \mathcal{J}=\mathcal{J}\circ\left(\sum_\nu p_\nu  \mathcal{W}_{\nu}\right)\circ \mathcal{J}=\sum_\nu p_\nu  \mathcal{W}_{\nu}.
\end{equation}
From the final expression, we can directly confirm that Eq.~\eqref{eq:LCS_formula} holds.
For the better understanding, Fig.~\ref{fig:lcs-schematic} provides a simple example for $\sum_{\nu} p_{\nu} \mathcal{W}_\nu = \mathcal{W}$ and $p_{\nu}=\delta_{\nu 1}$.

\begin{figure}[htb]
 \centering
\includegraphics[width=0.8\linewidth,page=1]{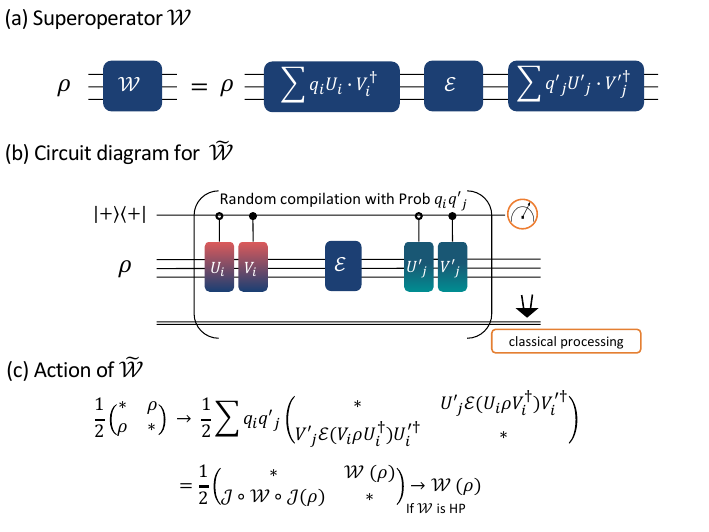}
 \caption{An example of LCS simulation. (a) A target superoperator $\mathcal{W}$, which is composed of 3 superoperators ($p_{\nu}=\delta_{\nu 1}$ and $\#_1=3$). 
 (b) A circuit realization of $\widetilde{\mathcal{W}}$ on the system $\rho$ and 1-qubit ancilla system $\ketbra{+}$. Controlled $U_i, V_i, U_j'$ and $V_j'$ are randomly compiled according to the probability $q_iq_j'$. 
 $X$-measurement is performed on the single ancilla system. 
 This measurement with outcomes $\pm 1$ forms the non-physical operation ${\rm tr}_{\rm anc}[X_{\rm anc}\bullet]$ via classical post-processing in expectation. 
 (c) Action of $\widetilde{\mathcal{W}}$. We can apply the target operation on the off-diagonal component in the dilated quantum system. The $X$-measurement and classical processing allow us to access the off-diagonal component virtually.}
 \label{fig:lcs-schematic}
\end{figure}

\subsection{Diamond norm}
For a given superoperator $\mathcal{X}$, which is a linear map between the spaces of (bounded) linear operators, 
the diamond norm~\cite{kitaev1997quantum} of $\mathcal{X}$ is defined as
\begin{equation}
    \|\mathcal{X}\|_{\diamond}:=\max_{R} \|\mathcal{X}\otimes \mathcal{I}_{R}\|_{1\to 1},
\end{equation}
where the maximization is taken over all ancilla system $R$, and $\mathcal{I}_{R}$ is the identity channel on the system $R$.
The norm $\|\bullet\|_{1\to 1}$ is the (induced) $1\to 1$ norm~\cite{10.5555/2011608.2011614}.
From the definition, we can easily check that $\|\mathcal{Y}\circ \mathcal{X}\|_{\diamond}\leq \|\mathcal{Y}\|_{\diamond}\|\mathcal{X}\|_{\diamond}$ holds. 
Also, we can prove $\|\mathcal{X}\otimes \mathcal{Y}\|_{\diamond}=\|\mathcal{X}\|_{\diamond}\|\mathcal{Y}\|_{\diamond}$ \cite{watrous2018theory}.

For the subsequent analysis, we here provide useful inequalities for the difference of two superoperators.
For any positive integer $K$, let $\mathcal{E}_1,...,\mathcal{E}_{K+1}$ be CPTP maps, and let $\{\mathcal{A}_{l}\}_{l=1}^K,\{\mathcal{B}_l\}_{l=1}^K$ be superoperators.
Defining 
\begin{equation}
    \mathcal{X}:=\mathcal{E}_{K+1}\circ \mathcal{A}_{K}\circ \mathcal{E}_{K}\circ\cdots \circ \mathcal{A}_1\circ \mathcal{E}_1,
\end{equation}
\begin{equation}
    \mathcal{Y}:=\mathcal{E}_{K+1}\circ \mathcal{B}_{K}\circ \mathcal{E}_{K}\circ\cdots \circ \mathcal{B}_1\circ \mathcal{E}_1,
\end{equation}
we can bound the difference between $\mathcal{X}$ and $\mathcal{Y}$ as
\begin{equation}\label{eq:xy_error_bound_general}
    \|\mathcal{X}-\mathcal{Y}\|_{\diamond}\leq \sum_{k=0}^{K-1} \left(\left[\prod_{l=1}^{K-k-1}\|\mathcal{A}_l\|_{\diamond}\right]\|\mathcal{A}_{K-k}-\mathcal{B}_{K-k}\|_{\diamond} \left[\prod_{l=K-k+1}^K\|\mathcal{B}_{l}\|_{\diamond}\right] \right)
\end{equation}
due to the triangle inequality and the fact that $\|\mathcal{E}\|_{\diamond}=1$ for any CPTP map $\mathcal{E}$.
An example is the case where each superoperator $\mathcal{A}_l$ is CPTP e.g., 
\begin{equation}
    \mathcal{A}_l=\mathcal{U}_l\otimes \mathcal{I}_l
\end{equation}
for some unitary channel $\mathcal{U}_l$ and identity channel $\mathcal{I}_l$.
If a superoperator $\mathcal{B}_l$ approximates the CPTP map $\mathcal{A}_l$ as $\|\mathcal{A}_l-\mathcal{B}_l\|_{\diamond}\leq \Delta$ for all $l$, 
then we have
\begin{equation}
\label{diamond_norm_of_xy}
    \|\mathcal{X}-\mathcal{Y}\|_{\diamond}\leq K\Delta (1+\Delta)^{K-1}\leq (K\Delta) e^{K\Delta}
\end{equation}
from Eq.~\eqref{eq:xy_error_bound_general}.
This is summarized as in Fig.~\ref{fig:diamond_norm_of_xy}.

\begin{figure}[htb]
 \centering
 \includegraphics[scale=0.45]{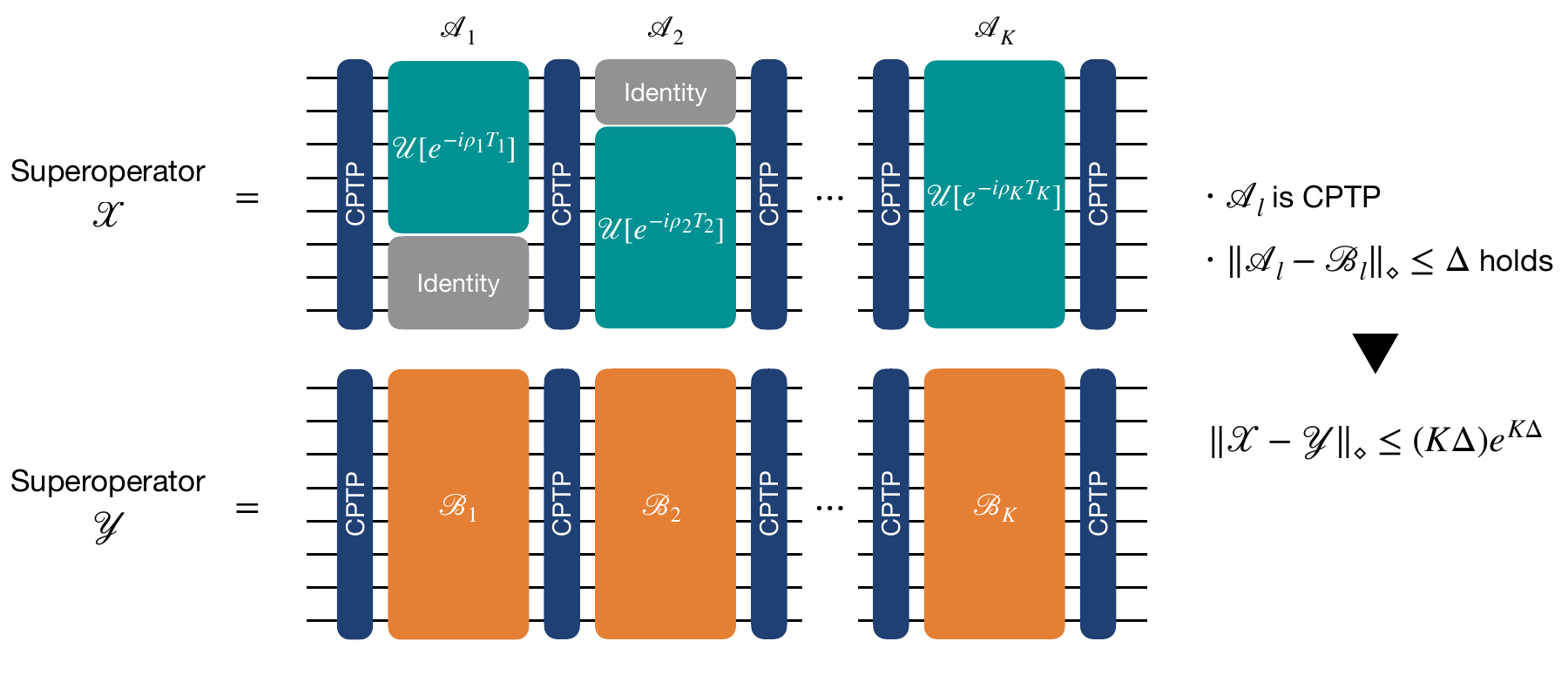}
 \caption{The distance between two superoperators $\mathcal{X}$ and $\mathcal{Y}$.
 If $\{\mathcal{A}_l\}$ are CPTP (not limited to the unitary channel for $e^{-i\rho_lT_l}$, yet this will be taken as $\mathcal{A}_l$ later) and the superoperators $\{\mathcal{B}_l\}$ satisfy  $\|\mathcal{A}_l-\mathcal{B}_l\|_{\diamond}\leq \Delta$ for an approximation error $\Delta>0$, then the error bound for $\mathcal{X},\mathcal{Y}$ holds.}
 \label{fig:diamond_norm_of_xy}
\end{figure}

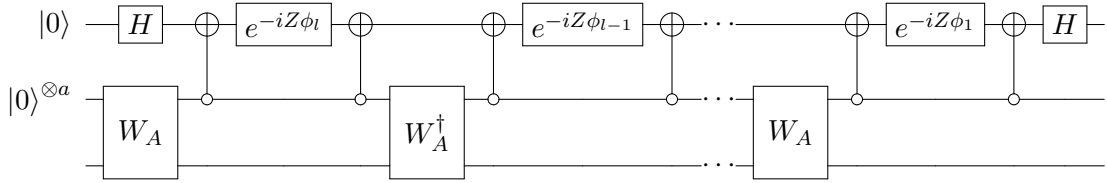
\begin{figure}[bt]
\centering
\begin{tabular}{c}
\\
\\
~~~~~\Qcircuit @C=0.6em @R=1.4em {
  \lstick{\ket{0}}&\gate{H}                           & \targ  & \gate{e^{-iZ\phi_l }} & \targ 
  &\qw & \targ & \gate{e^{-iZ\phi_{l-1} }} & \targ
  &\qw&\cdots&& \qw & \targ  & \gate{e^{-iZ\phi_1 }} & \targ
  &\gate{H} & \qw \\
  \lstick{\ket{0}^{\otimes a}}&\multigate{1}{W_A}  & \ctrlo{-1} & \qw & \ctrlo{-1}  
  &\multigate{1}{W_A^\dagger}& \ctrlo{-1}& \qw & \ctrlo{-1} 
  &\qw&\cdots& & \multigate{1}{W_A}& \ctrlo{-1}& \qw & \ctrlo{-1}
  & \qw&\qw \\
                              &\ghost{W_A}         & \qw & \qw & \qw                             
  &\ghost{W_A^\dagger}& \qw & \qw & \qw
  &\qw &\cdots&& \ghost{W_A} & \qw & \qw & \qw
  & \qw&\qw \\
  }
\\
\\
\end{tabular}
\caption{Quantum circuit for quantum singular value (eigenvalue) transformation for real polynomials $P$ of odd degree $l$.
$W_A$ denotes a $(\beta,a,0)$-block-encoding of a Hermitian operator $A$.
The NOT gate is controlled by $\ket{0}^{\otimes a}$ of the $a$ qubits, which is represented by the white circle $\circ$.
For a given real polynomial $P$ of degree $l$, the $l$ circuit parameters $\{\phi_i\}_{i=1}^{l}$ are calculated in $\mathcal{O}({\rm poly}(l))$ classical computation time~\cite{gilyen2019-qsvt}.
With such parameters, this circuit results in a $(1,a+1,0)$-block-encoding of $P(A/\beta)$.
}
\label{fig:qet_circuit}
\end{figure}

\subsection{Nonlinear transformation of matrices encoded in quantum systems}\label{sec:nonlinear-transformation}
Here, we review basic tools for efficiently transforming block-encoded matrices on a quantum computer. 
To encode target matrices into unitary operators of dilated quantum systems, we use the framework of block encoding (for simplicity, we here assume that the target matrix acts on a system whose dimension $d$ is a power of 2).
\begin{dfn}[Block encoding]
Suppose that $A$ is a $\log_2 d$-qubit operator, $\beta, \varepsilon > 0$
and $a \in \mathbb{N}$.
Then we say that a $(\log_2 d + a)$-qubit unitary $W_A$ is
 a $(\beta, a ,\varepsilon$)-block encoding of $A$, if
\begin{equation}
\| \beta (\bra{0}^{\otimes a} \otimes \bm{1} )W_A (\ket{0}^{\otimes a} \otimes \bm{1}) - A \| \leq \varepsilon,
\end{equation}
where $\| \bullet \|$ denotes the operator norm.
\end{dfn}
\noindent
When $\beta=1,\varepsilon=0$, the unitary matrix $W_A$ is represented as
\begin{equation}
    W_A=\begin{pmatrix}
        A&*\\
        *&*
    \end{pmatrix}.
\end{equation}
We can construct $W_A$ for various types of matrix $A$, e.g., sparse matrices and linear combinations of Pauli strings~\cite{gilyen2019-qsvt}.

Once we construct the block-encoding circuit for a target matrix, we can systematically transform the matrix based on a wide range of polynomial functions $P$.
Specifically, for a Hermitian matrix $A$ (see~\cite{gilyen2019-qsvt} for more general matrices), we can construct a quantum circuit to realize
\begin{equation}
    W_{P(A/\beta)}\simeq \begin{pmatrix}
        P(A/\beta)&*\\
        *&*
    \end{pmatrix}
\end{equation}
from the multiple uses of $W_A$. 
A formal statement of this fact is given in the following lemma together with Fig.~\ref{fig:qet_circuit}.

\begin{lem}[Quantum eigenvalue transformation~\cite{gilyen2019-qsvt}]
    Let $W_A$ be an $(\beta,a,0)$-block encoding of a Hermitian matrix $A$.
    Let $P(x)$ be a degree $l$ even or odd real polynomial satisfying $|P(x)|\leq 1$ for all $x\in[-1,1]$.
    Then, we can construct a quantum circuit for an $(1,a+1,0)$-block encoding of $P(A/\beta)$, with $l$ uses of $W_A~( W_A^\dagger)$ and $\mathcal{O}(al)$ one- or two-qubit elememtary gates.
\end{lem}





\clearpage

\section{Density matrix exponentiation with a few state copies}\label{sec:Virtual_DME}

\subsection{General theory}\label{supple_sec:sim_dme}

We now prove a generalized version of \green{Theorem~\ref{main_thm:general_result}} in the main text.

\begin{thm}\label{thm:main_supple}
    Let $\rho$ be an unknown but accessible $d$-dimensional quantum state, and let $\mathcal{U}[e^{-i\rho T}]$ be the unitary channel of $e^{-i\rho T}$ for $T\in \mathbb{R}\backslash\{0\}$.
    For any integer $r\geq \max\{1,|T|\}$, 
    we can construct a random modified-quantum map $\hat{\Phi}_r$ using at most $N$ copies of $\rho$ and a single-ancilla measurement such that
    \begin{equation}
    \label{eq:diamond_distance_alpha}
        \|C^{2r}\mathbb{E}[\hat{\Phi}_r]-\mathcal{U}[e^{-i\rho T}]\|_{\diamond}\leq \varepsilon
    \end{equation}
    holds for a given $\varepsilon\in(0,1/2)$, if $N$ satisfies 
    \begin{equation}
    \label{number of copies in Th S2}
        N\geq 2r+4r\left\lceil\frac{1}{2}\frac{\ln(3r/\varepsilon)}{W_0((1/e)\ln(3r/\varepsilon))}-\frac{1}{2}\right\rceil.
    \end{equation}
    Here, $W_0$ is the principal branch of the Lambert-W function (that is, $W_0(c)$ is defined as the unique solution of $xe^x = c$ for $x\geq -1$ and $c\geq -1/e$), and $C$ is a positive value that satisfies 
    \begin{equation}
    \label{bound of C in Th S2}
        1< C^{2r}\leq e^{2T^2/r}.
    \end{equation}
\end{thm}

We remark that $C$ is related to the total number of sampling of the random map $\hat{\Phi}_r$ and should be small, as described in detail later.
\green{Theorem~\ref{main_thm:general_result}} in the main text is a special case focusing on this aspect, that realizes a small $\alpha\equiv C^{2r}$
by setting $r=\lceil2T^2\rceil$ with $|T|\geq 1$ (and omitting the division by the $W_0$ term).
After the proof, we will summarize a concrete procedure to sample and simulate the random modified-quantum map $\hat{\Phi}_r$.

In advance of the proof, we sketch the key idea of Theorem \ref{thm:main_supple}.
From the SWAP trick described in Section~\ref{sec:swap-operation}, we observe the following fact. 
Using the asymmetric superoperator $\bm{1}\bullet e^{i\delta S}:A\mapsto A e^{i\delta S}$ followed by the trace-out operation, we can realize the multiplication of $(\cos(\delta)\bm{1}+i\sin(\delta)\rho)$ from the right-side without any error as in Eq.~\eqref{eq:generalized_swap_tric_2}. 
The multiplication of $(\cos(\delta)\bm{1}-i\sin(\delta)\rho)$ from the left-side is also realized by using $e^{-i\delta S}\bullet \bm{1}$.
Since the non-physical operation Eq.~\eqref{eq:generalized_swap_tric_2} is a product of an asymmetric map and CPTP maps, we can simulate it
by using a quantum circuit with a single-ancilla measurement followed by classical post-processing (or equivalently, multiplication of Pauli X followed by the trace-out operation of the single-ancilla register) as discussed in Section~\ref{sec:theory_of_LCS}.
We found that these multiplication operations from both sides serve as a basis set for a very efficient decomposition of the target unitary channel $\mathcal{U}[e^{-i\rho T}]$.

The crucial piece connecting $\mathcal{U}[e^{-i\rho T}]$ with $(\cos(\delta)\bm{1}\pm i\sin(\delta)\rho)$ is the following Taylor expansion with special grouping:
\begin{align}\label{eq:explain_expansion}
    e^{-i\rho T/r}&=\bm{1}-i({T}/{r})\rho + (-i)^2 \frac{(T/r)^2}{2!}\rho^2+(-i)^3 \frac{(T/r)^3}{3!}\rho^3+\cdots\notag\\
    &=\sqrt{1+(T/r)^2}\left(\cos(\theta_0)\bm{1}-i{\rm sgn}(T)\sin(\theta_0)\rho\right) \notag\\
    &~~~~~+ \frac{(T/r)^2}{2!}\sqrt{1+\left(\frac{T/r}{3}\right)^2}(-i\rho)^2\left(\cos(\theta_1)\bm{1}-i{\rm sgn}(T)\sin(\theta_1)\rho\right) 
    +\cdots.
\end{align}
Here, by appropriately determining $\theta_0,\theta_1,...$, we combined the $(2l,2l+1)$-th terms ($l=0,1,...$) in the first line to obtain the second and third lines.
The expansion indicates that $\mathcal{U}[e^{-i\rho T/r}]$ (and thus $\mathcal{U}[e^{-i\rho T}]$) can be expanded by a sum of products of superoperators in the form of 
\begin{equation}
\label{eq:explain_random_map}
    (\cos(\delta)-i\sin(\delta)\rho)\bullet (\cos(\delta')+i\sin(\delta')\rho).
\end{equation}
Since the truncation error of the Taylor expansion becomes exponentially small with respect to the number of expansion terms, the truncated series of $\mathcal{U}[e^{-i\rho T/r}]$ only contains $(\cos(\delta)\bm{1}\pm i\sin(\delta)\rho)^{l}$ with a logarithmically small number $l$ with respect to the truncation error.
This is the reason for the logarithmic dependence in Theorem~\ref{thm:main_supple}.
Then, our random map $\hat{\Phi}_r$ effectively realizes a product of Eq.~\eqref{eq:explain_random_map} according to its expansion coefficient.

Note that $C$ is given by the absolute sum of the expansion coefficients and, as mentioned above, it is connected with the total number of sampling and accordingly the total computational cost in our quantum algorithm (see Section~\ref{sec:nearlyopt_qalg}); the use of the expansion~\eqref{eq:explain_expansion}, instead of direct use of Taylor expansion, is crucial to keep the value $C$ small. 
Specifically, if we expand $\mathcal{U}[e^{-i\rho T}]$ based on the naive Taylor expansion, then $C^{2r}$ must be proportional to $e^{\mathcal{O}(T)}$ and thus suffers from exponentially large overhead. 
Note that a similar expansion to Eq.~\eqref{eq:explain_expansion} is also used in recent Hamiltonian simulation methods~\cite{PhysRevLett.129.030503, PRXQuantum.6.010359,Chakraborty2024implementingany}; the crucial difference from these previous results is that our expansion basis operation $(\cos(\delta)\bm{1}\pm i\sin(\delta)\rho)$ is not a unitary operation and thus Eq.~\eqref{eq:explain_expansion} differs from a linear combination of unitaries. 
Hence, each expansion operation is a non-physical process and cannot be directly implemented on a quantum computer. 
This is the reason why we need the theory of LCS in Section~\ref{sec:theory_of_LCS} and as a result the {\it modified} quantum map $\hat{\Phi}_r$ for approximating $\mathcal{U}[e^{-i\rho T}]$.

\begin{proof}[Proof of Theorem \ref{thm:main_supple}]
We focus on the Taylor expansion of $e^{-i\rho T/r}$ with specific grouping, up to the $(2L+1)$-th term:
\begin{align}\label{eq:def_SL_mainthm1}
    S_{L}(T/r)&:=\sum_{q=0}^{2L+1}\frac{(T/r)^q}{q!}(-i\rho)^q=\sum_{l=0}^{L}\frac{(T/r)^{2l}}{(2l)!}(-i\rho)^{2l}\left(\bm{1}-i\frac{T/r}{2l+1}\rho\right)\notag\\
    &=\sum_{l=0}^{L}\frac{(T/r)^{2l}}{(2l)!}(-i\rho)^{2l}\left(\bm{1}-i\frac{|T|/r}{2l+1}{\rm sgn}(T)\rho\right)\notag\\
    &=\sum_{l=0}^{L}C_l(T/r)(-i\rho)^{2l}\left(\cos(\theta_{l,|T|/r})\bm{1}-i{\rm sgn}(T)\sin(\theta_{l,|T|/r})\rho\right),
\end{align}
where $\theta_{l,|T|/r}$ and $C_l(x)$ are defined as
\begin{equation}\label{eq:angle_for_partialswap}
    {\theta_{l,|T|/r}}:=\arccos\left(\left\{1+\left(\frac{T/r}{2l+1}\right)^{2}\right\}^{-1/2}\right),
\end{equation}
\begin{equation}
    C_l(x):=\frac{x^{2l}}{(2l)!}\sqrt{1+\left(\frac{x}{2l+1}\right)^{2}}\geq 0.
\end{equation}
Also, we here define the probability distribution $p_l(T/r)$ as 
\begin{equation}
\label{def of C}
    p_l(T/r):=\frac{C_l(T/r)}{C(T/r)},~~~
    C(T/r):=\sum_{l=0}^{L} C_l(T/r) \leq e^{T^2/r^2},
\end{equation}
where $C(T/r)\leq e^{T^2/r^2}$ holds for any $L$ as proven in Proposition~13 in~\cite{PhysRevLett.129.030503}. 
We now choose the integer $L$ to be 
\begin{equation}
\label{def of L}
    L:=\left\lceil\frac{1}{2}\frac{\ln(3r/\varepsilon)}{W_0((1/e)\ln(3r/\varepsilon))}-\frac{1}{2}\right\rceil,~~\left(\mbox{or}~~
    L:=\left\lceil\frac{\ln \left({3r}/{\varepsilon}\right)}{\ln \ln \left({3r}/{\varepsilon}\right)}-\frac{1}{2}\right\rceil\right), 
\end{equation}
where the second $L$ in Eq.~\eqref{def of L} is a simplified version of the first choice. 
Under this choice of $L$, from the technical Lemma~\ref{lem:DME_taylor_truncation} shown later, the operator $S_L(T/r)$ is $\varepsilon/3r$-close to the target unitary $e^{-i\rho T/r}$ in the operator norm:
\begin{equation}
\label{eq:approx_error_SL}
    \left\|S_L(T/r)-e^{-i\rho T/r}\right\|\leq\frac{\varepsilon}{3r}.
\end{equation}

To simulate each term in $S_L(T/r)$, we define the following two superoperators acting on linear operators of $d$-dimensional Hilbert space: for any $d$-dimensional operator $A$, $\delta\in \mathbb{R}$, integer $k\geq 1$,
\begin{align}\label{eq:basic_op1_main_thm}
    \Upsilon_{\delta,k}(A)&:={\rm tr}_{\rm E}\left[\left(\prod_{j=k}^1 e^{-i\delta S_{{\rm targ,E}_j}}\right)A_{\rm targ}\otimes \rho_{\rm E_1}\otimes \rho_{\rm E_2}\otimes \cdots \rho_{{\rm E}_k}\right]\\
    &=\left(\cos(\delta)\bm{1}-i\sin(\delta)\rho\right)^{k}A,
\end{align}
and 
\begin{align}\label{eq:basic_op2_main_thm}
    \Upsilon^{\S}_{\delta,k}(A)&:={\rm tr}_{\rm E}\left[A_{\rm targ}\otimes \rho_{\rm E_1}\otimes \rho_{\rm E_2}\otimes \cdots \rho_{{\rm E}_k}\left(\prod_{j=1}^k e^{i\delta S_{{\rm targ,E}_j}}\right)\right]\\
    &=A\left(\cos(\delta)\bm{1}+i\sin(\delta)\rho\right)^{k}.
\end{align}
Here, $S_{{\rm targ, E}_j}$ denotes the SWAP operator between the target system (denoted by the subscript ${\rm targ}$) and the system ${\rm E}_j$.
The second equalities follow from the properties of the SWAP operator; see Eqs.~\eqref{eq:generalized_swap_tric_1} and \eqref{eq:generalized_swap_tric_2}.
The superoperators $\Upsilon_{\delta,k}$ and $\Upsilon^{\S}_{\delta,k}$ consist of three sequential operations: (i) couple $A$ with $k$ copies of $\rho$, (ii) apply an asymmetric superoperator
\begin{equation}
    \left(\prod_{j=k}^1 e^{-i\delta S_{{\rm targ},E_j}}\right)\bullet\bm{1}~~~\mbox{or}~~~\bm{1}\bullet\left(\prod_{j=1}^k e^{i\delta S_{{\rm targ},E_j}}\right),
\end{equation}
and (iii) trace out all systems except for the original system of $A$.
We then observe
\begin{align*}
    &\left(\Upsilon_{\pi/2,2l}\circ\Upsilon_{{\rm sgn}(T)\theta_{l,|T|/r},1}\circ \Upsilon^{\S}_{\pi/2,2l'}\circ\Upsilon^{\S}_{{\rm sgn}(T)\theta_{l',|T|/r},1}\right)(A)\notag\\
    &= \left(\Upsilon_{\pi/2,2l}\circ\Upsilon_{{\rm sgn}(T)\theta_{l,|T|/r},1}\right)(A\left(\cos(\theta_{l',|T|/r})\bm{1}+i{\rm sgn}(T)\sin(\theta_{l',|T|/r})\rho\right) (i\rho)^{2l'})\notag\\
    &= (-i\rho)^{2l} \left(\cos(\theta_{l,|T|/r})\bm{1}-i{\rm sgn}(T)\sin(\theta_{l,|T|/r})\rho\right)A\left(\cos(\theta_{l',|T|/r})\bm{1}+i{\rm sgn}(T)\sin(\theta_{l',|T|/r})\rho\right) (i\rho)^{2l'}.
\end{align*}
Using this map, we obtain (omitting the argument $T/r$ on $S_L(T/r)$ and $C(T/r)$ for simplicity)
\begin{align}
    S_{L}AS_L^\dagger&=\sum_{l,l'}C_lC_{l'}\left(\Upsilon_{\pi/2,2l}\circ\Upsilon_{{\rm sgn}(T)\theta_{l,|T|/r},1}\circ \Upsilon^{\S}_{\pi/2,2l'}\circ\Upsilon^{\S}_{{\rm sgn}(T)\theta_{l',|T|/r},1}\right)(A)\notag\\
    &\equiv C^2\sum_{l,l'}p_lp_{l'}\Gamma_{l,l'}(A)
\end{align}
for any $A$, where Eq.~\eqref{def of C} is used. 
By repeating this operation, we have
\begin{equation}
\label{eq:SAS^r in the middle}
    S_L^{r}A(S^\dagger_L)^r=C^{2r}\sum_{l_1,l'_1,...l_r,l_r'} p_{l_1}p_{l_1'}\cdots p_{l_r}p_{l_r'} \left(\Gamma_{l_r,l_r'}\circ \cdots \circ \Gamma_{l_1,l_1'}\right)(A).
\end{equation}

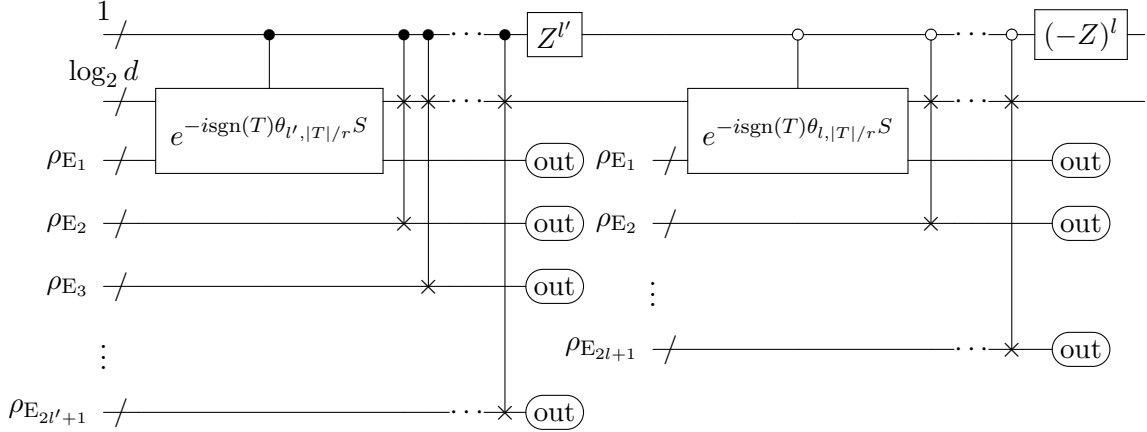
\begin{figure}[tb]
\centering
\begin{tabular}{c}
~~~~~\Qcircuit @C=0.6em @R=1em {
  \ustick{1}&{/} \qw&\qw &\ctrl{1}&\ctrl{3}&\ctrl{4}&\qw&\cdots &&\ctrl{6}&\gate{Z^{l'}}&\qw&\qw&\qw&\qw&\qw&\ctrlo{1}&\ctrlo{3}&\qw&\cdots &&\ctrlo{5}&\gate{(-Z)^{l}}&\qw\\
  \ustick{\log_2d}&{/}\qw &\qw &\multigate{1}{e^{-i{\rm sgn}(T)\theta_{l',|T|/r}S}}&\qswap&\qswap &\qw&\cdots &&\qswap&\qw&\qw&\qw&\qw&\qw&\qw&\multigate{1}{e^{-i{\rm sgn}(T)\theta_{l,|T|/r}S}}&\qswap&\qw&\cdots &&\qswap&\qw&\qw\\
  \lstick{\rho_{{\rm E}_1}}&{/}\qw &\qw &\ghost{e^{-i{\rm sgn}(T)\theta_{l',|T|/r}S}}&\qw&\qw&\qw&\qw&\qw&\qw&\measure{\mbox{out}}&&&&\lstick{\rho_{{\rm E}_1}}&{/}\qw&\ghost{{e^{-i{\rm sgn}(T)\theta_{l,|T|/r}S}}}&\qw&\qw&\qw&\qw&\qw&\measure{\mbox{out}}\\
  \lstick{\rho_{{\rm E}_2}}&{/}\qw &\qw&\qw &\qswap&\qw&\qw&\qw&\qw&\qw&\measure{\mbox{out}}&&&&\lstick{\rho_{{\rm E}_2}}&{/}\qw&\qw&\qswap&\qw&\qw&\qw&\qw&\measure{\mbox{out}}\\
  \lstick{\rho_{{\rm E}_3}}&{/}\qw &\qw&\qw &\qw&\qswap&\qw&\qw&\qw&\qw&\measure{\mbox{out}}&&&&\vdots\\
  \vdots&&&&&&&&&&&&&&\lstick{\rho_{{\rm E}_{2l+1}}}&{/}\qw&\qw&\qw&\qw&\cdots &&\qswap&\measure{\mbox{out}}\\
  \lstick{\rho_{{\rm E}_{2l'+1}}}&{/}\qw &\qw&\qw&\qw&\qw&\qw&\cdots &&\qswap&\measure{\mbox{out}}&&&\\
  }
\\
\end{tabular}
\caption{Quantum circuit for $\widetilde{\Gamma}_{l,l'}$ in $\hat{\Phi}_r$ to simulate $\mathcal{U}[e^{-i\rho T}]$.
The angle for the partial SWAP gate is given by Eq.~\eqref{eq:angle_for_partialswap}.}
\label{fig:element_modified_random_map}
\end{figure}

From the theory of LCS in Section~\ref{sec:theory_of_LCS}, we can construct the corresponding quantum channel $\widetilde{\Phi}_{l_1,l'_1,...,l_r,l_r'}$ for the superoperator $\Gamma_{l_r,l'_r}\circ \cdots \circ \Gamma_{l_1,l_1'}$ by replacing each superoperation according to Table~\ref{tab:translationlist}.
More precisely, we define
\begin{equation}\label{eq:sample_partial_expression}
    \widetilde{\Phi}_{l_1,l_1',...,l_r,l_r'}:=\widetilde{\Gamma}_{l_r,l'_r}\circ\cdots\circ\widetilde{\Gamma}_{l_1,l'_1},
\end{equation}
where $\widetilde{\Gamma}_{l,l'}$ is defined as in Fig.~\ref{fig:element_modified_random_map}. 
Importantly, from Eq.~\eqref{def of L}, the quantum channel $\widetilde{\Phi}_{l_1,l_1',...,l_r,l_r'}$ contains at most 
\begin{equation}\label{eq:hs_cal_copycount}
    2r+4rL
    = 2r + 4r\left\lceil\frac{1}{2}\frac{\ln(3r/\varepsilon)}{W_0((1/e)\ln(3r/\varepsilon))}-\frac{1}{2}\right\rceil
\end{equation}
copies of the target quantum state $\rho$; see Fig.~\ref{fig:element_modified_random_map} on how to count the number $r\times(2+4L)$. 
This is indeed the number of copies of $\rho$ necessary to simulate the channel approximating $\mathcal{U}[e^{-i\rho T}]$, which thus proves Eq.~\eqref{number of copies in Th S2}.
We remark that to obtain $\widetilde{\Gamma}_{l,l'}$, we translate each $\Upsilon^{\S}_{\delta,k}$ in ${\Gamma}_{l,l'}$ into Fig.~\ref{fig:qcircuit_for_upsilon} (which is the case for $k=3$).

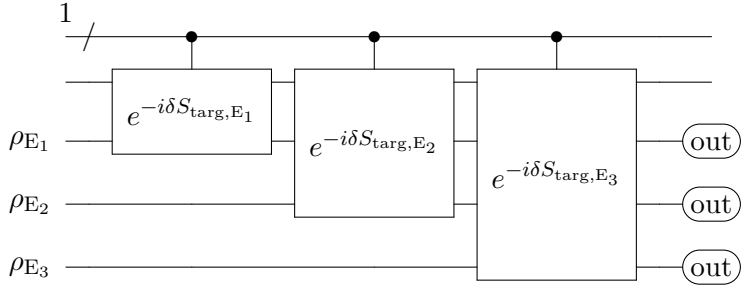
\begin{figure}[h]
\centering
\begin{tabular}{c}
\\
~~~~~\Qcircuit @C=0.8em @R=1em {
    \ustick{1}&{/} \qw&\ctrl{1}&\ctrl{1}&\ctrl{1}&\qw&\qw\\
    &\qw&\multigate{1}{e^{-i\delta S_{\rm targ, E_1}}}&\multigate{2}{e^{-i\delta S_{\rm targ, E_2}}}&\multigate{3}{e^{-i\delta S_{\rm targ, E_3}}}&\qw&\qw\\
    \lstick{\rho_{\rm E_1}}&\qw&\ghost{e^{-i\delta S_{\rm targ, E_1}}}&\ghost{e^{-i\delta S_{\rm targ, E_2}}}&\ghost{e^{-i\delta S_{\rm targ, E_3}}}&\qw&\measure{\mbox{out}}\\
    \lstick{\rho_{\rm E_2}}&\qw&\qw &\ghost{e^{-i\delta S_{\rm targ, E_2}}}&\ghost{e^{-i\delta S_{\rm targ, E_3}}}&\qw&\measure{\mbox{out}}\\
    \lstick{\rho_{\rm E_3}}&\qw&\qw&\qw&\ghost{e^{-i\delta S_{\rm targ, E_3}}}&\qw&\measure{\mbox{out}}\\
  }
\\
\\
\end{tabular}
\caption{Translated quantum circuit for $\Upsilon^{\S}_{\delta,k}$ (for the case $k=3$), where ``out'' means the trace out.}
\label{fig:qcircuit_for_upsilon}
\end{figure}
 
\noindent
Note that the circuit for $\Upsilon_{\delta,k}$ is obtained by replacing $1$-control with $0$-control in the above circuit.
We can directly check that the quantum circuit Fig.~\ref{fig:qcircuit_for_upsilon} has the following action:
\begin{equation}
    \begin{pmatrix}
        *&A_{01}\\
        A_{10}&*
    \end{pmatrix}\mapsto 
    \begin{pmatrix}
        *&\Upsilon^\S_{\delta,k}(A_{01})\\
        (\mathcal{J}\circ \Upsilon^\S_{\delta,k}\circ \mathcal{J})(A_{10})&*
    \end{pmatrix}.
\end{equation}
Therefore, 
\begin{align}
    &C^{2r}\sum_{l_1,l'_1,...l_r,l_r'} p_{l_1}p_{l_1'}\cdots p_{l_r}p_{l_r'}\widetilde{\Phi}_{l_1,l_1',...,l_r,l_r'}\left(\ketbra{+}_{\rm anc}\otimes A\right)\notag\\
    &=C^{2r}\sum_{l_1,l'_1,...l_r,l_r'} p_{l_1}p_{l_1'}\cdots p_{l_r}p_{l_r'}\frac{1}{2}\begin{pmatrix}
        *&\left(\Gamma_{l_r,l'_r}\circ \cdots \circ \Gamma_{l_1,l'_1}\right)(A)\\
        \left(\mathcal{J}\circ \Gamma_{l_r,l'_r}\circ \cdots \circ \Gamma_{l_1,l'_1}\circ \mathcal{J}\right)(A)&*
    \end{pmatrix}\notag\\
    &=\frac{1}{2}\begin{pmatrix}
        *&S_L^{r}A(S^\dagger_L)^r\\
        S_L^{r}A(S^\dagger_L)^r&*
    \end{pmatrix}.
\end{align}
In the final line, we used the Hermitian-preserving property of the linear map $S_L^r\bullet (S_L^\dagger)^r$.
Now, we define a random modified-quantum map $\hat{\Phi}_r$ that realizes
\begin{equation}
    \hat{\Phi}_r(\bullet)={\rm tr}_{\rm anc}\left[(X_{\rm anc}\otimes \bm{1})\widetilde{\Phi}_{l_1,l_1',...,l_r,l_r'}\left(\ketbra{+}_{\rm anc}\otimes\bullet\right)\right]
\end{equation}
with the probability 
\begin{equation}
    p_{l_1}p_{l_1'}\cdots p_{l_r}p_{l_r'}.
\end{equation}
In other words, we randomly sample the map $\hat{\Phi}_r$ according to the probability distribution $p_{l_1}p_{l_1'}\cdots p_{l_r}p_{l_r'}$.
Using this notation, Eq.~\eqref{eq:SAS^r in the middle} has the following expression:
\begin{equation}
    S_L^{r}A(S^\dagger_L)^r=C^{2r}\mathbb{E}_{\hat{\Phi}_r}\left[\hat{\Phi}_r(A)\right].
\end{equation}
Here, from $C(T/r)\leq e^{T^2/r^2}$ stated below Eq.~\eqref{def of C}, $C^{2r}$ can be upper bounded by 
\begin{equation}
    C^{2r}=[C(T/r)]^{2r}\leq e^{2T^2/r},
\end{equation}
which thus proves Eq.~\eqref{bound of C in Th S2}.


Finally, we bound the approximation error. 
We observe that if an operator $A$ and a unitary $U$ satisfy $\|A-U\|\leq \varepsilon/(3r)$ ($\varepsilon>0$), then
\begin{align}
    \|A^r-U^r\| &= \left\|\sum_{k=0}^{r-1} A^{k}(A-U)U^{r-1-k}\right\|\notag\\
    &\leq r\|A-U\|\left(\max\{\|A\|,\|U\|\}\right)^{r-1}\leq (\varepsilon/3)(1+(\varepsilon/3r))^{r}\leq (\varepsilon/3)e^{\varepsilon/3},
\end{align}
where the second inequality follows from the triangle inequality $\|A-U \| + \| U \| \geq \|A\|$, and the last inequality uses the bound $(1+\frac{a}{r})^{r}\leq e^{a}$, which holds for all $r>0$ and $a>0$.
Using $\|S_L-e^{-i\rho T/r}\|\leq \varepsilon/(3r)$ in Eq.~\eqref{eq:approx_error_SL} 
and the above inequality, we have
\begin{align}
    &\left\|(e^{-i\rho T}\otimes \bm{1}_R)A(e^{-i\rho T}\otimes \bm{1}_R)^\dagger-(S_L^r\otimes \bm{1}_R)A(S_L^r\otimes \bm{1}_R)^\dagger\right\|_{1}\notag\\
    &~~~\leq \left\|(e^{-i\rho T}\otimes \bm{1}_R - S_L^r\otimes \bm{1}_R)A(e^{-i\rho T}\otimes \bm{1}_R)^\dagger\|_{1}  + \|(S_L^r\otimes \bm{1}_R)A(e^{-i\rho T}\otimes \bm{1}_R - S_L^r\otimes \bm{1}_R)^\dagger\right\|_{1}\notag\\
    &~~~\leq \|e^{-i\rho T}\otimes \bm{1}_R-S_L^r\otimes \bm{1}_R\|\|A\|_1+\|S_L^r\otimes \bm{1}_R\|\|A\|_1\|(e^{-i\rho T}\otimes \bm{1}_R-S_L^r\otimes \bm{1}_R)^\dagger\|\notag\\
    &~~~\leq \|e^{-i\rho T}-S_L^r\|\|A\|_1+({\|e^{-i\rho T}\|+\|S_L^r-e^{-i\rho T}\|})\|A\|_1\|e^{-i\rho T}-S_L^r\|\notag\\
    &~~~\leq 2\|A\|_1\|e^{-i\rho T}-S_L^r\| \left(1+\frac{1}{2}\|e^{-i\rho T}-S_L^r\|\right)\leq \|A\|_1\frac{2\varepsilon}{3}e^{\varepsilon/3}\left(1+\frac{\varepsilon}{6}e^{\varepsilon/3}\right)\leq \varepsilon\|A\|_1
\end{align}
for any operator $A$ and any finite-dimensional ancilla system $R$. In the second-to-third line, we bound each term using Hölder's inequality, $\|XYZ\|_{1}\leq\|X\|\|Y\|_{1}\|Z\|$, and
the final inequality follows from the bound $\frac{2\varepsilon}{3}e^{\varepsilon/3}\left(1+\frac{\varepsilon}{6}e^{\varepsilon/3}\right)\leq\varepsilon$ holds under
the assumption $\varepsilon<1/2$. 
Thus, we conclude that
\begin{equation}
    \left\|\mathcal{U}[e^{-i\rho T}]-C^{2r}\mathbb{E}_{\hat{\Phi}_r}[\hat{\Phi}_r]\right\|_{\diamond}
    = \max_R \left\|\left(\mathcal{U}[e^{-i\rho T}]-S_L^r\bullet (S_L^r)^\dagger\right)\otimes \mathcal{I}_R\right\|_{1\to 1}\leq \varepsilon.
\end{equation}
Thus, Eq.~\eqref{eq:diamond_distance_alpha} has been proven. 
\end{proof}

\bigskip
\noindent
Finally, we show the technical lemma to complete the proof of Theorem~\ref{thm:main_supple}.
\begin{lem}\label{lem:DME_taylor_truncation}
    Let $\rho$ be an arbitrary density matrix.
    Given $T\in\mathbb{R}\backslash\{0\}$, $r\geq \max\{1,|T|\}$, $\varepsilon\in (0,1/e)$, for any positive integer $L$ satisfying 
    \begin{equation}
        L\geq\frac{1}{2}\frac{\ln(r/\varepsilon)}{W_0((1/e)\ln(r/\varepsilon))}-\frac{1}{2},
    \end{equation}
    we obtain
    \begin{equation}\label{eq:taylor_oneside}
        \left\| S_L(T/r) - e^{-i\rho T/r}\right\|
        = \left\|\sum_{q=0}^{2L+1}\frac{(T/r)^{q}}{q!}(-i\rho)^{q}
             - e^{-i\rho T/r}\right\|
        \leq \frac{\varepsilon}{r}.
    \end{equation}
    Here, $W_0$ is the principal branch of the Lambert-W function.
    Also, a simple sufficient choice of $L$ satisfying Eq.~\eqref{eq:taylor_oneside} is 
    \begin{equation}\label{eq:l-condition}
    L\geq  \frac{\ln \left({r}/{\varepsilon}\right)}{\ln \ln \left({r}/{\varepsilon}\right)}-\frac{1}{2}.
    \end{equation}
\end{lem}

\begin{proof}
From the Taylor expansion, 
\begin{align}
    \left\|\sum_{q=0}^{2L+1}\frac{(T/r)^{q}}{q!}(-i\rho)^{q}-e^{-i\rho T/r}\right\|
    &=\left\|
    \sum_{q=2L+2}^{\infty} \frac{(T/r)^q}{q!}(-i\rho)^q
    \right\|\leq \sum_{q=2L+2}^{\infty}\left\| \frac{|T|^q}{r^q q!} (-i\rho)^q \right\|\nonumber\\
    &\leq \sum_{q=2L+2}^{\infty} \frac{1}{q!}\nonumber\\
    &< \frac{1}{(2L+1)!}\leq \left(\frac{e}{2L+1}\right)^{2L+1}.
\end{align}
In passing from the second to third line, we used the relation $\sum_{k=1}^{\infty}\frac{1}{(n+k)!}  <\frac{1}{(2L+1)!}$ with $n:=2L+1$. Indeed, for every $k$ we have $(n+k)!\geq n!(n+1)^k $ and hence
\begin{equation}
    \sum_{k=1}^{\infty} \frac{1}{(n+k)!} \leq \frac{1}{n!} \sum_{k=1}^{\infty} \left( \frac{1}{n+1} \right)^k = \frac{1}{n!} \cdot \frac{\frac{1}{n+1}}{1-\frac{1}{n+1}} = \frac{1}{n \cdot n!} \leq \frac{1}{(2L+1)!}.
\end{equation}
In the last line, we applied the lower bound of the factorial $n!\geq(n/e)^n$, which holds for all positive integers $n$.
    For any $c\in(0,1)$, let $s(c)\geq e$ be the unique solution to the equation $(e/s)^s=c$ ($s>e$).
    Observing that
    \begin{equation}
        \frac{\ln(1/c)}{e}=\frac{\ln(1/c)}{s(c)}e^{\frac{\ln(1/c)}{s(c)}},
    \end{equation}
    we can write $s(c)$ with the principal branch of the Lambert-W function $W_0(x)$ satisfying $W_0(x)\exp({W_0(x)})=x$:
    \begin{equation}
        s(c)=\frac{\ln(1/c)}{W_0((1/e)\ln(1/c))}.
    \end{equation}
    Setting $c=\varepsilon/r\in(0,1)$, we have Eq.~\eqref{eq:taylor_oneside} if the integer $L$ satisfies 
    \begin{equation}
        2L+1\geq s(c)=\frac{\ln(r/\varepsilon)}{W_0((1/e)\ln(r/\varepsilon))}.
    \end{equation}
    
    Now, for any $\delta'\in (0,1/e)$, we show that an integer $Q$ satisfies 
    ${1}/{Q!}\leq {\delta'}$ if $Q\geq 2\ln(1/\delta')/\ln (\ln(1/\delta'))$. 
    Defining  $\kappa:=\ln(1/\delta')>1$ and an integer $Q:=\lceil 2\kappa/\ln(\kappa)\rceil>2e$, we obtain 
    \begin{equation}
        \ln \frac{1}{Q!}<Q-Q\ln Q=Q\ln \kappa\cdot \frac{1-\ln Q}{\ln \kappa}\leq -\frac{1}{2}Q\ln\kappa<-\kappa.
    \end{equation}
    Thus, by taking $Q = 2L+1$ and $\delta' = \varepsilon/r$, we arrive at Eq.~\eqref{eq:l-condition}.
\end{proof}

\subsection{Construction of the random modified-quantum map}
\label{sec:S22}

We here show the explicit construction method of the random modified-quantum map $\hat{\Phi}_r$ which approximates $\mathcal{U}[e^{-i\rho T}]$ in the sense of mean.
First of all, for a desired error $\varepsilon\in(0,1/2)$, we calculate the integer $L$ 
\begin{equation}
    L:=\left\lceil\frac{1}{2}\frac{\ln(3r/\varepsilon)}{W_0((1/e)\ln(3r/\varepsilon))}-\frac{1}{2}\right\rceil.
\end{equation}
Then, we sample a sequence of random variables $(l_1,l_1',...,l_r,l_r')$, the elements of which are independent random variables $l$ following the identical probability distribution \eqref{def of C}, i.e., 
\begin{equation}\label{eq:sample_distribution_l}
    p_l(T/r):=\frac{1}{C(T/r)}\frac{(T/r)^{2l}}{(2l)!}\sqrt{1+\left(\frac{T/r}{2l+1}\right)^{2}},
\end{equation}
\begin{equation}
    C(T/r):=\sum_{l=0}^{L}\frac{(T/r)^{2l}}{(2l)!}\sqrt{1+\left(\frac{T/r}{2l+1}\right)^{2}}.
\end{equation}
We note that the support of the random variable $l$ is $\{0,1,...,L\}$.
Examples of the probability distribution are provided in Fig.~\ref{fig:probdist_DME}. 
For the sampled sequence $(l_1,l_1',...,l_r,l_r')$, we construct the quantum channel 
\eqref{eq:sample_partial_expression}, i.e., 
\begin{equation}
\label{eq:sample_each_expression}
     \widetilde{\Phi}_{l_1,l_1',...,l_r,l_r'}
       :=\widetilde{\Gamma}_{l_r,l'_r}\circ\cdots\circ\widetilde{\Gamma}_{l_1,l'_1},
\end{equation}
where $\widetilde{\Gamma}_{l,l'}$ is a quantum channel defined in Fig.~\ref{fig:element_modified_random_map}.
As a result, we realize the map $\hat{\Phi}_r$:
\begin{equation}
\label{eq:sample_total_expression}
    \hat{\Phi}_r(\bullet)={\rm tr}_{\rm anc}\left[(X_{\rm anc}\otimes \bm{1})\widetilde{\Phi}_{l_1,l_1',...,l_r,l_r'}\left(\ketbra{+}_{\rm anc}\otimes \bullet\right)\right].
\end{equation}
Note again that this map has been randomly chosen according to the probability distribution of $(l_1,l_1',...,l_r,l_r')$, each of which is independently and identically subjected to Eq.~\eqref{def of C}.


As shown in Fig.~\ref{fig:element_modified_random_map}, 
to implement a quantum circuit for $\widetilde{\Phi}_{l_1,l_1',...,l_r,l_r'}$, we use the controlled version of the SWAP gate and the partial SWAP gate; the numbers of these gates contained in the circuit are summarized in Table~\ref{tab:resources}.
We remark that the controlled SWAP and controlled partial SWAP are further decomposed into $\mathcal{O}(\log d)$ elementary gates~\cite{Kimmel2017-hv}, where $d$ is the dimension of the target density matrix $\rho$.
Also, the controlled partial SWAP gate can be implemented using three controlled SWAP gates and a constant number of two-qubit elementary gates as shown in Fig.~\ref{fig:controlled_partial_swap}.
The circuit construction is based on the linear combination of unitaries method followed by a single-step oblivious amplitude amplification~\cite{berry2014exponential}.

\begin{figure}[htb]
 \centering
 \includegraphics[scale=0.8]{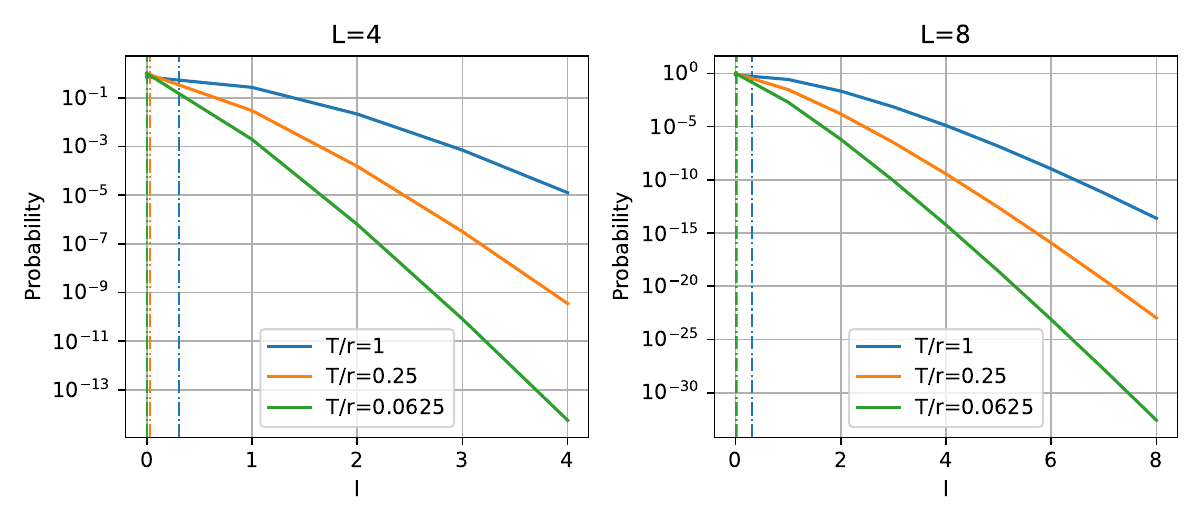}
 \caption{Examples of the probability distribution $p_l(T/r)$ of the random variable $l$ whose support is $\{0,1,...,L\}$.
 The vertical dashdot lines represent the mean value $\mathbb{E}[l]$ of each probability distribution. 
 The probability of getting a large $l$ is exponentially small.}
 \label{fig:probdist_DME}
\end{figure}

\renewcommand{\baselinestretch}{1.3}
\begin{table}[ht]
\centering
\begin{tabular}{c|c}
\hline
\textbf{Quantum resource} & \textbf{Number per circuit} \\ \hline\hline
Target quantum state $\rho$ in $d$ dim.& $2r+\sum_{j=1}^r 2(l_j'+l_j)$ \\
Controlled SWAP gate $S$& $\sum_{j=1}^r 2(l_j'+l_j)$ \\
Controlled partial SWAP gate $e^{-i\delta S}$& $2r$ \\
\hline
\end{tabular}
\renewcommand{\baselinestretch}{1.0}
\caption{Summary of quantum resources for constructing the random map $\hat{\Phi}_r$, 
which is determined from the sequence of independent random variables $(l_1,l_1',...,l_r,l_r')$ with each $l$ subjected to the probability distribution Eq.~\eqref{eq:sample_distribution_l}. 
The SWAP gate $S$ acts on a $2d$ dimensional system, with $d$ the dimension of $\rho$. 
The total gate complexity scales as $\mathcal{O}(\log d)$. }
\label{tab:resources}
\end{table}

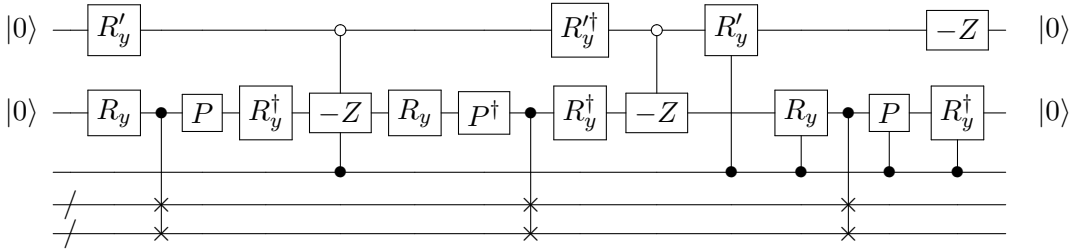
\begin{figure}[tb]
\centering
\begin{tabular}{c}
~~~~~\Qcircuit @C=0.6em @R=1em {
  \lstick{\ket{0}}&\qw & \gate{R_y'}&\qw &\qw&\qw&\ctrlo{1}&\qw&\qw&\qw&\gate{R_y'^\dagger}&\ctrlo{1}&\gate{R_y'}&\qw&\qw&\qw&\gate{{-Z}}&\qw&\rstick{\ket{0}}\\
  \lstick{\ket{0}}&\qw & \gate{R_y} &\ctrl{3} &\gate{P}&\gate{R_y^\dagger}&\gate{-Z}&\gate{R_y}&\gate{P^\dagger}&\ctrl{3}&\gate{R_y^\dagger}&\gate{-Z}&\qw&\gate{R_y}&\ctrl{3}&\gate{P}&\gate{R_y^\dagger}&\qw&\rstick{\ket{0}}\\
  &\qw & \qw &\qw &\qw&\qw&\ctrl{-1}&\qw&\qw&\qw&\qw&\qw&\ctrl{-2}&\ctrl{-1}&\qw&\ctrl{-1}&\ctrl{-1}&\qw\\
  &{/}\qw &\qw &\qswap&\qw&\qw&\qw&\qw&\qw&\qswap&\qw&\qw&\qw&\qw&\qswap&\qw&\qw&\qw\\
  &{/}\qw &\qw &\qswap&\qw&\qw&\qw&\qw&\qw&\qswap&\qw&\qw&\qw&\qw&\qswap&\qw&\qw&\qw\\
  }
\\
\end{tabular}
\caption{Quantum circuit for constructing the controlled $e^{-i\theta S}$ gate using additional two ancilla qubits. 
The first two qubits initialized $\ket{00}$ return $\ket{00}$ at the end of the circuit.
The third qubit is the control qubit.
For $\gamma:=|\cos\theta|+|\sin\theta|$, we define $R_y,R_y',P$ as $R_y:\ket{0}\mapsto \sqrt{|\cos\theta|/\gamma}\ket{0}+\sqrt{|\sin\theta|/\gamma}\ket{1}$, $R_y':\ket{0}\mapsto (\gamma/2)\ket{0}+\sqrt{1-\gamma^2/4}\ket{1}$, 
$P={\rm diag}({\rm sgn(\cos\theta)},-i{\rm sgn}(\sin\theta))$.}
\label{fig:controlled_partial_swap}
\end{figure}


\subsection{Procedure for expectation value estimation}
\label{sec:S23}


We can use the random modified-quantum map $\hat{\Phi}_r$ to simulate (approximate) $\mathcal{U}[e^{-i\rho T}]$. 
We here focus on the task of computing the expectation value 
\begin{equation}
    {\rm tr}\left[O\mathcal{U}[e^{-i\rho T}](\sigma )\right]
     = {\rm tr}\left[O e^{-i\rho T}\sigma e^{i\rho T}\right],
\end{equation}
for a given initial state $\sigma$ and an observable $O$.
A more general setup will be discussed in Section~\ref{sec:nearlyopt_qalg}.
The key observation is that from the proof of Theorem~\ref{thm:main_supple}, the map $\mathbb{E}[\hat{\Phi}_r]$ can also be written as
\begin{equation}
\label{eq:mean_modified channel_explicit}
    \mathbb{E}[\hat{\Phi}_r](\bullet)={\rm tr}_{\rm anc}\left[(X_{\rm anc}\otimes \bm{1})\left(\sum_{l_1,l_1',...,l_r,l_r'}p_{l_1}p_{l_1'}\cdots p_{l_r}p_{l_r'}\widetilde{\Phi}_{l_1,l_1',...,l_r,l_r'}\right)\left(\ketbra{+}_{\rm anc}\otimes \bullet\right)\right].
\end{equation}
Using this expression, we obtain
\begin{align}
    &{\rm tr}\left[O\mathcal{U}[e^{-i\rho T}](\sigma )\right]
    \simeq C^{2r}{\rm tr}\left[O\mathbb{E}[\hat{\Phi}_r](\sigma)\right]
    \notag\\
    &=C^{2r}{\rm tr}\left[(X_{\rm anc}\otimes O)\left(\sum_{l_1,l_1',...,l_r,l_r'}p_{l_1}p_{l_1'}\cdots p_{l_r}p_{l_r'}\widetilde{\Phi}_{l_1,l_1',...,l_r,l_r'}\right)\left(\ketbra{+}_{\rm anc}\otimes \sigma\right)\right].
\end{align}
This indicates that the following procedure generates an estimate of the target expectation value:
\begin{itemize}
    \item[1.] Randomly generate $(l_1,l_1',...,l_r,l_r')$ according to the probability $p_{l_1}p_{l_1'}\cdots p_{l_r}p_{l_r'}$
    \item[2.] Measure the observable $X_{\rm anc}\otimes O$ with the quantum state $\widetilde{\Phi}_{l_1,l_1',...,l_r,l_r'}(\ketbra{+}_{\rm anc}\otimes \sigma)$
    \item[3.] Multiply $C^{2r}$ on the measurement outcome. 
\end{itemize}
By repeating the procedure independently and taking the sample mean of the outputs, we can reduce the variance of the estimator. 
Although the final multiplication of $C^{2r}~(\leq e^{2T^2/r})$ amplifies the variance, we can make $C^{2r}$ arbitrary close to one by setting $r\geq T^2$. 
Hence, with this choice, we can ignore the measurement overhead in our procedure.
Since the approximation error of $\hat{\Phi}_r$ is at most $\varepsilon$ under the conditions of Theorem~\ref{thm:main_supple}, the estimation bias is ensured to be $\varepsilon\|O\|$ because of 
\begin{align}
     \left|C^{2r}{\rm tr}\left[O\mathbb{E}[\hat{\Phi}_r](\sigma)\right]-{\rm tr}\left[O\mathcal{U}[e^{-i\rho T}](\sigma )\right]\right|\leq \|O\|\|C^{2r}\mathbb{E}[\hat{\Phi}_r]-\mathcal{U}[e^{-i\rho T}]\|_{\diamond}
     \leq \varepsilon \|O\|.
\end{align}
Consequently, to estimate the target expectation value within an additive error $\mathcal{O}(\varepsilon \|O\|)$ (in a high probability), each of the quantum circuits in our procedure has at most 
\begin{equation}
\label{eq:simple_dme_copycost}
     N=2r+4rL=\mathcal{O}\left(T^2 \frac{\log(T/\varepsilon)}{\log\log(T/\varepsilon)}\right)
\end{equation}
copies of $\rho$, where we have employed the simplified version of $N$ given in Eq.~\eqref{def of L} together with $r=T^2$.
Also, from the discussion around Table~\ref{tab:resources}, its gate complexity needs additional $\log(d)$ factor with $d$ the dimension of $\rho$, which is thus given by 
\begin{equation}\label{eq:simple_dme_gatecost}
    \mathcal{O}\left(T^2 \frac{\log(T/\varepsilon)}{\log\log(T/\varepsilon)}\log(d)\right).
\end{equation}
These complexities logarithmically depend on the error $\varepsilon$.
Note that the total number of copies in the whole estimation procedure is equal to Eq.~\eqref{eq:simple_dme_copycost} multiplied by the final measurement cost $\mathcal{O}(1/\varepsilon^2)$, i.e., the number of repetitions to reduce the variance; 
consequently, for the task of observable estimation, the total complexity is also reduced from $\mathcal{O}(1/\varepsilon^3)$~\cite{Lloyd2014-nn} to $\tilde{\mathcal{O}}(1/\varepsilon^2)$ via the proposed method.

\subsection{Several variants}\label{supplesec_variants}
\subsubsection{Controlled version of $e^{-i\rho T}$}

From Theorem~\ref{thm:main_supple}, we can directly simulate the controlled version of $e^{-i\rho T}$.
The idea simply comes from:
\begin{equation}
    \mathcal{U}\left[e^{-iT\ketbra{1}\otimes \rho}\right]=\mathcal{U}\left[\ketbra{0}\otimes \bm{1}+\ketbra{1}\otimes e^{-i\rho T}\right].
\end{equation}
Thus, by taking $\ketbra{1}\otimes \rho$ as the target state in Theorem~\ref{thm:main_supple} instead of $\rho$, we can simulate the controlled operation of $e^{-i\rho T}$.

Note that in some applications, we use the following controlled operation:
\begin{equation}\label{eq:double_control_DME}
    \mathcal{U}\left[\ketbra{0}\otimes e^{i\rho T}+\ketbra{1}\otimes e^{-i\rho T}\right].
\end{equation}
This operation can be rewritten as
\begin{equation}
    \mathcal{U}\left[e^{iT\ketbra{0}\otimes \rho}\right]\circ \mathcal{U}\left[e^{-iT\ketbra{1}\otimes \rho}\right],
\end{equation}
so we can simulate it with two times use of the procedure in Theorem~\ref{thm:main_supple}; hence the number of necessary copies is doubled. 
However, instead of this approach, we here provide a more efficient procedure to simulate Eq.~\eqref{eq:double_control_DME}.
The following approach described in Theorem~\ref{thm:controlledDME} requires the same number of copies as the case without control in Theorem~\ref{thm:main_supple}.

\begin{thm}\label{thm:controlledDME}
    Let $\rho$ be an unknown but accessible $d$-dimensional quantum state.
    For any $T\in \mathbb{R}\backslash\{0\}$ and any integer $r\geq \max\{1,|T|\}$, 
    we can construct a random modified-quantum map $\hat{\Phi}'_r$ using at most $N$ copies of $\rho$ and a single-ancilla measurement such that
    \begin{equation}
        \left\|C^{2r}\mathbb{E}[\hat{\Phi}'_r] 
          - \mathcal{U}\left[\ketbra{0}\otimes e^{i\rho T}+\ketbra{1}\otimes e^{-i\rho T}\right] \right\|_{\diamond}\leq \varepsilon
    \end{equation}
    holds for a given $\varepsilon\in(0,1/2)$, if $N$ satisfies 
    \begin{equation}
        N\geq 2r+4r\left\lceil\frac{1}{2}\frac{\ln(3r/\varepsilon)}{W_0((1/e)\ln(3r/\varepsilon))}-\frac{1}{2}\right\rceil.
    \end{equation}
    Here, $W_0$ is the principal branch of the Lambert-W function, and $C$ is a positive value that satisfies 
    \begin{equation}
        1<C^{2r}\leq e^{2T^2/r}.
    \end{equation}
\end{thm}
\begin{proof}
    The following has the same flow as the proof of Theorem~\ref{thm:main_supple}.
    For the operator $S_L(T/r)$ defined by Eq.~\eqref{eq:def_SL_mainthm1}, we consider the following operation:
    \begin{equation}
        {\rm c}S_L:=\ketbra{0}\otimes S_L(-T/r)+\ketbra{1}\otimes S_L(T/r).
    \end{equation}
    This can be rewritten as
    \begin{align}
        {\rm c}S_L
        &= \sum_{l=0}^L C_l(T/r)\left({\bm{1}}\otimes (-i\rho)^{2l} 
    \right)\cdot \left(\ketbra{0}\otimes \left(\cos(\theta_{l,|T|/r}){\bm{1}}+i{\rm sgn}(T)\sin(\theta_{l,|T|/r})\rho\right) \right)\notag\\
        &~~~~~+\sum_{l=0}^L C_l(T/r)\left({\bm{1}}\otimes (-i\rho)^{2l}\right) \cdot \left( \ketbra{1}\otimes \left(\cos(\theta_{l,|T|/r}){\bm{1}}-i{\rm sgn}(T)\sin(\theta_{l,|T|/r})\rho\right)\right)\notag\\
        &= \sum_{l=0}^L C_l(T/r)\sum_{j=0}^1\ket{j}\bra{j}\otimes  (-i\rho)^{2l} \left(\cos(\theta_{l,|T|/r}){\bm{1}}-i(2j-1){\rm sgn}(T)\sin(\theta_{l,|T|/r})\rho\right)
    \end{align}
    Now, similar to $\Gamma_{l,l'}$ in the proof of Theorem~\ref{thm:main_supple}, we define the superoperator $\Gamma'_{l,l'}$ by
    \begin{align}
        &\Gamma'_{l,l'}(\bullet):=\notag\\
        &{\rm tr}_{\rm E}\left[\left({\bm{1}}\otimes \prod_{k=1}^{2l} e^{-i\pi/2 S_{{\rm targ},{\rm E}_{k}}} \right)\cdot
        \left(\sum_{j=0}^1\ket{j}\bra{j}\otimes e^{-i(2j-1){\rm sgn}(T)\theta_{l,|T|/r}S_{{\rm targ},{\rm E}_1}}\right)(\bullet)\otimes \rho^{\otimes 2l+1}\right]\notag\\
        &~~~~~\circ {\rm tr}_{\rm E}\left[(\bullet)\otimes \rho^{\otimes 2l'+1}\cdot \left(\sum_{j=0}^1\ket{j}\bra{j}\otimes e^{i(2j-1){\rm sgn}(T)\theta_{l',|T|/r}S_{{\rm targ},{\rm E}_1}}\right)\cdot \left({\bm{1}}\otimes \prod_{k=1}^{2l'} e^{i\pi/2 S_{{\rm targ},{\rm E}_k}}\right)\right].
    \end{align}
    Then, we have
    \begin{equation}
        {\rm c}S_L A({\rm c}S_L)^\dagger =\sum_{l,l'}C_lC_{l'}\Gamma'_{l,l'}(A)
    \end{equation}
    and 
    \begin{equation}
        {\rm c}S_L^r A({\rm c}S^r_L)^\dagger =\sum_{l_1,l'_1,...,l_r,l_r'}C_{l_1}C_{l'_1}\cdots C_{l_r}C_{l'_r}\left(\Gamma'_{l_r,l'_r}\circ \cdots \circ \Gamma'_{l_1,l'_1}\right)(A).
    \end{equation}
    
    From the theory of LCS, we can construct the corresponding quantum channel $\widetilde{\Phi}'_{l_1,l'_1,...,l_r,l_r'}$ for the superoperator $\Gamma'_{l_r,l'_r}\circ \cdots \circ \Gamma'_{l_1,l_1'}$ as 
    \begin{equation}
        \widetilde{\Phi}'_{l_1,l'_1,...,l_r,l_r'}=\widetilde{\Gamma}'_{l_r,l'_r}\circ \cdots \circ \widetilde{\Gamma}'_{l_1,l_1'},
    \end{equation}
    where $\widetilde{\Gamma}'_{l,l'}$ is defined as in Fig.~\ref{fig:element_modified_random_map_Wcontrol}. 
    From the circuit description, the quantum channel $\widetilde{\Phi}'_{l_1,l'_1,...,l_r,l_r'}$ has at most $2r+4rL$ copies of $\rho$.
    Using these superoperators $\widetilde{\Phi}'_{l_1,l'_1,...,l_r,l_r'}$, we can define a random modified-quantum map $\hat{\Phi}'_r$ that realizes 
    \begin{equation}
    \hat{\Phi}'_r(\bullet)={\rm tr}_{\rm anc}\left[(X_{\rm anc}\otimes \bm{1})\widetilde{\Phi}'_{l_1,l_1',...,l_r,l_r'}\left(\ketbra{+}_{\rm anc}\otimes \bullet\right)\right]
    \end{equation}
    with the probability 
    \begin{equation}
    p_{l_1}p_{l_1'}\cdots p_{l_r}p_{l_r'}.
    \end{equation}
    Then, we arrive at
    \begin{equation}
        C^{2r}\mathbb{E}[\hat{\Phi}'_r(A)]={\rm c}S_L^r A({\rm c}S^r_L)^\dagger.
    \end{equation}
    The approximation error between ${\rm c}S_L^r A({\rm c}S^r_L)^\dagger$ and the target unitary channel can be bounded in the same way as the proof of Theorem~\ref{thm:main_supple}.
\end{proof}

\begin{figure*}[tb]
\centering
\begin{tabular}{c}
\\
\\
~~~~~\Qcircuit @C=0.55em @R=1.5em {
  \ustick{1}&{/} \qw&\qw &\ctrl{1}&\ctrl{4}&\ctrl{5}&\qw&\cdots &&\ctrl{7}&\gate{Z^{l'}}&\qw&\qw&\qw&\qw&\qw&\ctrlo{1}&\ctrlo{4}&\qw&\cdots &&\ctrlo{6}&\gate{(-Z)^{l}}&\qw\\
  &{/} \qw&\ustick{~~~~~~~~~~~~~~~~~~~~~~~~~j}\qw &\ctrl{1}&\qw&\qw&\qw&\cdots &&\qw&\qw&\qw&\qw&\qw&\qw&\ustick{~~~~~~~~~~~~~~~~~~~~~~~~~j}\qw&\ctrl{1}&\qw&\qw&\cdots &&\qw&\qw&\qw
  \inputgroupv{2}{3}{.8em}{.8em}{{1+\log_2 d}~~~~~~~~}\\
  &{/}\qw &\qw &\multigate{1}{e^{-i(2j-1){\rm sgn}(T)\theta_{l',|T|/r}S}}&\qswap&\qswap &\qw&\cdots &&\qswap&\qw&\qw&\qw&\qw&\qw&\qw&\multigate{1}{e^{-i(2j-1){\rm sgn}(T)\theta_{l,|T|/r}S}}&\qswap&\qw&\cdots &&\qswap&\qw&\qw\\
  \lstick{\rho_{{\rm E}_1}}&{/}\qw &\qw &\ghost{e^{-i(2j-1){\rm sgn}(T)\theta_{l',|T|/r}S}}&\qw&\qw&\qw&\qw&\qw&\qw&\measure{\mbox{out}}&&&&\lstick{\rho_{{\rm E}_1}}&{/}\qw&\ghost{{e^{-i(2j-1){\rm sgn}(T)\theta_{l,|T|/r}S}}}&\qw&\qw&\qw&\qw&\qw&\measure{\mbox{out}}\\
  \lstick{\rho_{{\rm E}_2}}&{/}\qw &\qw&\qw &\qswap&\qw&\qw&\qw&\qw&\qw&\measure{\mbox{out}}&&&&\lstick{\rho_{{\rm E}_2}}&{/}\qw&\qw&\qswap&\qw&\qw&\qw&\qw&\measure{\mbox{out}}\\
  \lstick{\rho_{{\rm E}_3}}&{/}\qw &\qw&\qw &\qw&\qswap&\qw&\qw&\qw&\qw&\measure{\mbox{out}}&&&&\vdots\\
  \vdots&&&&&&&&&&&&&&\lstick{\rho_{{\rm E}_{2l+1}}}&{/}\qw&\qw&\qw&\qw&\cdots &&\qswap&\measure{\mbox{out}}\\
  \lstick{\rho_{{\rm E}_{2l'+1}}}&{/}\qw &\qw&\qw&\qw&\qw&\qw&\cdots &&\qswap&\measure{\mbox{out}}&&&\\
  }
\\
\\
\end{tabular}
\caption{Quantum circuit for $\widetilde{\Gamma}'_{l,l'}$ in $\hat{\Phi}_r'$ to simulate $\mathcal{U}[\ketbra{0}\otimes e^{i\rho T}+\ketbra{1}\otimes e^{-i\rho T}]$. Here, the angle for the partial SWAP gate is given by Eq.~\eqref{eq:angle_for_partialswap}.}
\label{fig:element_modified_random_map_Wcontrol}
\end{figure*}

\subsubsection{Pure target state $\rho=\ketbra{\psi}$}

Here we focus on the case where the target state $\rho=\ketbra{\psi}$ is a pure state. 
In this case, surprisingly, the accuracy dependence in the required number of copies in Theorem~\ref{thm:main_supple} can be completely eliminated.
To prove the following theorem, it is crucial that the power of a pure state remains the same pure state.
\begin{thm}\label{thm:DME_pure_state}
    Let $\rho=\ketbra{\psi}$ be an unknown but accessible $d$-dimensional pure quantum state.
    For any $T\in [-2\pi,2\pi]\backslash\{0\}$ and any integer $r\geq \max\{1,|T|\}$, 
    we can construct a random modified-quantum map $\hat{\Phi}_{{\rm pure},r}$ using at most $N$ copies of $\rho$ and a single-ancilla measurement such that
    \begin{equation}
        C_{\rm pure}^{2r}\mathbb{E}[\hat{\Phi}_{{\rm pure},r}]=\mathcal{U}[e^{-i\ketbra{\psi}T}]
    \end{equation}
    holds, if $N$ satisfies 
    \begin{equation}
        N\geq 2r.
    \end{equation}
    Here, $C_{\rm pure}$ is a positive value that satisfies 
    \begin{equation}
        1<C_{\rm pure}^{2r}\leq e^{2T^2/r}\leq e^{8\pi^2/r}.
    \end{equation}
\end{thm}

We remark that if $\rho$ is pure, then $e^{-i\rho(T+2k\pi)}=e^{-i\rho T}$ holds for any integer $k$, and thus it is sufficient to consider $|T| \in (0,2\pi]$ without loss of generality.
Due to this fact, the sampling overhead $C^{2r}_{\rm pure}$ is always smaller than a constant; additionally, it can be arbitrarily close to one by increasing $r$.
Therefore, we can exactly simulate the DME of any pure state $\rho=\ketbra{\psi}$ using $N=2r=\mathcal{O}(1)$ copies of the state with a constant sampling overhead.
We emphasize that this improvement is highly powerful as the DME of pure states can be directly applied to simulate a (controlled) reflection operator with respect to an unknown pure state $\ketbra{\psi}$. 
Such a simulated reflection operator appears in various applications, e.g., universal quantum emulator in Section~\ref{apsec:uqe} and quantum precomputation in Section~\ref{apsec:precomp}.


\begin{proof}[Proof of Theorem~\ref{thm:DME_pure_state}]

We first observe that the Taylor series of $e^{-i\ketbra{\psi} T/r}$ can be written as
\begin{align}
    &e^{-i\ketbra{\psi} T/r}\notag\\
    &=\sum_{q=0}^{\infty} \frac{(-iT/r)^q}{q!}(\ketbra{\psi})^q\notag\\
    &=\bm{1}-i(T/r)\ketbra{\psi}+\sum_{q=2}^{\infty} \frac{(-iT/r)^q}{q!}\ketbra{\psi}\notag\\
    &=\sqrt{1+\left(\frac{T}{r}\right)^2}\left(\cos(\theta_{0,|T|/r})\bm{1}-i{{\rm sgn}(T)}\sin(\theta_{0,|T|/r})\ketbra{\psi}\right)+\left[e^{-iT/r}-\left(1-i\frac{T}{r}\right)\right]\ketbra{\psi},
\end{align}
where $\theta_{0,|T|/r}$ is defined by Eq.~\eqref{eq:angle_for_partialswap}.
To simplify the notation, we define 
\begin{equation}
    C_{\rm pure}(T/r):=\sqrt{1+(T/r)^2}+\sqrt{(\cos(T/r)-1)^2+(\sin(T/r)-T/r)^2},
\end{equation}
\begin{equation}
    p_0^{(\rm pure)}:=\frac{\sqrt{1+(T/r)^2}}{C_{\rm pure}(T/r)},~~~p_1^{(\rm pure)}:=\frac{\sqrt{(\cos(T/r)-1)^2+(\sin(T/r)-T/r)^2}}{C_{\rm pure}(T/r)},
\end{equation}
\begin{equation}
    e^{i\phi(T/r)}:=\frac{e^{-iT/r}-(1-iT/r)}{\sqrt{(\cos(T/r)-1)^2+(\sin(T/r)-T/r)^2}}.
\end{equation}
Using these, we have
\begin{align}
    \frac{e^{-i\ketbra{\psi}T/r}}{C_{\rm pure}}&=p_0^{\rm (pure)}\left(\cos(\theta_{0,|T|/r})\bm{1}-i{{\rm sgn}(T)}\sin(\theta_{0,|T|/r})\ketbra{\psi}\right)+p_1^{\rm (pure)}e^{i\phi(T/r)}\ketbra{\psi}
\end{align}
and consequently, for any $A$ 
\begin{align}
    &e^{-i\ketbra{\psi}T/r}A e^{i\ketbra{\psi}T/r}\notag\\
    &=[C_{\rm pure}]^2\sum_{l,l'}p^{(\rm pure)}_lp^{(\rm pure)}_{l'}\left({(e^{i\phi' l}\bm{1})\bullet(e^{-i\phi' l'}\bm{1})}\circ\Upsilon_{\delta_{l1}\pi/2+\delta_{l0}{{\rm sgn}(T)}\theta_{0,|T|/r},1}\circ\Upsilon^{\S}_{\delta_{l'1}\pi/2+\delta_{l'0}{{\rm sgn}(T)}\theta_{0,|T|/r},1}\right)(A)\notag\\
    &\equiv [C_{\rm pure}]^2\sum_{l,l'}p^{(\rm pure)}_lp^{(\rm pure)}_{l'}\Gamma_{l,l'}^{(\rm pure)}(A).
\end{align}
Here, the superoperators $\Upsilon_{\delta,1}$ and $\Upsilon_{\delta,1}^\S$ are defined as Eqs.~\eqref{eq:basic_op1_main_thm} and~\eqref{eq:basic_op2_main_thm}, respectively. Also, $\delta_{l1}$ and $\delta_{l0}$ denote the Kronecker delta.
The superoperator $(e^{i\phi' l}\bm{1})\bullet(e^{-i\phi' l'}\bm{1})$ acts as $A\mapsto e^{i\phi' l}Ae^{-i\phi' l'}$, where $\phi':=\phi+\pi/2$.
By repeating this, we arrive at 
\begin{align}
    e^{-i\ketbra{\psi}T}A e^{i\ketbra{\psi}T}=[C_{\rm pure}]^{2r}\sum_{l_1,l_1',...,l_r,l_r'}p^{(\rm pure)}_{l_1}p^{(\rm pure)}_{l'_1}\cdots p^{(\rm pure)}_{l_r}p^{(\rm pure)}_{l'_r}\left(\Gamma_{l_r,l'_r}^{(\rm pure)}\circ\cdots \circ\Gamma_{l_1,l'_1}^{(\rm pure)}\right)(A).
\end{align}

From the LCS theory, we can derive a quantum channel $\widetilde{\Phi}^{(\rm pure)}_{l_1,l_1',...,l_r,l_r'}$ for the superoperator $\Gamma_{l_r,l'_r}^{(\rm pure)}\circ\cdots \circ\Gamma_{l_1,l'_1}^{(\rm pure)}$ by replacing each superoperator according to Table~\ref{tab:translationlist}.
Each $\widetilde{\Phi}^{(\rm pure)}_{l_1,l_1',...,l_r,l_r'}$ has $2r$ copies of $\rho$.
Then, we conclude that the expectation of the random map
\begin{equation}
    \hat{\Phi}_{{\rm pure},r}(\bullet):={\rm tr}_{\rm anc}\left[(X_{\rm anc}\otimes \bm{1})\widetilde{\Phi}^{(\rm pure)}_{l_1,l_1',...,l_r,l_r'}\left(\ketbra{+}_{\rm anc}\otimes \bullet\right)\right]
\end{equation}
over all possible $l_1,l_1',...,l_r,l_r'$ exactly recovers the target unitary channel after multiplying the factor $C_{\rm pure}^{2r}$.
Because
\begin{equation}
    C_{\rm pure}(T/r)\leq 1+(T/r)^2 \leq e^{T^2/r^2}
\end{equation}
holds, we can bound $C_{\rm pure}^{2r}\leq e^{2T^2/r}$.
\end{proof}

\clearpage
\section{Optimal DME-based quantum algorithm for property estimation}
\label{sec:nearlyopt_qalg}

Let us consider a general problem to obtain the expectation value
\begin{equation}
\label{eq:target_exp_for_UM}
    f_M(\rho)={\rm tr}\left[OU_M\sigma U_M^\dagger\right]
\end{equation}
for a given initial state $\sigma$ and a given observable $O$.
Here, $U_M$ is an arbitrary unitary circuit that includes 
$M$ uses of a unitary gate in the form of $e^{-i\rho T}\otimes \bm{1}$ {or $e^{-i\rho T}$. Hence, the dimension of $\sigma$ and $O$ does not necessarily match that of $\rho$. 
Also, $\sigma$ may be a function of $\rho$ such as $\sigma=\rho\otimes \ketbra{0}$. 
The evolution time $T$ may be different for each gate. 
Without loss of generality, we can restrict $|T|\leq 1$, because $M$ can be taken as an arbitrary positive integer if one wants to simulate a long time evolution. 

{Our central question is the same as before;} when $\rho$ is unknown but we are accessible to its copies, how many copies are needed to estimate the expectation value $f_M(\rho)$ for a given additive error with a high probability? 
Importantly, as will be discussed in Section~\ref{supple_sec:applications}, the problem of estimating $f_M(\rho)$ contains a wide range of applications such as Hamiltonian simulation, emulation of unknown dynamics, quantum entropy estimation, principal component analysis, and linear system solver with precomputed resource states.

\subsection{Procedure and performance guarantee}
\label{subsection S3.1}

We can construct a quantum algorithm for efficiently computing $f_M(\rho)$, based on Theorem~\ref{thm:main_supple}, as follows. 
Note that the following algorithm also works when $e^{-i\rho T_1},e^{-i\rho T_2},...,e^{-i\rho T_M}$ have distinct generators $\{\rho_m\}$ as $e^{-i\rho_1 T_1},e^{-i\rho_2 T_2},...,e^{-i\rho_M T_M}$ in the target expectation value. 
The unitary channel $\mathcal{U}[U_M]=U_M\bullet U_M^\dagger$ can be written as
\begin{equation}
    \mathcal{U}[U_M]=\mathcal{V}_{M+1}\circ (\mathcal{U}[e^{-i\rho T_M}]\otimes \mathcal{I}_M)\circ \mathcal{V}_{M}\circ\cdots \circ (\mathcal{U}[e^{-i\rho T_1}]\otimes \mathcal{I}_1)\circ \mathcal{V}_1
\end{equation}
for some unitary channels $\{\mathcal{V}_m\}_{m=1}^{M+1}$ (which are independent from $\rho$) without loss of generality.
From Theorem~\ref{thm:main_supple}, we can construct random modified-quantum maps $\hat{\Phi}^{(m)}_{r_m}$ satisfying  
\begin{equation}\label{eq:error_seg}
    \|[C^{(m)}]^{2r_m}\mathbb{E}[\hat{\Phi}^{(m)}_{r_m}]-\mathcal{U}[e^{-i\rho T_m}]\|_{\diamond}=\mathcal{O} (\varepsilon/M).
\end{equation}
Then, $\mathcal{U}[U_M]=U_M\bullet U_M^\dagger$ can be well approximated by
\begin{equation}
    {\prod_{m}[C^{(m)}]^{2r_m}}\times \mathcal{V}_{M+1}\circ (\mathbb{E}[\hat{\Phi}^{(M)}_{r_M}]\otimes \mathcal{I}_M)\circ \mathcal{V}_{M}\circ\cdots \circ (\mathbb{E}[\hat{\Phi}^{(1)}_{r_1}]\otimes \mathcal{I}_1)\circ \mathcal{V}_1.
\end{equation}
Applying a randomly-chosen quantum map (including Pauli X measurements)
\begin{equation}\label{eq:sampled_circuit_for_UM}
    \mathcal{V}_{M+1}\circ (\hat{\Phi}^{(M)}_{r_M}\otimes \mathcal{I}_M)\circ \mathcal{V}_{M}\circ\cdots \circ (\hat{\Phi}^{(1)}_{r_1}\otimes \mathcal{I}_1)\circ \mathcal{V}_1
\end{equation}
to the initial state $\sigma$ and measuring the observable $O$, we obtain measurement outcomes. 
These outcomes constitute a sample of an estimator $\hat{g}$ that has a finite variance (at most $\|O\|^2$) and the following mean:
\begin{equation}\label{eq:final_exp_before_rescaling}
    {\rm tr}\left[O\left(\mathcal{V}_{M+1}\circ (\mathbb{E}[\hat{\Phi}^{(M)}_{r_M}]\otimes \mathcal{I}_M)\circ \mathcal{V}_{M}\circ\cdots \circ (\mathbb{E}[\hat{\Phi}^{(1)}_{r_1}]\otimes \mathcal{I}_1)\circ \mathcal{V}_1\right)(\sigma)\right].
\end{equation}
Then, multiplying this estimate by the rescaling factor $\prod_m [C^{(m)}]^{2r_m}$, we can estimate the target expectation value $f_M(\rho)$ of Eq.~\eqref{eq:target_exp_for_UM}.

Now, we evaluate the performance of the above quantum algorithm to generate $\hat{g}$.
The resulting estimator $\hat{g}':=(\prod_m [C^{(m)}]^{2r_m})\hat{g}$ has the bias (i.e., the difference between the mean and the target value) of $\mathcal{O}(\varepsilon)$, which can be confirmed by 
\begin{align}
        &\left|{\rm tr}[OU_M\sigma U_M^\dagger]-\mathbb{E}[\hat{g}']\right|\notag\\
        &=\left|{\rm tr}[OU_M\sigma U_M^\dagger]{-\prod_{m}[C^{(m)}]^{2r_m}\rm tr}\left[O\left(\mathcal{V}_{M+1}\circ (\mathbb{E}[\hat{\Phi}^{(M)}_{r_M}]\otimes \mathcal{I}_M)\circ\cdots \circ (\mathbb{E}[\hat{\Phi}^{(1)}_{r_1}]\otimes \mathcal{I}_1)\circ \mathcal{V}_1\right)(\sigma)\right]\right|\notag\\
        &\leq \|O\|\left\|\mathcal{U}[U_M]-\prod_{m}[C^{(m)}]^{2r_m}\mathcal{V}_{M+1}\circ (\mathbb{E}[\hat{\Phi}^{(M)}_{r_M}]\otimes \mathcal{I}_M)\circ\cdots \circ (\mathbb{E}[\hat{\Phi}^{(1)}_{r_1}]\otimes \mathcal{I}_1)\circ \mathcal{V}_1\right\|_{\diamond}\notag\\
        &\leq \mathcal{O}(\|O\|\varepsilon),
\end{align}
where the last line follows from Eqs.~\eqref{diamond_norm_of_xy} and \eqref{eq:error_seg} (and the summary shown in Fig.~\ref{fig:diamond_norm_of_xy}). 
On the other hand, the variance of the estimator is amplified by the factor of $\prod_{m}[C^{(m)}]^{2r_m}$, but this amplification can be upper bounded by a given constant $\gamma>1$. 
This is because, by setting $r_m=2M\ln^{-1}\gamma$, we have 
\begin{equation}
    \prod_{m}[C^{(m)}]^{2r_m}\leq e^{2\sum_{m} T_m^2/r_m}\leq e^{2\sum_{m} 1/r_m}
     =\gamma, 
\end{equation}
where $|T_m|\leq 1~\forall m$ is used. 
For instance, if one allows for doubling the variance, then $\gamma=\sqrt{2}$ and $\ln^{-1}\gamma \simeq 2.89$.
Therefore, generating $\mathcal{O}(1/\varepsilon^2)$ independent samples of the estimator $\hat{g}$ and taking the sample-mean, we can estimate the target expectation value $f_M(\rho)$ within the additive error $\mathcal{O}(\varepsilon)$ with a high probability. 
Finally, to generate a single sample of $\hat{g}$, each of the quantum circuits of Eq.~\eqref{eq:sampled_circuit_for_UM} uses at most
\begin{equation}\label{eq:general_alg_numcopies}
    \sum_{m=1}^M \left[2r_m+r_m\cdot  \mathcal{O}\left(\frac{\log(r_m M/\varepsilon)}{\log\log(r_m M/\varepsilon)}\right)\right]=\mathcal{O}\left(M^2\frac{\log(M/\varepsilon)}{\log\log(M/\varepsilon)}\right)
\end{equation}
copies of the unknown quantum state $\rho$ due to the choice of $r_m=2M\ln^{-1}\gamma=\mathcal{O}(M)$. 
The gate complexity except for the interleaving maps $\{\mathcal{V}_m\}$ is, again from the discussion around Table~\ref{tab:resources}, given by
\begin{equation}
    \mathcal{O}\left(M^2\frac{\log(M/\varepsilon)}{\log\log(M/\varepsilon)}\log (d)\right)
\end{equation}
where $d$ is the dimension of the target state $\rho$.

We remark that if the target state $\rho$ is pure, we can use the random modified-quantum map $\hat{\Phi}_{\rm pure,r}$ instead of the general version $\hat{\Phi}_{r}$.
Since the random map $\hat{\Phi}_{\rm pure,r}$ exactly recovers the target unitary channel for $e^{-i\rho T}$ in expectation, we conclude that the number of copies required for each circuit is at most 
\begin{equation}\label{eq:special_alg_numcopies}
    \sum_{m=1}^M 2r_m=\mathcal{O}\left(M^2\right)
\end{equation}
and the gate complexity is
\begin{equation}
    \mathcal{O}\left(M^2\log (d)\right).
\end{equation}
Similar complexities can be derived when we use our algorithm to estimate ${\rm tr}[OU_M\sigma U_M^\dagger]$ for $U_M$ composed of time propagators $\{e^{-i\rho_m T_m}\}$ with distinct generators $\{\rho_m\}$.
In the whole procedure of the quantum algorithm, the total number of copies is calculated as the product of the final measurement cost $\mathcal{O}(1/\varepsilon^2)$ and Eq.~\eqref{eq:general_alg_numcopies} (or Eq.~\eqref{eq:special_alg_numcopies}).

\subsection{Lower bound}\label{supple_sec:alg_lower_bound}

To show the optimality (or near-optimality) of our quantum algorithm described in Subsection~\ref{subsection S3.1}, we now prove \green{Theorem~\ref{main_thm:lower_bound_M_alg}} in the main text. 
\\
\\
\noindent
{\bf Theorem~3.}
{\it 
Let $M$ be a positive integer, and let $f_M{(\rho)}$ be a function from a state $\rho$ to the expectation value in the form of Eq.~\eqref{eq_main:target_fM} {(or Eq.~\eqref{eq:target_exp_for_UM})} having $M$ density matrix exponential operations with a finite time.
Suppose that $\mathcal{A}$ is any quantum algorithm that, given $k$ copies of an unknown state $\rho$ and a description of $f_M$, outputs an estimate of the expectation value $f_M(\rho)$ within a constant additive error with a high probability (at least 2/3).
Then, for any such algorithm $\mathcal{A}$, there exists a function $f_M{(\rho)}$ such that $\mathcal{A}$ requires $k=\Omega(M^2)$ copies of $\rho$.
}

\begin{proof}
    To prove this theorem, we consider the following \textit{sample-based Grover's search problem}~\cite{Kimmel2017-hv}.
    Let $\mathcal{L}$ be a target subspace of a Hilbert space, and $U_\mathcal{L}$ is the following unitary: 
    \begin{equation}
        U_{\mathcal{L}}\ket{\psi}\ket{0}=\begin{dcases}
        \ket{\psi}\ket{1}&\mbox{if}~\ket{\psi}\in \mathcal{L},\\
        \ket{\psi}\ket{0}&\mbox{if}~\ket{\psi}\perp \mathcal{L}.
        \end{dcases}
    \end{equation}
    Let $\ket{s}$ be an unknown quantum state ensured that (i) $\lambda\geq w$ or (ii) $\lambda=0$ holds for $w>0$ and
    \begin{equation}
        \lambda:=\left|(\bm{1}\otimes \bra{1})U_{\mathcal{L}}\ket{s}\otimes \ket{0}\right|^2.
    \end{equation}
    The goal is then to determine whether (i) or (ii) by using the copies of $\ket{s}$ and the unlimited access to $U_{\mathcal{L}}$ and its inverse; this problem is called the sample-based Grover's search problem. 
    Ref.~\cite{Kimmel2017-hv} proved that to solve this problem with high probability, $\Omega(1/w)$ copies of $\ket{s}$ are required.
    
    Now, we explicitly construct a form of $f_M(\rho)$ that can solve the sample-based Grover's search problem.
    Note that $U_{\mathcal{L}}\ket{s}\ket{0}$ can always be represented as 
    \begin{equation}
        U_{\mathcal{L}}\ket{s}\ket{0}=\sin\frac{\phi}{2}\ket{\psi_1}\ket{1}+\cos\frac{\phi}{2}\ket{\psi_0}\ket{0},
    \end{equation}
    using some quantum states $\ket{\psi_0}$ and $\ket{\psi_1}$, where $\phi=2\arcsin(\sqrt{\lambda})\in [0,\pi]$. 
    In the two-dimensional subspace spanned by $\ket{\psi_0}\ket{0}$ and $\ket{\psi_1}\ket{1}$, let $\overline{X},\overline{Z}$ be Pauli X, Z gates in the two basis.
    Then, in this subspace, we can identify the Pauli rotations $e^{-i \overline{Z}\beta/2}$ and $e^{-i(\cos\phi \overline{Z}+\sin\phi \overline{X})\alpha/2}$ with the following operations, respectively:
    \begin{equation}
        e^{-i(\beta/2)\bm{1}\otimes Z}=\bm{1}\otimes e^{-i(\beta/2)Z},~~~
        U_{\mathcal{L}}\cdot e^{-i\alpha\ket{s}\ket{0}\bra{s}\bra{0}}e^{i(\alpha/2)\bm{1}}\cdot U_{\mathcal{L}}^\dagger.
    \end{equation}
    By using these operations, we can generate a quantum state 
    \begin{equation}
        c_l^{(1)}(\lambda)\ket{\psi_1}\ket{1}+c_l^{(0)}(\lambda)\ket{\psi_0}\ket{0}
    \end{equation}
    such that if $\lambda\geq w$, 
    \begin{equation}
    \label{approx error Ther 3}
        1-(\delta')^2\leq |c_l^{(1)}(\lambda)|^2\leq 1
    \end{equation}
    holds and if $\lambda=0$, then $|c_l^{(1)}(\lambda)|=0$.
    More specifically, Ref.~\cite{yoder2014fixed} proves that for a positive integer $l\geq \log(2/\delta')/\sqrt{w}$, one can choose finite rotational angles $\{\alpha_j,\beta_j\}_{j=1}^l$ such that 
    \begin{equation}\label{eq:q_state_thm3}
        G(\alpha_l,\beta_l)\cdots G(\alpha_1,\beta_1)U_{\mathcal{L}}\ket{s}\ket{0}=c_l^{(1)}(\lambda)\ket{\psi_1}\ket{1}+c_l^{(0)}(\lambda)\ket{\psi_0}\ket{0}
    \end{equation}
    holds, where
    \begin{equation}
        G(\alpha,\beta)=U_{\mathcal{L}}\cdot e^{-i\alpha\ket{s}\ket{0}\bra{s}\bra{0}}e^{i(\alpha/2)\bm{1}}\cdot U_{\mathcal{L}}^\dagger \cdot e^{-i(\beta/2)\bm{1}\otimes Z}.
    \end{equation}

    We take $f_M(\rho)$ as the expectation value $\langle \bm{1}\otimes Z\rangle $ of $\bm{1}\otimes Z$ with respect to the quantum state Eq.~\eqref{eq:q_state_thm3}; $M$ is chosen to be $M=\Theta( \log(1/\delta')/\sqrt{w})$ to ensure the error Eq.~\eqref{approx error Ther 3}. 
    Now, we can take any algorithm specified in the assumption to obtain the expectation value $\langle \bm{1}\otimes Z\rangle $ within a constant additive error using $k=k(M)$ copies of $\ket{s}$.
    Then, one can also determine (i) $\lambda\geq w$ or (ii) $\lambda=0$ by a decision rule; e.g., if $\langle \bm{1}\otimes Z\rangle $ is bigger than a threshold, then return (ii).
    Note that $\delta'$ can be taken as a constant.
    This indicates that any such algorithm $\mathcal{A}$ can solve the sample-based Grover's search problem, with the use of $k(M)$ copies of $\ket{s}$ for $M\propto 1/\sqrt{w}$.
    Since $\Omega(1/w)=\Omega(M^2)$ copies of $\ket{s}$ are required for any procedure to solve the sample-based Grover's search problem, we conclude that $k(M)=\Omega(M^2)$.
\end{proof}

\clearpage
\section{Non-physical operation in our procedure}
\label{supple_sec:non_physical_ope}

Non-physical operations (i.e., non-CPTP maps) play a key role in our procedure. 
To clarify the mechanism of those operations, we first recall the conventional procedure (i.e., the LMR protocol)~\cite{Lloyd2014-nn} for implementing $e^{-i\rho T}$ from multiple copies of an unknown state $\rho$.
As described in Section~\ref{sec:LMR}, the original protocol uses 
a CPTP map $\Lambda_{\delta T}$ defined as
\begin{equation}
    \Lambda_{\delta T} (\bullet)={\rm tr}_{2}\left[e^{-i\delta T S}(\bullet \otimes \rho )e^{i\delta T S}\right],
\end{equation}
where ${\rm tr}_{2}$ denotes the partial trace over the second register and $S$ is the SWAP gate.
We note that this operation uses a unitary channel (thus, a physical operation) between a target system and a single copy of $\rho$.
Then, by using this operation repeatedly, we can implement the unitary $e^{-i\rho T}$ up to the error $\varepsilon$ in the diamond norm, using $N=\mathcal{O}(T^2/\varepsilon)$ copies of an unknown state $\rho$; also see Section~\ref{sec:LMR}. 
The number of copies matches the lower bound $\Omega(T^2/\varepsilon)$, which was proven in Refs.~\cite{Kimmel2017-hv,Go2024-pb} for the case where we can use any quantum process (i.e., any CPTP map) between a target system and copies of $\rho$. 
The above procedure is optimal in this sense.

We here compare the above LMR protocol with our result in Theorem~\ref{thm:main_supple}, which says that using $\mathcal{O}(T^2\log(T/\varepsilon))$ copies of the unknown state $\rho$, we can simulate the unitary channel $\mathcal{U}[e^{-i\rho T}]$ up to the error $\varepsilon$ in the diamond norm. 
In particular, when $\rho$ is a pure state, Theorem~\ref{thm:DME_pure_state} says that $\mathcal{O}(1)$ copies are sufficient for the perfect simulation. 
Here, in those theorems, we take $r=\mathcal{O}(T^2)$ for simplicity, which ensures the constant sampling (or measurement) overhead.
These results seem to be inconsistent with the lower bound, but the important difference is in allowable operations as illustrated in Fig.~\ref{fig:venn_diag}. 
That is, our method cannot be simulated with the use of only CPTP maps.

\begin{figure}[tb]
 \centering
 \includegraphics[scale=0.45]{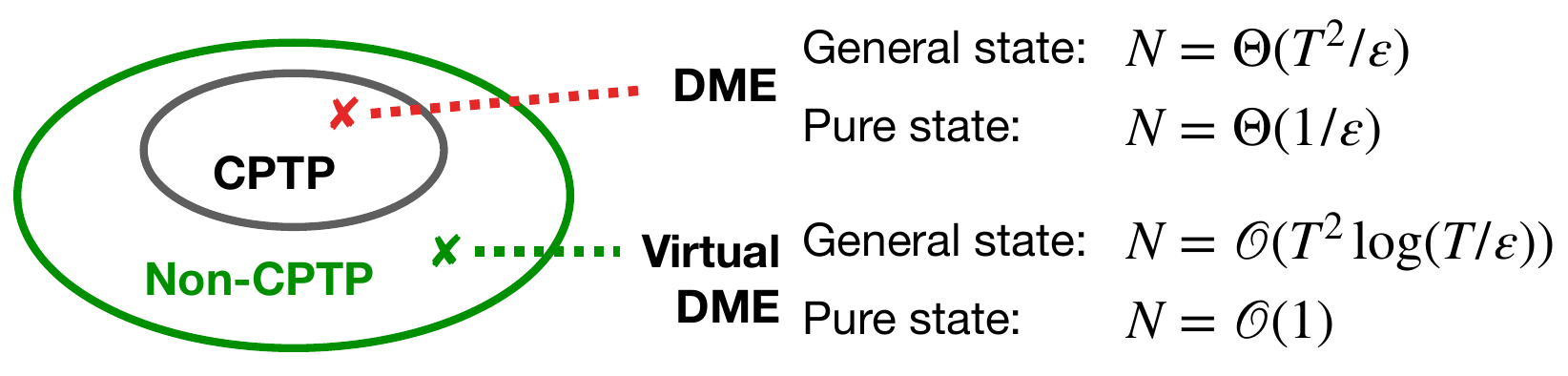}
 \caption{Hierarchy of superoperators for the conventional DME (i.e., the LMR protocol) and the virtual DME.
 Technically, the virtual DME uses Hermitian-preserving maps, which are in a broader set including all of the completely positive and trace-preserving (CPTP) maps.}
 \label{fig:venn_diag}
\end{figure}

By comparing our procedure and the conventional one, we find that our procedure uses (i) random unitary gates and (ii) Pauli X measurement on the single-ancilla system followed by classical post-processing, which are not used in the previous method. 
More specifically, our random map is given by Eq.~\eqref{eq:sample_total_expression}
\begin{equation}
    \hat{\Phi}_r(\bullet)=
    {\rm tr}_{\rm anc}\left[(X_{\rm anc}\otimes \bm{1})\widetilde{\Phi}_{\bm{l}}\left(\ketbra{+}_{\rm anc}\otimes\bullet\right)\right],
\end{equation}
where $\widetilde{\Phi}_{\bm{l}}$ is a CPTP map Eq.~\eqref{eq:sample_each_expression} 
with a set of indices $\bm{l}=(l_1,l_1',...,l_r,l_r')$, and the mean of the random map is given by Eq.~\eqref{eq:mean_modified channel_explicit}:
\begin{equation}
   \mathbb{E}[\hat{\Phi}_r](\bullet)=
    {\rm tr}_{\rm anc}\left[(X_{\rm anc}\otimes \bm{1})\left(\sum_{\bm{l}}p_{\bm{l}}\widetilde{\Phi}_{\bm{l}}\right)\left(\ketbra{+}_{\rm anc}\otimes\bullet\right)\right]
\end{equation}
for a probability distribution $p_{\bm{l}}$. 
Here, $\sum_{\bm{l}}p_{\bm{l}}\widetilde{\Phi}_{\bm{l}}$ is a CPTP map, so the random operation (i) is rewritten as a deterministic physical process by additionally using ancilla qubits (similar to the purification of a mixed state $\mathbb{E}_\psi[\ketbra{\psi}]$), beyond the single-ancilla qubit required in Theorem~\ref{thm:main_supple}.
That is, we can eliminate the process of taking expectation in principle. This implies that, while the randomness has a practically important role to achieve the minimal-size ancilla system, using the randomization alone does not go beyond the established lower bound $\Omega(T^2/\varepsilon)$. 
The essential operation for the exponential improvement should be (ii), which forms a non-physical process as follows.
That is, the mean of the random map, $\mathbb{E}[\hat{\Phi}_r](\bullet)$, can be represented as
\begin{align}
    &{\rm tr}_{\rm anc}\left[(\ketbra{+}_{\rm anc}\otimes \bm{1})\left(\sum_{\bm{l}}p_{\bm{l}}\widetilde{\Phi}_{\bm{l}}\right)\left(\ketbra{+}_{\rm anc}\otimes\bullet\right)\right]-{\rm tr}_{\rm anc}\left[(\ket{-}\bra{-}_{\rm anc}\otimes \bm{1})\left(\sum_{\bm{l}}p_{\bm{l}}\widetilde{\Phi}_{\bm{l}}\right)\left(\ketbra{+}_{\rm anc}\otimes\bullet\right)\right].
\end{align}
The negative sign in the second term indicates that the entire operation is not CPTP. 
If the negative sign is replaced with  $+1$, the resulting operation goes back to a CPTP process (the Pauli $X$ measurement process).
Also, the multiplication of the prefactor $C^{2r}$ (in Eq.~\eqref{eq:diamond_distance_alpha}) to match the mean of the random map with the target unitary channel breaks the trace-preserving property.
Therefore, our method essentially requires the use of a non-CPTP map.

Such non-CPTP maps can be used for the property estimation of quantum systems with the help of classical post-processing; see Section~\ref{sec:S23} for the concrete procedure in our cases.
Actually, the usefulness of such non-CPTP maps is also confirmed in recent developments for quantum information processing such as quantum error mitigation~\cite{PhysRevLett.119.180509,PhysRevX.8.031027,Kim_2023,koczor2021-esd,Huggins2021-vd,PhysRevApplied.23.054021,RevModPhys.95.045005}, quantum circuit knitting~\cite{PhysRevLett.125.150504,Mitarai_2021,10.1109/TIT.2023.3310797,PRXQuantum.5.040308,PRXQuantum.6.010316,Harada2025densitymatrix,yamamoto2024virtualentanglementpurificationnoisy}, virtual cooling~\cite{PhysRevX.9.031013}, resource distillation~\cite{PhysRevLett.132.050203,PhysRevA.109.022403}, quantum broadcasting~\cite{PhysRevLett.132.110203}, and quantum error correction~\cite{cao2023quantummapscptphptp}. 
Usually, the cost to use non-CPTP maps is reflected in the increase of the measurement cost compared to the standard scaling $\mathcal{O}(1/\varepsilon^2)$.
In our case, the measurement overhead is given by $C^{2r}$, leading to the overall measurement cost $\mathcal{O}(C^{4r}/\varepsilon^2)$.
However, as we already clarified, the measurement overhead $C^{2r}\leq e^{2T^2/r}$ can be arbitrarily close to 1 by taking the parameter $r=\mathcal{O}(T^2)$ which determines the circuit depth. 
That is, we have an interpretation that our method can convert the measurement overhead into the circuit depth (and the number of copies per circuit); in other words, by appropriately choosing the circuit depth, the measurement overhead can be ignored.

\clearpage
\section{Applications}\label{supple_sec:applications}

The proposed algorithm builds an essential foundation in general DME, by reducing the 
complexity in the necessary number of copies. 
Due to the ubiquity of DME, our algorithm leads to a significant reduction of complexity in various important applications. 
Here we show some examples, with information on how the DME is used in each application. 
\begin{itemize}

\item 
Sample-based Hamiltonian simulation: The goal is to construct a simulator for $e^{-i H T}$, where $H$ can be unknown and its copies are given to us. 
This problem is directly formulated using DME, assuming we can access the copies of $\rho \propto H+cI$ with $c$ a constant. 

\item 
Unknown unitary emulation: 
Given a set of input and output states of an unknown unitary $U$, i.e., $\{\ket{\phi_i^{\rm in}}, \ket{\phi_i^{\rm out}}\}$, the task is to (approximately) produce $U\ket{\psi}$ for an unknown new input state $\ket{\psi}$. 
Such direct use of quantum states, specifically the use of DMEs $\{e^{-i \pi\ketbra{\phi_i^{\rm in}}},~e^{-i \pi\ketbra{\phi_i^{\rm out}}} \}$, enables efficient emulation of $U$.

\item 
Estimation of nonlinear functions such as entropy:
If the DME $e^{-i\rho t}$ is available, it can be used to realize the block encoding of $\rho$, which can then be transformed to ${\mathcal G}(\rho)$ with $f$ a polynomial function. 
This allows calculating the entropy $-{\rm tr}[\rho \ln \rho]$ via setting ${\mathcal G}(\rho)\approx \ln(\rho)$. 

\item
Precomputation for several algorithms including a linear system solver: 
Consider the quantum solver for a linear system of equations $A\ket{x}=\ket{b}$; if the target $\ket{b}$ is complicated but can be prepared beforehand, its DME $e^{-i\pi\ketbra{b}}$ can be used to make the solver exponentially efficient. 

\item 
Quantum principal component analysis: 
For a noisy state $\rho=(1-\lambda)\ketbra{\psi} + \lambda \rho_{\rm err}$, we are interested in extracting the principal component $\ket{\psi}$. 
This can be efficiently executed by constructing the projection filter onto $\ket{\psi}$ from the DME $e^{-i\rho}$. 

\end{itemize}
The last application, quantum principal component analysis, will be described in detail together with numerical simulation, in the next section.


\subsection{Hamiltonian simulation based on unknown $\rho$}

The first yet a most primitive application of DME is Hamiltonian simulation. 
Rather than the conventional setting where a classical description of a Hamiltonian is given to us, we here assume that a given density matrix encodes a Hamiltonian in the form 
\begin{equation}
    \rho =\frac{H+c\bm{1}}{{\rm tr}[H+c\bm{1}]},
\end{equation}
where $c$ is a real value chosen so that $\rho$ is non-negative. 
Then, by consuming multiple copies of $\rho$, we aim to implement the unitary channel of $e^{-i\rho T}$. 
A notable point of this setting is that $\rho$, and thus $H$, is unknown to us.
This problem is called the \textit{sample-based Hamiltonian simulation problem} and is well studied by Refs.~\cite{Lloyd2014-nn,Kimmel2017-hv,Go2024-pb} with particular interest in the necessary number of copies of $\rho$ to achieve a given precision for simulation.

Beyond the original problem, we here consider estimating properties of a quantum system evolved under the time propagator. 
That is, if one can implement the time propagator $e^{-i\rho T}$ in the accuracy $\varepsilon$, then the expectation value for a given observable $O$,
\begin{equation}
    {\rm tr}[Oe^{-i\rho T}\sigma e^{i\rho T}],
\end{equation}
can be measured within $\mathcal{O}(\|O\|\varepsilon)$ error with a high probability, by $\mathcal{O}(1/\varepsilon^2)$ times measurement of $O$ with the state $e^{-i\rho T} \sigma e^{i\rho T}$. 
Since the previous method~\cite{Lloyd2014-nn} sequentially consumes $\mathcal{O}(T^2/\varepsilon)$ copies of $\rho$ to implement $e^{-i\rho T}$, the corresponding circuit depth is polynomial in $1/\varepsilon$, leading to a deep circuit. 
The proposed virtual DME method resolves this issue; that is,
recalling that the necessary number of copies for virtual DME is given by Eq.~\eqref{eq:simple_dme_copycost}, the circuit depth for this task can be exponentially reduced to $\mathcal{O}(T^2\log(T/\varepsilon))$, while keeping the same scaling of measurement cost $\mathcal{O}(1/\varepsilon^2)$. 
The explicit construction of the algorithm is provided in Section~\ref{sec:S23}.
The crucial point is that the property estimation 
does not necessarily require the direct implementation of the time propagator $e^{-i\rho T}$.
Actually, our procedure simulates the time propagator effectively by using the low-depth circuit with a quantum measurement followed by classical post-processing.

\subsection{Emulation of an unknown unitary from input-output data}
\label{apsec:uqe}

Next, let us consider the task of learning an unknown unitary operation~\cite{PhysRevA.81.032324,kiani2020learningunitariesgradientdescent,PhysRevApplied.19.034017,PhysRevA.109.022429,PhysRevLett.112.240501,PhysRevA.106.062434,PhysRevB.96.020408,marvian2024universalquantumemulator}.
The task of learning an unknown unitary is defined in two ways: (i) learning a unitary to physically implement it on a quantum device or (ii) learning specific features of a unitary for solving particular problems such as estimating expectation values or reversing scrambling processes~\cite{PhysRevA.106.062434,PhysRevA.109.022429}.
In our application, we focus on the latter. 

The problem is as follows. 
Let $U$ be an unknown unitary operator acting on a $d$-dimensional Hilbert space, and suppose we are given access to multiple copies of unknown input-output states of $U$. 
The goal is to emulate the action of $U$ on an unknown new input state $\ket{\psi}$, to generate the state $U\ket{\psi}$.
A straightforward approach to this problem is based on quantum process tomography~\cite{Chuang01111997,PhysRevLett.86.4195,PhysRevLett.78.390}, where one first reconstructs a classical description of $U$ from input-output samples and then applies it to produce $U\ket{\psi}$. 
However, both the run-time and the sample complexity (i.e., the total number of copies required to run the algorithm) scale at least linearly with the system dimension $d$~\cite{haah2023query}.
In contrast, the {\it Universal Quantum Emulator} (UQE) proposed in \cite{marvian2024universalquantumemulator} enables one to solve the problem by learning the input-output relation, rather than reconstructing the unitary itself. 
Remarkably, the depth\footnote{The depth of UQE corresponds to the run-time for reproducing $U\ket{\psi}$.}
of UQE scales logarithmically with $d$ and the number of copies in each circuit (i.e., sample complexity of UQE) is independent of $d$; see Lemma~\ref{lemma:uqe_with_simulated_reflection}. 
Due to this favorable scaling in $d$, UQE is not only considered as a powerful quantum analytical tool for quantum data alongside quantum simulation~\cite{Biamonte_2017,doi:10.1126/science.273.5278.1073}, but also it has broad application in privacy-preserving secure quantum computation (i.e., secure quantum software)~\cite{5438603,10.5555/2011670.2011674}, quantum cryptanalysis~\cite{Arapinis2021quantumphysical,doosti2021unifiedframeworkquantumunforgeability}, quantum resource theory~\cite{RevModPhys.91.025001,gour2024resourcesquantumworld,doi:10.1142/S0217979213450197}, and quantum reference frame~\cite{RevModPhys.79.555,Gour_2008,PhysRevA.78.022304,popescu2018quantum,PhysRevA.90.062110}. 
However, UQE still suffers from a linear dependence on the inverse of the target approximation error $\varepsilon$; that is, both depth and sample complexity scale as $\mathcal{O}(\varepsilon^{-1})$, which poses a major bottleneck for practical implementations.

\begin{figure*}[t]
\centering
\begin{center}
 \includegraphics[width=160mm]{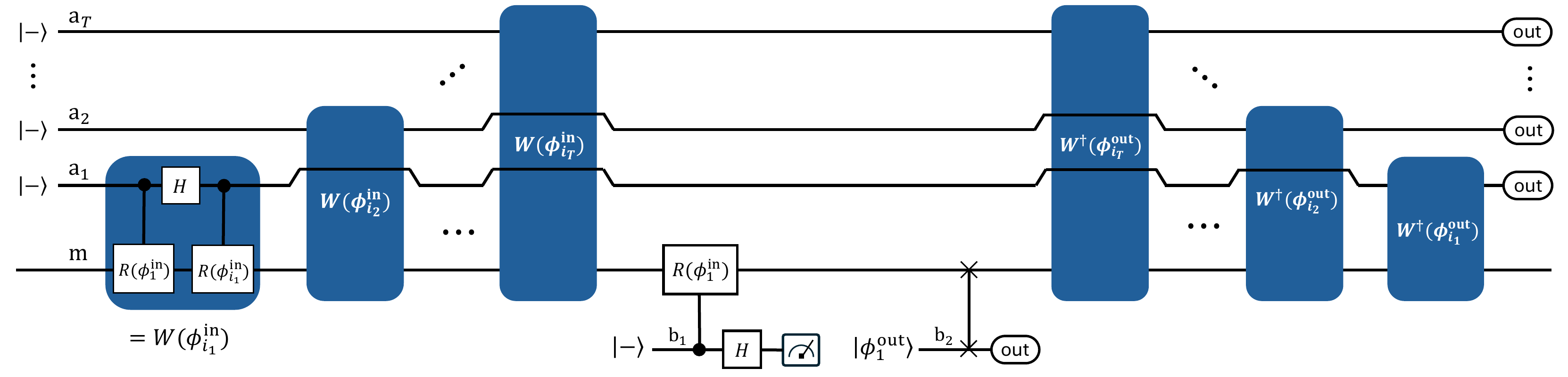}  
\end{center}
\vspace{-0.6cm}
\caption{Quantum circuit of the universal quantum emulator proposed in Ref.~\cite{marvian2024universalquantumemulator}. 
$R(\phi) :=e^{-i\pi\ket{\phi}\bra{\phi}}$ denotes the reflection about the state $\ket{\phi}$, and $W(\phi)$s in the blue shaded area are unitary operators composed of the controlled-$R(\phi)$s and a Hadamard gate; see Eq.~\eqref{eq:def_of_W_operator} for its detailed definition. 
The indices $\{i_1,i_2,...,i_T\}$ are uniform-randomly selected from $\{2,3,...,K\}$. 
In this algorithm, the input state $\ket{\psi}$ is provided at the left end of the wire labeled $\rm m$, and an approximation of the target state $U\ket{\psi}$ is output at the right end of the same wire.}
\label{Fig:universal_quantum_emulator}
\end{figure*}

In the following, we show that if the task is restricted to property estimation on the emulated state $U\ket{\psi}$, both the depth 
and the number of copies in each circuit can be exponentially improved in terms of $\varepsilon$, while preserving the same scaling with respect to other parameters including $d$. 
Our main idea is to replace the reflection operators simulated via the previous DME technique~\cite{Lloyd2014-nn} in the UQE algorithm with those constructed using our virtual DME method.

To state our result accurately, we give a brief review of the original UQE. 
Let us define
\begin{equation*}
    S_{\rm in} := \left\{ \ket{\phi_i^{\rm in}} :i=1,\cdots,K\right\}~~~\text{and}~~~ S_{\rm out} := \left\{\ket{\phi_i^{\rm out}}:=U\ket{\phi_i^{\rm in}}:i=1,\cdots,K\right\}
\end{equation*}
as the sets of input and output quantum states, respectively. 
Also, let $\mathcal{H}_{\rm in}$ denote the $d_{\rm eff}$-dimensional Hilbert subspace spanned by the input states. 
The UQE takes an arbitrary unknown state $\ket{\psi}\in \mathcal{H}_{\rm in}$ as input 
and finally outputs the state $U\ket{\psi}$ by consuming the copies of unknown input-output quantum states.
This algorithm does not require any prior knowledge of the unknown $U$. 
Yet in order to successfully emulate the action of $U$ on the input subspace $\mathcal{H}_{\rm in}$, it is necessary to impose certain assumptions on the set of input states, $S_{\rm in}$.
Specifically, Ref.~\cite{marvian2024universalquantumemulator} shows that $S_{\rm in}$ should satisfy the following condition:
\begin{equation}\label{eq:necessary_and_sufficient_conditino_for_uqe}
    \mathrm{Alg}_{\mathbb{C}}(S_{\rm in}) = \mathcal{L}(\mathcal{H}_{\rm in})
\end{equation}
where $\mathrm{Alg}_{\mathbb{C}}(S_{\rm in})$ denotes the complex associative algebra generated by $S_{\rm in}$ and $\mathcal{L}(\mathcal{H}_{\rm in})$ denotes the set of all linear operators support on $\mathcal{H}_{\rm in}$. 
As shown in Proposition~1 in \cite{marvian2024universalquantumemulator}, this condition is the necessary and sufficient condition for the set $S_{\rm in}$ to uniquely determine the action of $U$ on $\mathcal{H}_{\rm in}$.

The UQE algorithm uses the quantum circuit in Fig.~\ref{Fig:universal_quantum_emulator} using $\mathcal{O}(T)$ controlled reflection gates for the input and output states, defined as
\begin{equation}
    \mathrm{c}R(\phi):= \ketbra{0} \otimes \bm{1} + \ketbra{1} \otimes e^{-i\pi\ket{\phi}\bra{\phi}},
\end{equation}
where $\phi$ takes $\phi_{i_k}^{\rm in}$ and $\phi_{i_k}^{\rm out}$ with the indices $\{i_1,i_2,...,i_T\}$ uniform-randomly selected from $\{2, 3, ..., K\}$. 
In the original UQE algorithm~\cite{marvian2024universalquantumemulator}, the reflection gates $\{e^{-i\pi\ket{\phi}\bra{\phi}}\}$ are simulated using the DME technique based on the LMR protocol~\cite{Lloyd2014-nn}.
Suppose we can perfectly implement the reflection gates, then the output state of this circuit is proved to be $\varepsilon$-close to $U\ket{\psi}$ if $T=\mathcal{O}((1-\lambda_{\perp})^{-1}\log(d_{\rm eff}/\varepsilon^2))$. 
Here, $\lambda_{\perp}$ is a positive value less than one under the condition that Eq.~\eqref{eq:necessary_and_sufficient_conditino_for_uqe} holds.
By consuming multiple copies of input and output states, we can implement the reflection gates $\{e^{-i\pi\ket{\phi}\bra{\phi}}\}$, but this results in the sample complexity of $\mathcal{O}(T^2\varepsilon^{-1})$ and 
the circuit depth $\mathcal{O}(T^2\varepsilon^{-1}\log d)$ due to the use of conventional DME~\cite{Lloyd2014-nn}. 
When we aim to probe properties of the emulated final state $U\ket{\psi}$, this polynomial dependence can be exponentially reduced by using our virtual DME method to simulate the reflection gates in the circuit diagram Fig.~\ref{Fig:universal_quantum_emulator}. 
More specifically, the controlled-reflection gates $\mathrm{c}R(\phi)$ can be simulated via the DME $e^{-i\rho' \pi}$ of $\rho'=\ket{1}\bra{1} \otimes \ket{\phi}\bra{\phi}$. 
When $\rho'$ is a pure state, our virtual DME protocol in Theorem~\ref{thm:DME_pure_state} enables the simulation of $\mathrm{c}R(\phi)$ using only $\mathcal{O}(r)$ copies of $\rho'$, with a sampling overhead of $\mathcal{O}(e^{1/r})$ that is independent of the target precision $\epsilon$ and $r\in \mathbb{N}$ is a tunable parameter. Since the UQE algorithm involves $2T+1$ controlled-reflection gates, replacing them with those simulated by our pure-state virtual DME leads to a total sampling overhead of $\mathcal{O}(e^{T/r})$ and a circuit depth of $\mathcal{O}(rT)$. By choosing $r=\mathcal{O}(T)$ to keep the total sampling overhead constant, we conclude that UQE can be simulated with a total of $\mathcal{O}(T^2)$ input–output pairs and a circuit depth of $\mathcal{O}(T^2 \log d)$. 
This result can be explicitly stated as follows. 

\begin{thm}\label{thm:universal_quantum_emulator}
    Suppose we have access to copies of input and output states $\{\ket{\phi^{\rm in}_i},\ket{\phi^{\rm out}_i}\}$ in $d$ dimension, where each pair is connected with an unknown unitary $U$ as $\ket{\phi^{\rm out}_i}=U\ket{\phi_i^{\rm in}}$.
    In addition, the input states satisfy the condition for emulation Eq.~\eqref{eq:necessary_and_sufficient_conditino_for_uqe} and form a $d_{\rm eff}$-dimensional subspace $\mathcal{H}_{\rm in}$.
    For a given observable $O$ and a state $\rho$ with support restricted to $\mathcal{H}_{\rm in}$, there exists a quantum algorithm that estimates the expectation value ${\rm tr}[OU\rho U^\dagger]$ within additive error $\varepsilon$ with a high probability, using random circuits with
    \begin{itemize}
        \item $T=\mathcal{O}\left((1-\lambda_{\perp})^{-1}\log(\|O\|^2d_{\rm eff}/\varepsilon^2)\right)$, where $\lambda_{\perp}$ is defined as in Lemma~\ref{lemma:uqe_with_ideal_controlled_reflection},
        \item Total number of input-output state copies: $\mathcal{O}(T^2)$,
        \item Total number of one- and two-qubit elementary gates: $\mathcal{O}(T^2\log d)$.
    \end{itemize}
    The total number of measurements is given by $\mathcal{O}(\|O\|^2/\varepsilon^2)$.
\end{thm}

A more detailed discussion, including a review of Ref.~\cite{marvian2024universalquantumemulator} and the proof of Theorem~\ref{thm:universal_quantum_emulator}, is provided in Section~\ref{apsec:theoretical_detail_of_uqe}.

\subsection{Quantum entropy estimation}
\label{apsec:q entropy}

Probing nonlinear properties of quantum states---such as entropy, fidelity, and various distance metrics---is a fundamental yet challenging task in quantum information theory~\cite{montanaro2013-survey, Bravyi2011-testing, Cerezo2020-variational, Tan2021-variational,gilyen2022-fidelity, Wang2023-fidelity, Wang2024-samplizer}.
Among these properties, we here mainly focus on the estimation of quantum entropies: von Neumann entropy and quantum relative entropy.
Due to the exceptional importance of these quantities, their efficient estimation has been extensively investigated as follows.

To estimate the entropy, one of the following two assumptions on target states is usually employed.
The first assumption is called the sample access model,
where multiple identical copies of an unknown quantum state $\rho$ are given.
In this assumption, a tomographic approach~\cite{acharya2020-entropy} can be used to estimate quantum entropy, though it requires a large number of state copies that scales with the system dimension. 
To deal with this issue, Ref.~\cite{wang2023-entropy} provides a method using the conventional DME procedure that does not depend on the dimension but is dependent on the rank of target states; however, it still uses a deep circuit that scales polynomially with respect to the target accuracy.
The second assumption is the purified query access model, where a more powerful oracle that prepares a purification of the target state is given.
This access model smoothly integrates with the block encoding technique for quantum singular value transformation (QSVT) framework~\cite{gilyen2019-qsvt, gilyen2019-world, Subramanian2021-renyi}, enabling a nonlinear transformation based on QSVT, and results in efficient algorithms for entropy estimation~\cite{gilyen2019-world,gur2021-sublinear,Li2019-query, wang2024-new, Wang2023-phase}.
However, constructing such purification oracles can be nontrivial as they may require large ancilla systems; 
if we cannot access the generation process of quantum states such as algorithms for dilated quantum systems, constructing the purification oracle itself would require
quantum state tomography~\cite{PhysRevLett.105.150401,10.1145/3055399.3055454,7956181,Guţă_2020}.

Here, we take the first assumption, the sample access model to circumvent the nontrivial purification oracles.
Also, we employ the time-evolution $e^{-i\rho T}$, which can be efficiently simulated from identical copies of $\rho$ via our virtual DME procedure; this is the trick to bypass the demanding tomography approach, and it can be further integrated with the eigenvalue transformation technique in Section~\ref{sec:nonlinear-transformation}.
As a result, we have a systematic framework to probe the entropy, similar to the QSVT framework under the purified query access model. 
We note that, while this direction has been discussed in Ref.~\cite{gilyen2022-fidelity}, unfavorable scaling in the number of state copies due to the use of the conventional DME, $\mathcal{O}(t^2/\varepsilon)$ limits the practical utility of the approach.
Our virtual-DME-based algorithm overcomes this bottleneck of the approach, enabling the integration of the eigenvalue transformation framework with the sample access model, thereby paving the road to efficient and scalable estimation of nonlinear quantum properties.
We remark that while we focus on the von Neumann entropy and the quantum relative entropy in the following, the same approach can be readily extended to other nonlinear properties such as alternative entropy measures and distance metrics, by applying appropriate polynomial approximations via the eigenvalue transformation.

Our algorithm uses a Hadamard-test circuit as depicted in Fig.~\ref{fig:entropy-estimation-circuit}.
The most challenging part of entropy estimation is the evaluation of the nonlinear function, e.g., $\ln(\rho)$.
To address this, we employ the eigenvalue transformation framework similarly to the previous work~\cite{wang2024-new}.
Specifically, by assuming the access to the controlled-$e^{i\rho/2}$, we construct a unitary $W_P$ that is an approximate block encoding of $P(\rho) \approx \ln (\rho)$. 
Since the implementation of this circuit follows the structure of Eq.~\eqref{eq:target_exp_for_UM},
we can leverage our virtual DME instead of directly implementing $e^{i\rho/2}$.
As a result, this algorithm achieves the necessary number of copies of 
$\mathcal{O}(\kappa^2 \mathrm{polylog}(\varepsilon^{-1}, \kappa))$ per circuit, 
for the accuracy $\varepsilon$ and the smallest nonzero eigenvalue $\kappa^{-1}$ of $\rho$.
This is an exponential speedup over the previous proposal in the sample access model~\cite{wang2023-entropy}, which requires at most $\mathcal{O}(\kappa^2\varepsilon^{-1}\mathrm{polylog}(\varepsilon^{-1}, \kappa))$ copies per circuit. 
Furthermore, the total number of copies including the measurement cost is also polynomially improved to $\mathcal{O}(\varepsilon^{-2}\kappa^2\mathrm{polylog}(\varepsilon^{-1}, \kappa))$ from $\mathcal{O}(\varepsilon^{-5}\kappa^2\mathrm{polylog}(\varepsilon^{-1}, \kappa))$ in the accuracy. 
This result is summarized as follows.

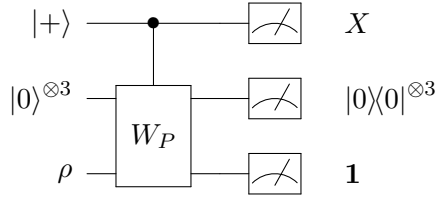
\begin{figure}[htbp]
\centering
\begin{tabular}{c}
\Qcircuit @C=1em @R=1.2em {
    \\
    & \lstick{\ket{+}}  &  \ctrl{1}                &  \qw  & \meter & \rstick{X} \\
    & \lstick{\ket{0}^{\otimes 3}}  &  \multigate{1}{W_P}  &  \qw  & \meter  & \rstick{\ketbra{0}^{\otimes 3}}\\ 
    & \lstick{\rho}     &  \ghost{W_P}         &  \qw  & \meter & \rstick{\bm{1}} \\
    \\
}
\\
\end{tabular}
\caption{The entropy estimation circuit.}\label{fig:entropy-estimation-circuit}
\end{figure}

\begin{thm}\label{thm:entropy-estimation}
Let $\rho$ be an unknown but accessible $d$-dimensional quantum state.
Suppose $1/\kappa$ is a known lower bound on the nonzero eigenvalue of $\rho$.
Then, there exists a quantum algorithm that estimates the von Neumann entropy
$S(\rho) = -\mathrm{tr}[\rho \ln \rho]$ within additive error $\varepsilon$ with a high probability, using random circuits with
\begin{itemize}
\item $\mathcal{O}\left(\kappa^2 \log^3\left(\frac{\kappa}{\varepsilon}\right)\log^2\left(\frac{1}{\varepsilon}\right)\right)$ copies of $\rho$
\item $\mathcal{O}\left(\kappa^2 \log^3\left(\frac{\kappa}{\varepsilon}\right)\log^2\left(\frac{1}{\varepsilon}\right)\log(d)\right)$ one- or two-qubit elementary gates
\end{itemize}
per circuit. 
The measurement cost is given by $\mathcal{O}\left(\frac{\log^2 (\kappa)}{\varepsilon^2}\right)$.
\end{thm}

This result is straightforwardly extended to the estimation of quantum relative entropy $D(\rho \,\|\, \sigma)$, as follows.

\begin{cor}\label{cor:relative-entropy-estimation}
Let $\rho$ and $\sigma$ be unknown but accessible $d$-dimensional quantum states.
Suppose $1/\kappa^{(\rho)}$ and $1/\kappa^{(\sigma)}$ are known lower bounds on the nonzero eigenvalues of $\rho$ and $\sigma$, respectively. 
Assume that $\mathrm{supp}(\rho)\subseteq \mathrm{supp}(\sigma)$.
Then, there exists a quantum algorithm that estimates the quantum relative entropy
$D(\rho \,\|\, \sigma) = \mathrm{tr}[\rho \ln \rho - \rho \ln \sigma]$ within additive error $\varepsilon$ with a high probability
using random circuits with
\begin{itemize}
    \item $N^{(\rho)}$ copies of $\rho$ and $\mathcal{O}(N^{(\rho)}\log d)$ one- or two-qubit elementary gates,
    \item or, $N^{(\sigma)}$ copies of $\sigma$, 1 copy of $\rho$, and $\mathcal{O}(N^{(\sigma)}\log d)$ one- or two-qubit elementary gates
\end{itemize}
per circuit. Here, $N^{(\sigma)}$ and $N^{(\rho)}$ are given by
\begin{equation}
N^{(\rho)} = \mathcal{O}\left((\kappa^{(\rho)})^2 \log^3\left(\frac{\kappa^{(\rho)}}{\varepsilon}\right)\log^2\left(\frac{1}{\varepsilon}\right)\right),~~~N^{(\sigma)} = \mathcal{O}\left((\kappa^{(\sigma)})^2 \log^3\left(\frac{\kappa^{(\sigma)}}{\varepsilon}\right)\log^2\left(\frac{1}{\varepsilon}\right)\right).
\end{equation}
The overall measurement cost is given by $\mathcal{O}\left(\frac{\log^2
 (\kappa^{(\rho)}) + \log^2 (\kappa^{(\sigma)})}{\varepsilon^2}\right)$.
\end{cor}
Proofs of Theorem~\ref{thm:entropy-estimation} and Corollary~\ref{cor:relative-entropy-estimation} are given in Section~\ref{sec:proof-entropy}.

\subsection{Quantum precomputation}\label{apsec:precomp}


We here consider the case where target quantum states $\rho$ are known unlike the previous sections.
In this case, the virtual DME can take a role in significantly reducing the ``wall-clock'' time of quantum computing, given reasonably precomputed quantum resources.
As in classical computing, reducing wall-clock time to solution is crucial for time-sensitive information processing.

Let us consider the task to probe properties of a quantum circuit including multiple $e^{-i\rho T}$ with not too large $T$ and known $\rho$; this setup includes quantum linear system solvers as described later. 
To solve the task, we have two options when we have a classical description of $\rho$: simply performing the desired circuit with the quantum gate $e^{-i\rho T}$ constructed from the classical description, or using our quantum algorithm with the consumption of $\rho$ prepared in a quantum memory. 
In the second option, we assume that the preparation of $\rho$ in a quantum memory can be done in advance with a reasonable computational cost. 
The number of elementary gates to implement the unitary gate $e^{-i\rho T}$ from a classical description of $\rho$ may scale linearly in its system dimension $d$ even for small $T$.
On the other hand, in our method, such expensive gates are replaced with copies of $\rho$ and random operations whose cost is logarithmic in both the dimension and accuracy.
Thus, our method offers an exponential reduction in circuit depth with respect to $d$ if we can neglect the cost of preparing $\rho$. 
Since the circuit depth (and gate counts) is a natural metric of the execution time of a quantum circuit, this indicates that our method is particularly suited for practical scenarios requiring time-sensitive information processing when preparing the resource state $\rho$ ahead of time is available.
A part of this advantage was already discussed in the framework of \textit{quantum precomputation}~\cite{Huggins2024-qf} because the exponential circuit depth reduction with respect to $d$ can also be achieved by the conventional method for density matrix exponential~\cite{Lloyd2014-nn}.
However, this is available only for a constant accuracy due to the poor scaling $\mathcal{O}(1/\varepsilon)$ of the previous method to implement $e^{-i\rho T}$, while the other part of the main circuit may have a logarithmic depth $\mathcal{O}(\log(1/\varepsilon))$.
Our method improves this accuracy requirement exponentially when we aim to probe a final state.
Also, the size of a quantum memory to store the copies of $\rho$ is exponentially smaller with respect to $\varepsilon$ than that for the previous method.

One of the important applications of this strategy is to solve a linear system problem $A\ket{x}=\ket{b}$~\cite{harrow2009quantum, childs2017quantum, lin2020optimal, martyn2021grand, an2022quantum, costa2022optimal, dalzell2024shortcut}, which is a fundamental task across numerous scientific fields such as differential equations, machine learning, and many-body physics; see the recent comprehensive review~\cite{morales2024quantum} in the quantum setup.
The most efficient quantum algorithms to prepare the solution $\ket{x}$ of the problem use multiple $e^{-i\pi \ket{b}\bra{b}}$, called the reflection operation~\cite{morales2024quantum}.
In these advanced algorithms, we can reasonably assume that the circuit depth is logarithmic (or poly-logarithmic) in both the dimension and accuracy, except for the reflection gates~\cite{morales2024quantum,Huggins2024-qf}.
Now, considering a situation where a classical description of $\ket{b}$ is available before specifying $A$, we prepare multiple copies of $\ket{b}$ in advance~\cite{Huggins2024-qf}.
By consuming the precomputed states, our method can probe the properties of the solution vector $\ket{x}$ using quantum circuits with a logarithmic depth regarding both the accuracy and dimension.
As mentioned before, this offers an exponential depth reduction thanks to being able to ignore the preparation cost of $\ket{b}$, compared to the standard algorithms performing the expensive reflections during the execution of a main circuit.


To apply our quantum algorithm to the linear system problem, it is sufficient to specify the corresponding function $f_M$ such that 
\begin{equation}
    f_M(\ket{b}\bra{b})\simeq c\bra{x}O\ket{x}
\end{equation}
for a constant value $c>0$, a desired observable $O$ and a certain $M$.
Note that by setting $O=\bm{1}$, we can obtain the constant $c$ and recover the target value $\bra{x}O\ket{x}$ by division.
By explicitly constructing one example of such $f_M$ based on one of the most efficient quantum linear system solvers~\cite{lin2020optimal}, we prove the following theorem.
\begin{thm}\label{thm:application_linearsyssolver}
    Let $U_A$ be an $(1,a,0)$-block encoding of a $d$-dimensional Hermitian $A>0$ whose eigenvalues are in $[1/\kappa,1]$ for known $\kappa>1$.
    Suppose we have access to copies of an unknown (or known) state $\ket{b}$ for a target linear system problem $A\ket{x}=\ket{b}$.
    For a given observable $O$, there exists a quantum algorithm that estimates the expectation value $\bra{x}O\ket{x}$ within a given additive error $\varepsilon$ with high probability, using random circuits with  
    \begin{itemize}
        \item $M=\mathcal{O}(\kappa \left[\log\kappa\log\log\kappa+\log(\|O\|/\varepsilon)\right])$ queries of $U_A$ and its inverse
        \item $\mathcal{O}(\kappa^2 \left[\log\kappa\log\log\kappa+\log(\|O\|/\varepsilon)\right]^2 )$ copies of $\ket{b}$, and  
        \item $\mathcal{O}({aM+M^2\log d})$ one- or two-qubit elementary gates.
    \end{itemize}
    The total number of measurements is given by $\mathcal{O}(\|O\|^2/\varepsilon^2)$.
\end{thm}
\noindent
The proof of this theorem is provided in Section~\ref{sec:thm_proof_QLS}.
If the circuit depth for $U_A$ is poly-logarithmic in the system dimension $d$, then the whole circuit depth is logarithmic (or poly-logarithmic) in both $\varepsilon$ and $d$ when we ignore the preparation cost of the initial state $\ket{b}$.
For simplicity, we here assume that the coefficient matrix $A$ is Hermitian and $A>0$.
A generalization to a more general $A$ would be possible by slightly modifying the procedure in the proof.
Also, the scaling of the query complexity to $A$ is almost linear with respect to $\kappa$, which is the upper bound of the condition number of $A$.
This is the same scaling (up to logarithmic factors) as the state-of-the-art methods for quantum linear system problems~\cite{morales2024quantum}.
The factor $\log\kappa$ in the complexities can be eliminated by additionally using the fixed-point amplitude amplification in the circuit of $f_M$~\cite{lin2020optimal}.

\clearpage
\section{Quantum principal component analysis}
\label{sec:quantum_PCA}

In quantum information processing, the principal component(s) of a quantum state, i.e., eigenstates with dominant eigenvalues of the density matrix, is a highly informative resource. 
The original formulation together with a solution to this problem, called the quantum principal component analysis (qPCA), was established in \cite{Lloyd2014-nn}. 
A particularly interesting scenario is the case where we use digital/analog quantum computers or future quantum sensors to generate quantum data in the form of a density matrix and store it in a quantum memory~\cite{huang2022quantum}. 
In this case, although the resulting state is inevitably subjected to noise such as decoherence and control error, the target information would be relatively robustly retained in the first principal component against the noise. 
A similar situation arises when considering a small-scale logical quantum computing~\cite{egan2021fault, PhysRevX.11.041058, paetznick2024demonstrationlogicalqubitsrepeated, bluvstein2024logical, reichardt2024logical, acharya2024quantum, sundaresan2023demonstrating, Acharya2023-yl, krinner2022realizing, Abobeih_2022} (i.e., when we can operate logical qubits with a modest number of logical operations), especially for the early stage of fault-tolerant quantum computing~\cite{katabarwa2024early}.
Specifically, when using error-corrected qubits with modest code-distance, one can naturally expect that a final output quantum state is dominated by the first principal component, namely the ideal pure state for noiseless quantum algorithms, while a full characterization of the remaining noise is challenging.
In addition, several other tasks such as data analysis~\cite{Lloyd2014-nn,PRXQuantum.3.030334}, process tomography~\cite{Lloyd2014-nn}, and electronic structure analysis~\cite{Lloyd2014-nn,PRXQuantum.3.030334} can also be performed by probing the principal component of quantum states.

Below, we consider the following setup: we have access to copies of an unknown quantum state $\rho$ which, without loss of generality, can be written as 
\begin{equation}
\label{eq:original noise_assumption}
\rho=(1-\lambda)\ketbra{\psi}+\lambda\sum_{k}p_k \ketbra{\phi_k}
:=(1-\lambda)\ketbra{\psi}+\lambda\rho_{\mathrm{err}},   
\end{equation}
where $\ket{\psi}$ and $\ket{\phi_k}$ are eigenstates of $\rho$ with eigenvalues $1-\lambda$ and $\lambda p_k$, respectively.
$\lambda$ is a parameter describing the noise intensity (assumed to be not too large), and $\rho_{\mathrm{err}}$ is an unknown noisy state satisfying $\mathrm{tr}[\rho_{\rm err} \ketbra{\psi}]=0$. 
Our goal is to estimate the expectation value $\bra{\psi}O\ket{\psi}$ of a given observable $O$ with respect to the dominant pure state $\ket{\psi}$ up to a small additive error $\varepsilon$. 
This problem was solved by the original DME-based protocol for qPCA~\cite{Lloyd2014-nn} with remarkably a logarithmic complexity for the target system dimension $d$.
Later, by Ref.~\cite{huang2022quantum}, an exponential advantage of this protocol compared to any classical algorithm with single-copy measurement outcomes is rigorously established even for a fixed $\varepsilon$.
However, the original qPCA protocol requires a deep circuit with a polynomial depth regarding $1/\varepsilon$ based on quantum phase estimation~\cite{kitaev1995quantum} with DMEs, and thus its practical implementation is highly difficult until fully unlocking the power of fault-tolerant quantum computing. 
Note that we can apply the quantum state purification method~\cite{PhysRevLett.82.4344,keyl2001rate,Childs2025streamingquantum,Yao_2025,yang2024quantumerrorsuppressionsubgroup,li2024optimal,brahmachari2025optimal} to this problem, but a similar difficulty arises.
On the other hand, the exponential advantage still holds when we use quantum virtual cooling~\cite{PhysRevX.9.031013} or quantum error mitigation strategies such as virtual distillation (VD)~\cite{Huggins2021-vd} and error suppression by derangement (ESD)~\cite{koczor2021-esd},
which have an aspect of near-term variants of the original qPCA protocol~\cite{huang2022quantum,Cai2023-error-mitigation-survey}.
These methods suppress the error in $\rho_{\rm err}$ at an exponential rate with respect to the number of state copies per circuit, thereby achieving an $\mathcal{O}(\log (1/\varepsilon))$ depth circuit over the copies; however, they require an exponentially large measurement overhead regarding the number of the state copies. 
As another direction, there have been proposed several methods based on variational quantum algorithms~\cite{larose2019variational,verdon2019quantum,xin2021experimental,cerezo2022variational}, but the number of measurements required for them is not clear (rather, it might be prohibitively large).
Overall, there was no method archiving (i) an $\mathcal{O}({\rm polylog}(d,1/\varepsilon))$-depth circuit, (ii) $\mathcal{O}({\rm polylog}(1/\varepsilon))$ state copies per circuit, and (iii) a standard measurement cost $\mathcal{O}(1/\varepsilon^2)$, that would significantly advance the achievement of the exponential quantum advantage along with practically small $\varepsilon$.

In the following, we only assume that $\lambda$ is not too large and we know an upper bound of the parameter $\lambda$. 
This assumption on $\lambda$ is naturally expected to be held, e.g., in the case where we can perform some quantum algorithms (such as Hamiltonian simulation) on small-scale logical qubits. 
Also, we ignore the difference between the dominant eigenstate $\ket{\psi}$ and an ideal pure state obtained from a noiseless target quantum process because in practical scenarios, such coherent errors can be converted into incoherent errors by using, e.g., Pauli twirling~\cite{PhysRevA.78.012347,PhysRevA.85.042311,Cai_2019,Cai_2020}.

\begin{figure}[htbp]
 \centering
 \begin{tabular}{cc}
 \includegraphics[width=0.45\linewidth,page=1]{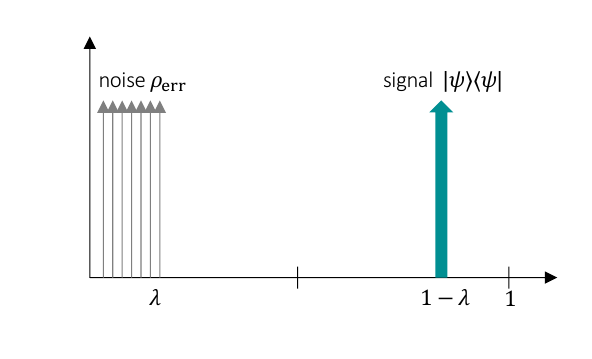} & 
 \includegraphics[width=0.45\linewidth,page=2]{Fig_filter2.pdf}\\
 (a) Original spectrum of $\rho$ & (b) Filtered spectrum
 \end{tabular}
 \caption{Filter function for the principal eigenstate projector. 
 Here, $\rho$, $\rho_{\mathrm{err}}$, and $\ket{\psi}\bra{\psi}$ are $d=2^n$ dimensional quantum states defined in Eq.~\eqref{eq:original noise_assumption}. 
 $\varepsilon_1$ and $\varepsilon_2$ are the filter errors in the pass band and the stop band, respectively.
 }
 \label{fig:error-mitigation-filter}
\end{figure}

\subsection{Main result}

Here we show that the virtual-DME-based protocol achieves the above-mentioned requirements (i)-(iii); i.e., we 
significantly lower the barrier for realizing the exponential quantum advantage of the qPCA protocol by providing a highly efficient way to filter out the contribution of unknown $\rho_{\rm err}$ via the virtual DME.
The main idea is to use $e^{-i\rho T}$ to construct a projector onto the most dominant eigenstate $\ket{\psi}$; then applying the projector to $\rho$, we can estimate the properties of $\ket{\psi}$. 
A quantum circuit for such a dominant eigenstate projector can be designed by using some advanced quantum algorithms~\cite{gilyen2019-qsvt, Dong2022-qetu}, when we have access to the controlled version of $e^{-i\rho T}$ even without knowing $\rho$. 
Specifically, we use a quantum circuit with the controlled $e^{-i\rho T}$ ($T$ is a unit time) to perform the eigenvalue transformation of $\rho$ based on an $\mathcal{O}(\log(1/\varepsilon))$-degree polynomial approximating a step function; see Fig.~\ref{fig:error-mitigation-filter}.
Based on this circuit, we can construct a function $f_M$ to estimate the property of the first principal component  $\ket{\psi}$ within the error $\varepsilon$ by using $M=\mathcal{O}({\rm log}(1/\varepsilon))$ DME operations, leading to poly-logarithmic depth circuits with respect to $d$ and $1/\varepsilon$.
Importantly, our quantum algorithm yields an estimate $f_M$ with only a standard measurement cost $\mathcal{O}(1/\varepsilon^2)$ (more specifically, $\mathcal{O}((1-\lambda)^{-2}/\varepsilon^2)$ as shown later) without any calibration cost to characterize the noisy state $\rho_{\rm err}$. 
Noting that instead of directly using the original qPCA protocol, which requires many additional qubits for phase estimation, we here leverage the eigenvalue transformation technique with a few additional ancilla qubits and a smaller number of gates.

Concretely, the algorithm uses the circuit depicted in Fig.~\ref{fig:mitigation-circuit}.
Let $W_P$ be a block encoding of a filter function, which eliminates the noisy state $\rho_{\rm err}$ without losing the principal eigenstate information (see schematic image in Fig.~\ref{fig:error-mitigation-filter}).
This eigenvalue filtering operation yields the projection onto the principal eigenstate $\ket{\psi}$. 
We can construct $W_P$ by QSVT with a polynomial approximating the step function and the time-evolution access to $e^{i\rho/2}$.
By measuring $\ketbra{0} \otimes O$ and $\ketbra{0} \otimes \bm{1}$ at the end of the circuits in Fig.~\ref{fig:mitigation-circuit}, we can construct an estimator that approximately recovers the noise-free expectation value:
\begin{equation}
\label{eq:error-suppression}
    \frac{\mathrm{tr}[(\ketbra{0} \otimes O)(W_P(\ketbra{0} \otimes \rho)W_P^\dagger)]}{\mathrm{tr}[(\ketbra{0} \otimes \bm{1})(W_P(\ketbra{0} \otimes \rho)W_P^\dagger)]}
    \simeq \frac{(1-\lambda)\Braket{\psi|O|\psi}}{1-\lambda} 
    = \Braket{\psi|O|\psi}.
\end{equation}
Thus, our projector filters out the noise while maintaining the signal from the principal eigenstate.
Also, we can show that the measurement overhead of our algorithm is increased only by a factor $\mathcal{O}((1-\lambda)^{-2})$ caused by the original signal loss.
Since the circuit in Fig.~\ref{fig:mitigation-circuit} has the form of that for calculating $f_M$ in Eq.~\eqref{eq:target_exp_for_UM}, we can use the general algorithm with the virtual DME in Section~\ref{sec:nearlyopt_qalg} to estimate the expectation value regarding $\ketbra{\psi}$.
The performance of this noise-agnostic quantum algorithm is summarized as follows.

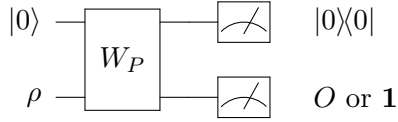
\begin{figure}[tbp]
\centering
\begin{tabular}{c}
\Qcircuit @C=1em @R=1.2em {
    \\
    & \lstick{\ket{0}}  &  \multigate{1}{W_P}  &  \qw  & \meter & \rstick{\ketbra{0}}  \\ 
    & \lstick{\rho}     &  \ghost{W_P}         &  \qw  & \meter & \rstick{O \mathrm{~or~} \bm{1}} \\
    \\
    }
    \\
\end{tabular}
\caption{The error filtering circuit. ß$W_P$ denotes the block encoding of the approximated projector, which is constructed by multiple uses of $e^{i \rho/2}$. 
}
\label{fig:mitigation-circuit}
\end{figure}

\begin{thm}\label{thm:error-mitigation}
Let $\rho$ be an accessible $d=2^n$-dimensional unknown quantum state with the spectral decomposition
\begin{equation}
\label{eq:noise_assumption}
    \rho = (1-\lambda) \ketbra{\psi} + \lambda \rho_{\mathrm{err}},
\end{equation}
where $\ket{\psi}$ is a principal eigenstate and $\rho_{\rm err}$ is an arbitrary unknown state satisfying $\mathrm{tr}[\rho_{\mathrm{err}} \ketbra{\psi}]=0$.
$\lambda \in [0, \frac{1}{2})$ is an unknown parameter.
Suppose $\eta \in [0 , \frac{1}{2})$ is a known upper bound of $\lambda$ (i.e., $\lambda \le \eta$).
For any observable $O$,
there is a quantum algorithm that estimates the expectation value $\bra{\psi} O \ket{\psi}$,
within additive error $\varepsilon \in (0, 8\|O \|]$ with high probability,
using
\begin{itemize}
    \item $N = \mathcal{O}\left(\frac{1}{(1 - 2\eta)^2}\log^2\left(\frac{\|O\|}{\varepsilon}\right)\log\left(\frac{1}{(1-2 \eta) \varepsilon}\right) \right)$ copies of $\rho$
    \item $\mathcal{O}\left(N\log(d)\right)$ one- or two-qubit elementary gates
\end{itemize}
per circuit.
The total number of measurements
is given by $\mathcal{O}(\frac{\|O\|^2}{(1-\eta)^2 \varepsilon^2})$.
\end{thm}

\noindent
We prove this theorem by explicitly constructing a quantum algorithm; see the proof in Section~\ref{sec:proof-error-mitigation}.
In the explicit algorithm, we use simple quantum primitive operations to simplify the proof; the number of copies, elementary gates, and additional qubits may be further reduced by optimizing our algorithm. 
Although we put the assumption of $\lambda<1/2$ for simplicity, the complexity depends only on the gap between the primary and the second largest eigenvalues; that is, our algorithm can be generalized to the case $\lambda \ge 1/2$. 
In addition, while we here consider only the first principal component, a generalization for multiple principal components would be straightforward by appropriately changing the filter polynomial. 

\subsection{Numerical simulation}

\begin{figure}[tb]
  \centering
  \begin{minipage}{\textwidth}
    \centering
    \textbf{(a) Coherent filter}\\
    ~~~\Qcircuit @C=0.3em @R=0.8em {
    \\
      \lstick{\ket{0}} & \gate{e^{i\phi_1 X}} & \ctrl{1} & \ctrlo{1} & \gate{e^{i\frac{Z}{2}}}&\qw&&& \cdots &&& \gate{e^{i\phi_{2M_{\rm f}} X}} & \ctrl{1} & \ctrlo{1} & \gate{e^{i\frac{Z}{2}}} &\gate{e^{i\phi_{2M_{\rm f}+1} X}} & \meter &\rstick{\ketbra{0}}\\
      \lstick{\rho}& {/}\qw   & \gate{e^{i\frac{\rho}{2}}} & \gate{e^{-i\frac{\rho}{2}}}&\qw&\qw &&& \cdots &&& \qw & \gate{e^{i\frac{\rho}{2}}} & \gate{e^{-i\frac{\rho}{2}}}& \qw& \qw & \meter &\rstick{O}
    \\
    \\
    }
  \end{minipage}
  \\
  \bigskip
  \begin{minipage}{\textwidth}
    \centering
    \textbf{(b) Classical and Quantum Hybrid filter}\\
    ~~~\Qcircuit @C=0.3em @R=0.8em {
    \\
      \lstick{\ket{0}} & \gate{H} & \ctrl{1} & \ctrlo{1}     &\gate{e^{-ikZ/2}} & \meter &\rstick{X}\\
      \lstick{\rho}& {/} \qw      & \gate{e^{-ik\frac{\rho}{2}}}    &  \gate{e^{ik\frac{\rho}{2}}}  &\qw & \meter &\rstick{O}\\
    \\
    }
  \end{minipage}

  \caption{Two types of filters. (a) Coherent filter. The parameters $\{\phi_j\}$ are efficiently calculated from the description of a target even real polynomial $P(x)$ satisfying $|P(x)|\leq 1$ for all $x\in[-1,1]$~\cite{Dong2022-qetu}. (b) Classical and quantum hybrid filter. The index $k$ is randomly chosen according to the probability distribution $\lambda_k$ at each circuit run.}
  \label{fig:coherent_hybrid_filters}
\end{figure}

We here provide a numerical simulation to compare the quantum resource of several algorithms: two virtual-DME-based protocols, the original qPCA protocol, and the VD method. 
The first two are the newly proposed quantum algorithms that effectively construct the quantum filter in Fig.~\ref{fig:error-mitigation-filter} using $e^{-i\rho T}$ and then apply it to the unknown state $\rho=(1-\lambda)\ketbra{\psi}+\lambda {\rho}_{\rm err}$. 
We will call these two filters the {\it coherent filter} and the {\it hybrid filter}, the detailed procedures of which will be given in Algorithm~\ref{alg:coherentfilter} and \ref{alg:hybridfilter}, respectively.

The coherent filter algorithm is essentially the same as the algorithm constructed in the proof of Theorem~\ref{thm:error-mitigation}, and the hybrid one is a more hardware-efficient approach, which utilizes a quantum-classical hybrid filtering technique. 
Both filters contain the following common component. 
For a simple numerical calculation, we consider a filter for $\bm{1}-\rho$ as follows, instead of Fig.~\ref{fig:error-mitigation-filter} for $\rho$. 
Let $F(x)=\sum_{k=0}^{M_{\rm f}} f(k)\cos(kx)$ be a real function such that
\begin{equation}\label{eq:condition_F}
    \max_{x\in [0,\eta]}|F(x)-1|\leq \varepsilon_1,~~~\max_{x\in [1-\eta,1]}|F(x)|\leq \varepsilon_2
\end{equation}
holds for a known upper bound $\eta$ of $\lambda$ and error parameters $\varepsilon_1,\varepsilon_2>0$.
Given the parameters $\eta, \varepsilon_1,\varepsilon_2$, we can numerically find the minimal possible filter order $M_{\rm f}$ and the coefficients $\{f(k)\}$ by using e.g., Remez algorithm~\cite{Remez1934a,Remez1934b,Remez1934c}.
Then, because $\bm{1}-\rho$ has an eigenvalue $\lambda\in[0,\eta]$ with corresponding eigenstate $\ket{\psi}$ while the other eigenvalues live in the range $[1-\eta,1]$, the resulting filtered operator $F(\bm{1}-\rho)$ becomes an approximate projector $F(\bm{1}-\rho) \simeq F(\lambda)\ketbra{\psi} \simeq \ketbra{\psi}$.

\renewcommand{\baselinestretch}{1.2}
\begin{figure}[htb]
\begin{algorithm}[H]
    \caption{Estimation of ${\rm tr}[OF(\bm{1}-\rho)\rho F(\bm{1}-\rho)]$ via Coherent Filter}\label{alg:coherentfilter}
    \begin{algorithmic}[1]
    \smallskip
    \REQUIRE 
    Copies of states $\rho$, real function $P(x)=\sum_{k=0}^{M_{\rm f}} f(k) \cos(2k\arccos(x))$ satisfying $|P(x)|\leq 1$ for all $x\in[-1,1]$, observable $O$, bias from virtual DME $\Delta'>0$

    \smallskip
    \ENSURE A sample $\mu_1$ of an estimator $\hat{\mu}_1$ satisfying 
    $\left|\mathbb{E}[\hat{\mu}_1]-{\rm tr}[OF(\bm{1}-\rho)\rho F(\bm{1}-\rho)]\right|\leq \Delta'$ and ${\rm Var}[\hat{\mu}_1]\leq e\|O\|^2$.
    
    \medskip
    \STATE For $T=-1/2, r=2M_{\rm f}$, and $\varepsilon=W_0(\Delta'/\|O\|)/2M_{\rm f}$, generate $2M_{\rm f}$ samples from $\hat{\Phi}_r'$ and calculate $C$ in Theorem~\ref{thm:controlledDME}.

    \STATE Run the circuit in Fig.~\ref{fig:coherent_hybrid_filters}(a) where $2M_{\rm f}$ gates in the form of $\ketbra{0}\otimes e^{-i\rho/2}+\ketbra{1}\otimes e^{i\rho/2}$ are replaced with these $\hat{\Phi}_r'$

    \STATE Set $\widetilde{g}\leftarrow$ the measurement outcome including the mid-circuit measurements in $\hat{\Phi}_r'$.

    \RETURN $\mu_1=[C]^{8M^2_{\rm f}}\widetilde{g}$ \COMMENT{Note that $[C]^{8M^2_{\rm f}}\leq \sqrt{e}$}
    
    \end{algorithmic}
\end{algorithm}
\end{figure}
\renewcommand{\baselinestretch}{1}

Algorithm~\ref{alg:coherentfilter} computes the target quantity Eq.~\eqref{eq:error-suppression} based on the quantum circuit depicted in Fig.~\ref{fig:coherent_hybrid_filters}(a).
Since $P(x):=F(2\arccos(x))$ is an even real $2M_{\rm f}$-degree polynomial, we can explicitly construct the circuit Fig.~\ref{fig:coherent_hybrid_filters}(a) \cite[Corollary 17]{Dong2022-qetu} corresponding to Fig.~\ref{fig:mitigation-circuit}, where now $W_P$ is an $(1,1,0)$ block encoding of $P(\cos((\bm{1}-\rho)/2))=F(\bm{1}-\rho)$.
Note that if $P(x)=F(2\arccos(x))$ does not satisfy the condition $|P(x)|\leq 1$ ($|x|\leq 1$), we can simply rescale it.
This circuit yields estimates for ${\rm tr}[OF(\bm{1}-\rho)\rho F(\bm{1}-\rho)]$ and ${\rm tr}[F(\bm{1}-\rho)\rho F(\bm{1}-\rho)]$ by the procedure based on the virtual DME in Algorithm~\ref{alg:coherentfilter}.
By taking the average to reduce the shot noise, we can estimate $\bra{\psi}O\ket{\psi}$ using the estimator given in Eq.~\eqref{eq:error-suppression}.
Here, we take a sufficiently small error $\Delta'$ for the virtual DME in Algorithm~\ref{alg:coherentfilter}.
From the choice of the filter $F(x)$, we can show the following, in a similar way to the proof of Theorem~\ref{thm:error-mitigation},
\begin{itemize}
    \item Approximation error $\Delta$: 
    \begin{equation}
    \left|\frac{{\rm tr}[OF(\bm{1}-\rho)\rho F(\bm{1}-\rho)]}{{\rm tr}[F(\bm{1}-\rho)\rho F(\bm{1}-\rho)]}-\bra{\psi}O\ket{\psi}\right|\leq \frac{\eta}{1-\eta}2\|O\|\left(\frac{\varepsilon_2}{1-\varepsilon_1}\right)^2=:\Delta,
    \end{equation}
    \item Measurement overhead reflecting the division of the denominator: 
    \begin{equation}
    \frac{{\rm Var}[\hat{\mu}_1]}{{\rm tr}[F^2(\bm{1}-\rho)\rho]^2}\leq \frac{e \|O\|^2}{{\rm tr}[F^2(\bm{1}-\rho)\rho]^2}=:\gamma^2_{\rm coherent},
    \end{equation}
    \item Expectation of the number of copies per circuit:
    \begin{equation}
    1 + \left(2M_{\rm f}\times \mbox{Average number of state copies in }\hat{\Phi}'_{r=2M_{\rm f}}\right)\simeq 1+8M_{\rm f}^2.
    \end{equation}    
\end{itemize}

\renewcommand{\baselinestretch}{1.2}
\begin{figure}[htb]
\begin{algorithm}[H]
    \caption{Estimation of ${\rm tr}[OF(\bm{1}-\rho)\rho]$ via Hybrid Filter}\label{alg:hybridfilter}
    \begin{algorithmic}[1]
    \smallskip
    \REQUIRE 
    Copies of states $\rho$, function $F(x)=\sum_{k=0}^{M_{\rm f}} f(k) \cos(kx)$ with real coefficients $f(k)$, observable $O$, bias $\Delta'>0$, and a probability distribution $\{\lambda_k\}_{k=0}^{M_{\rm f}}$ ($\lambda_k\neq 0$)

    \smallskip
    \ENSURE A sample $\mu_2$ from an estimator $\hat{\mu}_2$ satisfying $\left|\mathbb{E}[\hat{\mu}_2]-{\rm tr}[OF(\bm{1}-\rho)\rho]\right|\leq \Delta'$ and ${\rm Var}[\hat{\mu}_2]\leq \sum_{k=0}^{M_{\rm f}} \frac{f(k)^2}{\lambda_k}\sigma_k^2$
    where $\sigma_0=\|O\|$ and $\sigma_{k\neq 0}= \sqrt{e}\|O\|$.

    \medskip
    
    \STATE Sample $k$ from the probability distribution $\lambda_k$

    \STATE For $T=k/2,r=k^2$, and $\varepsilon=\Delta'/ \|O\|(\sum_{k=1}^{M_{\rm f}} |f(k)|)$, sample $\hat{\Phi}'_{r}$ and calculate $C$ in Theorem~\ref{thm:controlledDME}.
    (if $k=0$, skip this step)
    \STATE Run the circuit in Fig.~\ref{fig:coherent_hybrid_filters}(b) where $\ketbra{0}\otimes e^{ik\rho/2}+\ketbra{1}\otimes e^{-ik\rho/2}$ is replaced with the $\hat{\Phi}'_{r}$.
    \STATE Set $\widetilde{g}\leftarrow$ the measurement outcome including the Pauli X measurement in $\hat{\Phi}_r'$.
    
    \RETURN $\mu_2=C^{2k^2}f(k)\widetilde{g}/\lambda_k$ (if $k=0$, $\mu_2=f(k)\widetilde{g}/\lambda_k$) \COMMENT{Note that $1<C^{2k^2}\leq \sqrt{e}$}
    
    \end{algorithmic}
\end{algorithm}
\end{figure}
\renewcommand{\baselinestretch}{1}

Next, we consider the second approach based on the hybrid filter with the quantum circuit depicted in Fig.~\ref{fig:coherent_hybrid_filters}(b), which is more hardware efficient than (a).
Since we observe that 
\begin{equation}
    (1-\lambda)\bra{\psi}O\ket{\psi}\simeq {\rm tr}[OF(\bm{1}-\rho)\rho]=\sum_{k=0}^{M_{\rm f}}\lambda_k\frac{f(k)}{\lambda_k}{\rm tr}[O\cos(k(\rho-\bm{1}))\rho]
\end{equation} 
holds for any probability distribution $\lambda_k$ $(\neq 0~\forall k)$, we can estimate $(1-\lambda)\bra{\psi}O\ket{\psi}$ by randomly selecting $k$ according to the probability $\lambda_k$ followed by sampling from an estimator whose mean matches ${f(k)}{\rm tr}[O\cos(k(\rho-\bm{1}))\rho]/\lambda_k$.
The corresponding procedure is given in Algorithm~\ref{alg:hybridfilter}.
After independently repeating this algorithm to estimate ${\rm tr}[OF(\bm{1}-\rho)\rho]$ and ${\rm tr}[F(\bm{1}-\rho)\rho]$, we can estimate $\bra{\psi}O\ket{\psi}$ by calculating the ratio of the average values.
From the choice of $F(x)$ satisfying Eq.~\eqref{eq:condition_F}, we can show that (assuming that we take a sufficiently small $\Delta'$ in Algorithm~\ref{alg:hybridfilter})
\begin{itemize}
    \item Approximation error $\Delta$: 
    \begin{equation}
    \left|\frac{{\rm tr}[OF(\bm{1}-\rho)\rho]}{{\rm tr}[F(\bm{1}-\rho)\rho]}-\bra{\psi}O\ket{\psi}\right|\leq \frac{\eta}{(1-\eta)(1-\varepsilon_1)-\varepsilon_2\eta}2\|O\|\varepsilon_2=:\Delta,
    \end{equation}
    if $(1-\varepsilon_1)(1-\eta)>\varepsilon_2\eta$.
    \item Measurement overhead reflecting the division of the denominator: 
    \begin{equation}
    \frac{{\rm Var}[\hat{\mu}_2]}{{\rm tr}[F(\bm{1}-\rho)\rho]^2}\leq \frac{\sum_{k=0}^{M_{\rm f}} \lambda_k^{-1} f(k)^2 \sigma_k^2 }{{\rm tr}[F(\bm{1}-\rho)\rho]^2}=:\gamma^{2}_{\rm hybrid},
    \end{equation}
    where $\sigma_0=\|O\|$ and $\sigma_{k\neq0}=\sqrt{e}\|O\|$.
    \item Expected value of the number of copies per circuit:
    \begin{equation}
    \sum_{k=0}^{M_{\rm f}} \lambda_k\left(1+\mbox{Average number of state copies in }\hat{\Phi}'_{r=k^2}\right)\simeq \sum_{k=0}^{M_{\rm f}} \lambda_k (1+2k^2).
    \end{equation}    
\end{itemize}
\noindent
We here consider the optimization of the probability distribution $\lambda_k$ such that the product of $\gamma_{\rm hybrid}^2$ and the expected number of copies is minimized.
Such an optimal distribution can be found by the method of Lagrange multiplier, and the resulting optimal distribution is given by
\begin{equation}
    \lambda_k \mathrel{\underset{\sim}{\propto}} \frac{|f(k)|\sigma_k}{\sqrt{1+2k^2}}.
\end{equation}

Finally, we describe the original procedure~\cite{Lloyd2014-nn} of qPCA for comparison.
For simplicity, we write the spectral decomposition of the target state $\rho$ as $\rho\equiv \sum_{i}r_i \ketbra{\chi_i}$ ($r_i\in [0,1]$).
The original procedure uses a quantum phase estimation with multiple (controlled) unitaries $e^{i\rho T}$ and an initial entangled ancillary state $\sum_{k=0}^{2^m-1} a_k \ket{k}$ with $m$ qubits, unlike the uniform superposition state.
Here, $a_k$ is given by~\cite{Lloyd2014-nn}
\begin{equation}
    a_k=\sqrt{\frac{1}{2^{m-1}}}\sin\frac{\pi(k+1/2)}{2^m}.
\end{equation}
Then, the final state of the quantum phase estimation circuit including the inverse QFT becomes
\begin{equation}
    \sum_{l=0}^{2^m-1} \left[\sum_{k=0}^{2^m-1} \frac{a_k}{\sqrt{2^m}}e^{-2\pi i k (l/2^m-r_i T/2\pi)}\right]\ket{l}\ket{\chi_i}
     =: \sum_{l=0}^{2^m-1} C_l(r_i)\ket{l}\ket{\chi_i},
\end{equation}
when the initial state of the target register is an eigenstate $\ket{\chi_i}$.
If we take $\rho$ as the initial target state, the resulting state before measurement is calculated as
\begin{equation}
    (1-\lambda)\ketbra{\psi}\otimes \sum_{l,l'=0}^{2^m-1} C_l(1-\lambda)C^*_{l'}(1-\lambda)\ketbra{l}{l'}+\sum_{k}\lambda p_k \ketbra{\phi_k}\otimes \sum_{l,l'=0}^{2^m-1} C_l(\lambda p_k)C^*_{l'}(\lambda p_k)\ketbra{l}{l'}.
\end{equation}
By taking $T=2\pi$, it follows $|C_l(r_i)|\simeq \delta(l-r_i 2^m)$ when $2^m$ is sufficiently large. 
To extract the most dominant eigenstate, we now assume that we know $l^*$ for the target eigenvalue $1-\lambda$ such that $l^*\simeq (1-\lambda)2^m$.
After performing the post-selection of $\ket{l^*}$ in the ancilla register, we obtain the following state close to the most dominant eigenstate:
\begin{equation}
    \frac{(1-\lambda)|C_{l^*}(1-\lambda)|^2\ketbra{\psi}+\sum_{k}\lambda p_k|C_{l^{*}}(\lambda p_k)|^2 \ketbra{\phi_k}}{(1-\lambda)|C_{l^*}(1-\lambda)|^2+\sum_{k}\lambda p_k|C_{l^{*}}(\lambda p_k)|^2}\simeq \ketbra{\psi}.
\end{equation}
The denominator corresponds to the post-selection probability.

Measuring the observable $O$ on the post-selected state, we can approximately estimate the expectation value of $O$ with respect to the most dominant eigenstate.
This estimation has a bias from the approximation error between the most dominant eigenstate $\ketbra{\psi}$ and the post-selected state.
The approximation error is upper bounded as
\begin{align}\label{qPCA:approx_error}
    &\left|\frac{(1-\lambda)|C_{l^*}(1-\lambda)|^2\bra{\psi}O\ket{\psi}+\sum_{k}\lambda p_k|C_{l^{*}}(\lambda p_k)|^2 \bra{\phi_k}O\ket{\phi_k}}{(1-\lambda)|C_{l^*}(1-\lambda)|^2+\sum_{k}\lambda p_k|C_{l^{*}}(\lambda p_k)|^2}-\bra{\psi}O\ket{\psi}\right|\notag\\
    &\leq \frac{{\sum_{k}\lambda p_k|C_{l^{*}}(\lambda p_k)|^2 \left|\bra{\phi_k}O\ket{\phi_k}-\bra{\psi}O\ket{\psi}\right|}}{{(1-\lambda)|C_{l^*}(1-\lambda)|^2+\sum_{k}\lambda p_k|C_{l^{*}}(\lambda p_k)|^2}}\leq 2\|O\|\frac{{\sum_{k}\lambda p_k|C_{l^{*}}(\lambda p_k)|^2 }}{{(1-\lambda)|C_{l^*}(1-\lambda)|^2+\sum_{k}\lambda p_k|C_{l^{*}}(\lambda p_k)|^2}},
\end{align}
when we ignore the error from the conventional DME for implementing the controlled-$e^{i\rho T}$ operations.
In the numerical simulation shown below, for simplicity, we numerically calculate $C_{l*}(\lambda p_k)$ for the upper bound of the approximation error (while this cannot be done in practice).
From the evaluation of the approximation error, we can specify the number of ancilla qubits $m$ and $M=2^{m}-1$ times use of the controlled-$e^{i\rho T}$ according to the structure of the phase estimation circuit.
When the approximation error is $\Delta$, the implementation error of $e^{i\rho T}$ via the conventional DME procedure should be of the order of $(1-\lambda)\Delta/M$ to achieve the overall approximation error $\sim\Delta$.
As a result, for $m$ ensuring that Eq.~\eqref{qPCA:approx_error} is smaller than $\Delta$, the number of state copies per circuit required for the original qPCA is well approximated by
\begin{equation}
    M\times \frac{T^2}{(1-\lambda)\Delta/M}=\frac{4\pi^2}{\Delta}\frac{(2^{m}-1)^2}{1-\lambda}.
\end{equation}
Also, we remark that the measurement overhead is given by the inverse of the post-selection probability $\sim (1-\lambda)^{-1}$.
\\

We now provide the numerical simulation result depicted in Fig.~\ref{fig:numericalsim_VD} (and \green{Fig.~\ref{fig:main_numerical_res}} in the main text), showing the number of state copies and measurement overhead over the approximation error $\log_{10}\Delta$ for each algorithm. 
In this numerical simulation, we consider the case where $\rho_{\rm err}$ is an arbitrary pure state orthogonal to $\ketbra{\psi}$. 
Also, we take $\eta=\lambda=0.2,~0.3,~\mbox{or}~0.4$ and assume $\|O\|=1$.
The green lines represent the results of the coherent filter in Algorithm~\ref{alg:coherentfilter}, and the purple lines represent the results of the hybrid filter in Algorithm~\ref{alg:hybridfilter}.
The dotted lines in the left figures of Fig.~\ref{fig:numericalsim_VD} show the $25\%$, $50\%$, and $95\%$ percentiles for the hybrid filter.
Although the number of state copies in the coherent filter also varies probabilistically, we have found that its probability distribution is too narrow to plot the percentiles, and thus we have omitted them.
For comparison, we numerically calculate the performance of the previous work for quantum error mitigation (VD)~\cite{Huggins2021-vd, koczor2021-esd} and plot it with red lines.
In the case of VD, the approximation error $\Delta$ is given by $2Q_l/(1+Q_l)$ for the number $l$ of copies per circuit where
$Q_l:=(\lambda/(1-\lambda))^l {\rm tr}[\rho_{\rm err}^l]$.
The measurement overhead is given by ${\rm tr}[\rho^l]^{-2}$. 
Moreover, we plot the results of the original qPCA with the conventional DME procedure with blue lines.

From Fig.~\ref{fig:numericalsim_VD}, we find several advantageous features of our methods.
The coherent filter (Algorithm~\ref{alg:coherentfilter}) achieves the constant measurement overhead by coherently consuming a modest number of the state copies, which is exponentially smaller than that of the original qPCA.
This significant improvement arises from the exponential improvement of the virtual DME compared to the conventional DME.
Also, our method is in contrast to the previous work for VD, which also keeps the number of state copies per circuit small but suffers from an exponentially large measurement overhead. 
This difference can be understood in terms of the remaining signal intensity at the final measurement. 
In the coherent filter, the input state $\rho$ is effectively projected into the rescaled pure state $(1-\lambda)\ketbra{\psi}$.
So, the intrinsic signal intensity $(1-\lambda)$ is still preserved in the coherent filter, while the previous VD approach suppresses both the signal and noise by using the power of the target state $\rho^l$. 
This leads to the exponential separation in the measurement overhead, as observed in the right panels of Fig.~\ref{fig:numericalsim_VD}. 
On the other hand, the hybrid filter (Algorithm~\ref{alg:hybridfilter}) combines the advantages of both the coherent filter and VD. 
As a result, its measurement overhead is much smaller than VD (as observed in the right panels) by coherently consuming a comparable number of state copies in average (as observed in the left panels).
Further improvements of our procedure by combining it with previous methods are expected, and we leave it as an interesting future direction.



\begin{figure}[htbp]
  \centering

  \begin{minipage}{\textwidth}
    \centering
    \textbf{(a) $\lambda=\eta=0.20$ and arbitrary pure $\rho_{\rm err}$}\\
    \includegraphics[width=0.95\textwidth]{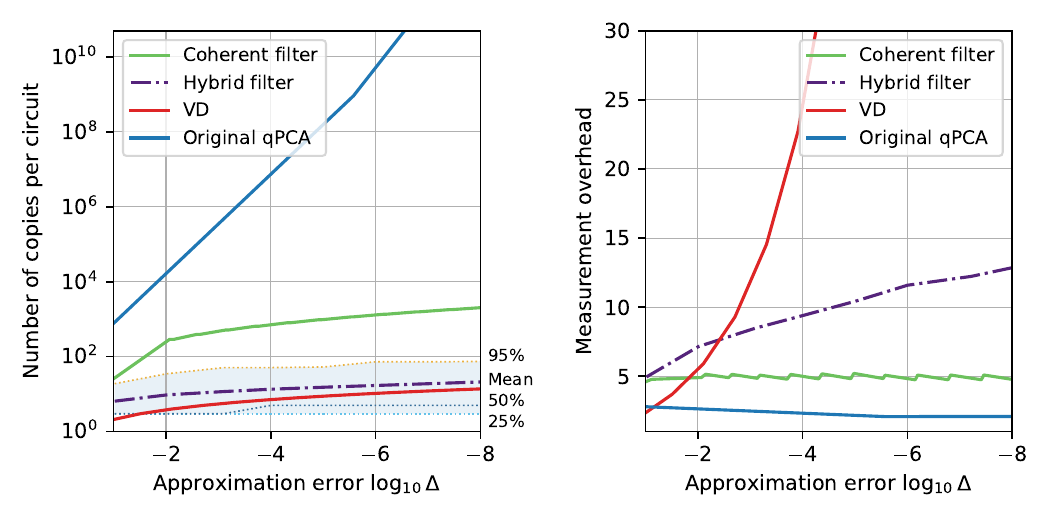}
  \end{minipage}

  \vspace{1em} 

  \begin{minipage}{\textwidth}
    \centering
    \textbf{(b) $\lambda=\eta=0.40$ and arbitrary pure $\rho_{\rm err}$}\\
    \includegraphics[width=0.95\textwidth]{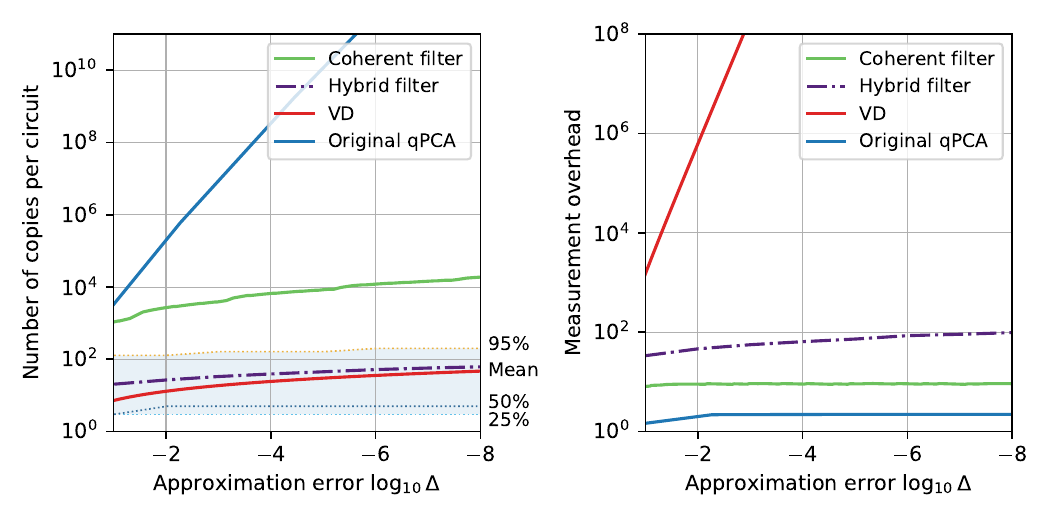}
  \end{minipage}

  \caption{Number of state copies and measurement overhead versus the approximation error $\log_{10}\Delta$ for computing the target mean value $\bra{\psi}O\ket{\psi}$ for any observable $O$ ($\|O\|=1$), where $\rho=(1-\lambda)\ketbra{\psi}+\lambda\rho_{\rm err}$.
  The green and purple lines represent the results based on the coherent filter (Algorithm~\ref{alg:coherentfilter}) and the hybrid filter (Algorithm~\ref{alg:hybridfilter}), respectively. 
  The red and blue lines represent the results of virtual distillation (VD)~\cite{Huggins2021-vd, koczor2021-esd} and the original protocol for qPCA~\cite{Lloyd2014-nn}, respectively.
  Since the number of state copies required for the hybrid filter varies probabilistically, we plot the $25\%$, $50\%$, and $95\%$ percentiles by dotted lines in the left panels.}
  \label{fig:numericalsim_VD}
\end{figure}

\clearpage
\section{Theoretical details in applications}\label{sec:proofs}

\subsection{Detailed description on universal quantum emulator}
\label{apsec:theoretical_detail_of_uqe}

The main components of the UQE algorithm proposed in Ref.~\cite{marvian2024universalquantumemulator} are reflection operators about the given input–output pairs of quantum states, which are approximated using the DME method~\cite{Lloyd2014-nn}. Here, by replacing these reflection operators with approximations based on our virtual DME method, we show that the resulting algorithm satisfies Theorem~\ref{thm:universal_quantum_emulator}. 
To this end, we first review the original UQE algorithm introduced in Ref.~\cite{marvian2024universalquantumemulator}, including its problem setup, procedure, and associated complexities. Then, we prove that our modified algorithm meets the conditions stated in Theorem~\ref{thm:universal_quantum_emulator}.

\paragraph{Problem Setup}
Let $U$ be an unknown unitary operator on $\mathbb{C}^d$, and define
\begin{equation}
    S_{\rm in} := \left\{  \ket{\phi_i^{\rm in}} :i=1,\cdots,K\right\}~~~\text{and}~~~ S_{\rm out} := \left\{\ket{\phi_i^{\rm out}}:=U\ket{\phi_i^{\rm in}}:i=1,\cdots,K\right\}
\end{equation}
as the sets of input and output quantum states, respectively. 
Let $\mathcal{H}_{\rm in}$ denote the $d_{\rm eff}$-dimensional subspace spanned by the input states. 
The UQE algorithm takes as input as an arbitrary unknown state $\ket{\psi}\in \mathcal{H}_{\rm in}$ and produce a state $\rho_{\psi}$ that satisfies $\| \rho_{\psi} - U\ketbra{\psi}U^{\dagger} \|_{1}\leq \varepsilon$, by coherently coupling the state $\ket{\psi}$ with sufficiently many input-output pairs sampled randomly from $S_{\rm in}$ and $S_{\rm out}$. 
In order to emulate the action of $U$ on the input subspace $\mathcal{H}_{\rm in}$, it is necessary to impose certain assumptions on $S_{\rm in}$. 
For instance, consider the case of $U=e^{i\theta Z}$ and $S_{\rm in}=\{\ket{0},\ket{1}\}$. 
In this case, the parameter $\theta$ only appears as a global phase in the output states, which cannot be inferred from the input-output data. 
To avoid such issue, we require $S_{\rm in}$ to satisfy the following condition:
\begin{equation}\label{apeq:covergence_condition}
    \mathrm{Alg}_{\mathbb{C}}(S_{\rm in}) = \mathcal{L}(\mathcal{H}_{\rm in}),
\end{equation}
where $\mathrm{Alg}_{\mathbb{C}}(S_{\rm in})$ denotes the complex associative algebra generated by $S_{\rm in}$ and $\mathcal{L}(\mathcal{H}_{\rm in})$ denotes the set of all linear operators supported on $\mathcal{H}_{\rm in}$. As shown in Proposition~1 in \cite{marvian2024universalquantumemulator}, this is the necessary and sufficient condition for the set $S_{\rm in}$ to uniquely determine the action of $U$ on $\mathcal{H}_{\rm in}$.


\paragraph{Procedures of original UQE algorithm~\cite{marvian2024universalquantumemulator}}

Here, we provide a brief overview of the main steps of the algorithm, while the sample complexity (namely, the total number of necessary input-output pairs) and the gate complexity of the algorithm will be discussed in detail later (Lemma~\ref{lemma:uqe_with_simulated_reflection}).
The main building block of the UQE algorithm is the controlled-reflection operator with respect to an input or output state $\ket{\phi}$, defined as
\begin{equation}
    \mathrm{c}R_{\rm c,t}(\phi):= \ketbra{0}_{\rm c} \otimes \bm{1}_{\rm t} + \ketbra{1}_{\rm c} \otimes e^{-i\pi\ket{\phi}\bra{\phi}_{\rm t}},
\end{equation}
where the subscripts $\rm c$ and $\rm t$ denote the control and target system respectively. 
In the original UQE algorithm, the controlled reflections are simulated using the DME technique~\cite{marvian2024universalquantumemulator}, which is realized by entangling the target system with multiple copies of the state $\ket{\phi}$. 
For simplicity, however, the following description of the procedure assumes ideal implementations of $\mathrm{c}R_{\rm c,t}(\phi)$.

As illustrated in the quantum circuit shown in Fig.~\ref{Fig:universal_quantum_emulator}, the algorithm starts by initializing $T$ ancilla qubits labeled by $\mathrm{a}_{1},\cdots,\mathrm{a}_{T}$ in the state $\ket{-}:=(\ket{0}+\ket{1})/\sqrt{2}$ and preparing the state $\ket{\psi}$ in the main system denoted by $\rm m$. 
Then, the algorithm consists of the following four steps.

\vspace{0.3 cm}

\noindent \textbf{Step 1.} In the first step, we are trying to coherently erase the initial state $\ket{\psi}_{\rm m}$ on the main system and drive it toward a fixed input state $\ket{\phi_{1}^{\rm in}}$. 
To this end, we repeat the following procedure for each $k=1,...,T$: we sample an integer $i_k$ from the set $\{2,\cdots,K\}$ at uniformly random, and apply a unitary gate $W_{\mathrm{a}_k,\mathrm{m}}(\phi_{i_k}^{\rm in})$ defined by
\begin{equation}\label{eq:def_of_W_operator}
    W_{\mathrm{a}_k,\mathrm{m}}(\phi_{i_k}^{\rm in}) := \mathrm{c}R_{\mathrm{a}_k,\mathrm{m}}(\phi_{i_k}^{\rm in}) \left\{ H_{\mathrm{a}_k} \otimes \mathbf{1}_{\rm m} \right\} \mathrm{c}R_{\mathrm{a}_k,\mathrm{m}}(\phi_{1}^{\rm in}),
\end{equation}
 to the $k$-th ancilla qubit $\mathrm{a}_k$ and the main system $\rm m$. 
 Here, let $\mathcal{W}$ be the CPTP map obtained by tracing out the qubit ${\rm a}_{k}$ of $W_{\mathrm{a}_k,\mathrm{m}}$ and taking the expectation over $i_k$. Then, if the relation~\eqref{apeq:covergence_condition} holds, the randomized channel $\mathcal{W}$ has a unique fixed point $\ket{\phi_1^{\rm in}}\bra{\phi_1^{\rm in}}$, and $\mathcal{W}^{T}(\ket{\psi}\bra{\psi})$ converges to this fixed point, i.e., $\lim_{T\rightarrow \infty} \mathcal{W}^{T}(\ket{\psi}\bra{\psi})=\ket{\phi_1^{\rm in}}\bra{\phi_1^{\rm in}}$; see Proposition 2 in Ref.~\cite{marvian2024universalquantumemulator}.
 Moreover, since the global state at the end of this step is a pure state\footnote{This is because the initial state of the algorithm is a pure state and the entire process consists of unitary transformations}, the resulting reduced state on the ancilla system is also close to pure, implying the preservation of information of the initial state $\ket{\psi}$ due to the unitarity of the entire process. 
 Thus, letting $\ket{\omega(\bm{i})}_{a_{1}\cdots a_{T}}$ denote the reduced state on the ancilla system when $T$ is large, the global state at the end of this step can be approximately written as
 \begin{equation}\label{apeq:emulator_step_1}
     \left[ W_{\mathrm{a}_T,\mathrm{m}}(\phi_{i_{T}}^{\rm in}) \cdots W_{\mathrm{a}_{1},\mathrm{m}}(\phi_{i_1}^{\rm in}) \right] \ket{\psi}_{\rm m} \ket{-}_{\mathrm{\bm{a}}}^{\otimes T} \approx \ket{\phi^{\rm in}_{1}}_{\rm m} \ket{\omega(\bm{i})}_{\bm{\mathrm{a}}}
 \end{equation}

\vspace{0.3 cm}

\noindent \textbf{Step 2.} Next, to verify whether \textbf{Step 1} has succeeded (that is, whether the initial state $\ket{\psi}$ has been projected into the state $\ket{\phi^{\rm in}_{1}}$), we perform a controlled-reflection $\mathrm{c}R_{\mathrm{b}_1,\mathrm{m}}(\phi_1^{\rm in})$ on $\ket{\phi_{1}^{\rm in}}$, where the control qubit is a newly prepared ancilla qubit $\mathrm{b}_{1}$ in the state $\ket{-}$. 
Next, we perform the Hadamard gate on a qubit $\rm b_1$, followed by the computational basis measurement. 
Then, the measurement outcome $0$ indicates that the projection has succeeded; otherwise, the projection has failed. 
At this point, we can choose one of two approaches: either we ignore the measurement outcome 
or continue only if the outcome is 0.

\vspace{0.3 cm}

\noindent \textbf{Step 3.} At this stage, we replace the 
reduced state of the main system which is pushed into
$\ket{\phi_1^{\rm in}}$ with its corresponding output state $\ket{\phi^{\rm out}_{1}}=U\ket{\phi_{1}^{\rm in}}$. 
This can be implemented by applying the SWAP operation between the main system and the ancillary system $\mathrm{b}_2$ in the state $\ket{\phi^{\rm out}_{1}}$ as
\begin{equation}
    \mathrm{SWAP}_{\mathrm{m},\mathrm{b}_2} \ket{\phi^{\rm in}_{1}}_{\rm m} \ket{\omega(\bm{i})}_{\bm{\mathrm{a}}} \ket{\phi_{1}^{\rm out}}_{\rm b_2} = \ket{\phi^{\rm out}_{1}}_{\rm m} \ket{\omega(\bm{i})}_{\bm{\mathrm{a}}} \ket{\phi_{1}^{\rm in}}_{\rm b_2},
\end{equation}
followed by discarding the state in the ancilla system.

\vspace{0.3 cm}

\noindent \textbf{Step 4.} In the final step, we apply a sequence of controlled reflections corresponding to the output states, in the reverse order of \textbf{Step 1}. Specifically, for each $k=T,...,1$, we apply the following operation to the ancilla qubit $a_k$ and the main system:
\begin{equation}
    W^{\dagger}_{\mathrm{a}_k,m}(\phi^{\rm out}_{i_k}) = \mathrm{c}R^{\dagger}_{\mathrm{a}_k,m}(\phi_{1}^{\rm out}) \left\{H_{\mathrm{a}_k} \otimes \mathbf{1}_{\rm m}\right\} \mathrm{c}R^{\dagger}_{\mathrm{a}_k,m}(\phi_{i_k}^{\rm out}).
\end{equation}
After the above procedure, we have
\begin{equation}\label{apeq:emulator_step_4}
    \left[ W^{\dagger}_{\mathrm{a}_1,\mathrm{m}}(\phi_{i_{1}}^{\rm out}) \cdots W^{\dagger}_{\mathrm{a}_{T},\mathrm{m}}(\phi_{i_T}^{\rm out}) \right]  \ket{\phi^{\rm out}_{1}}_{\rm m} \ket{\omega(\bm{i})}_{\bm{\mathrm{a}}} \approx (U \ket{\psi})_{\rm m}\ket{-}^{\otimes T}.
\end{equation}
In this way, we have emulated the action of the unknown unitary $U$ on a given input state $\ket{\psi}\in \mathcal{H}_{\rm in}$. Note that Eq.~\eqref{apeq:emulator_step_4} can be confirmed by the following calculation:
\begin{align}
    &\left[ W^{\dagger}_{\mathrm{a}_1,\mathrm{m}}(\phi_{i_{1}}^{\rm out}) \cdots W^{\dagger}_{\mathrm{a}_{T},\mathrm{m}}(\phi_{i_T}^{\rm out}) \right]  \ket{\phi^{\rm out}_{1}}_{\rm m} \ket{\omega(\bm{i})}_{\bm{\mathrm{a}}}\nonumber\\
    &~~~~~~~~~~~~= \left[ W^{\dagger}_{\mathrm{a}_1,\mathrm{m}}(\phi_{i_{1}}^{\rm out}) \cdots W^{\dagger}_{\mathrm{a}_{T},\mathrm{m}}(\phi_{i_T}^{\rm out}) \right] U_{\rm m} \ket{\phi^{\rm in}_{1}}_{\rm m} \ket{\omega(\bm{i})}_{\bm{\mathrm{a}}} \\
    &~~~~~~~~~~~~\approx \left[ W^{\dagger}_{\mathrm{a}_1,\mathrm{m}}(\phi_{i_{1}}^{\rm out}) \cdots W^{\dagger}_{\mathrm{a}_{T},\mathrm{m}}(\phi_{i_T}^{\rm out}) \right] U_{\rm m} \left[ W_{\mathrm{a}_T,\mathrm{m}}(\phi_{i_{T}}^{\rm in}) \cdots W_{\mathrm{a}_{1},\mathrm{m}}(\phi_{i_1}^{\rm in}) \right] \ket{\psi}_{\rm m} \ket{-}_{\mathrm{\bm{a}}}^{\otimes T}  \\
    &~~~~~~~~~~~~=(U \ket{\psi})_{\rm m}\ket{-}^{\otimes T}, 
\end{align}
where the second-to-third line uses Eq.~(\ref{apeq:emulator_step_1}) and the final line follows from the fact that $W^{\dagger}_{\mathrm{a}_k,\mathrm{m}}(\phi^{\rm out}_{i_k}) =(\bm{1}_{\mathrm{a}_k}\otimes U_{\rm m}) W^{\dagger}_{\mathrm{a}_k,\mathrm{m}}(\phi^{\rm in}_{i_k}) (\bm{1}_{\mathrm{a}_k}\otimes U_{\rm m}^{\dagger})$ and thereby
\begin{equation}
W^{\dagger}_{\mathrm{a}_1,\mathrm{m}}(\phi^{\rm out}_{i_1}) \cdots W^{\dagger}_{\mathrm{a}_T,\mathrm{m}}(\phi^{\rm out}_{i_T}) =(\bm{1}_{\mathrm{\bm{a}}}\otimes U_{\rm m}) W^{\dagger}_{\mathrm{a}_1,\mathrm{m}}(\phi^{\rm in}_{i_1}) \cdots W^{\dagger}_{\mathrm{a}_T,\mathrm{m}}(\phi^{\rm in}_{i_T}) (\bm{1}_{\mathrm{\bm{a}}}\otimes U_{\rm m}^{\dagger}).
\end{equation}

The approximation accuracy of the resulting main-system state 
with respect to $U\ket{\psi}$ depends on the success probability of the projection onto $\ket{\phi_{1}^{\rm in}}$ in \textbf{Step 1}, which is equal to the probability of observing 0 in \textbf{Step 2}, denoted by $p_0$. 
More precisely, let $\mathcal{E}_{U}$ denote the quantum channel that represents the quantum circuit in Fig.~\ref{Fig:universal_quantum_emulator} when no post-selection is performed and all controlled-reflection operations are implemented without error; then Ref.~\cite{marvian2024universalquantumemulator} showed that 
\begin{equation}    F\left(\mathcal{E}_U(|\psi\rangle\langle\psi|),U|\psi\rangle\langle\psi|U^\dagger\right) \geq p_{0},
\end{equation}
where $F$ denotes the Uhlmann fidelity~\cite{nielsen2010quantum}. 
We remark that if post-selection is performed (i.e., the protocol proceeds only when the outcome of Step 2 is 0), then the fidelity of the resulting main-system state is lower-bounded by $\sqrt{p_0}$.

\paragraph{Complexities of original UQE algorithm~\cite{marvian2024universalquantumemulator}}

We present the error analysis provided by \cite{marvian2024universalquantumemulator}, concerning the sample complexity and gate complexity of this algorithm. 
We first focus on the case where the controlled-reflection operations are ideal and the measurement outcome in \textbf{Step 2} is ignored. 

\begin{lem}[UQE with ideal controlled-reflection~\cite{marvian2024universalquantumemulator}]\label{lemma:uqe_with_ideal_controlled_reflection}
Let $\mathcal{E}_{U}$ denote the quantum channel that captures the overall effect of the quantum circuit in Fig.~\ref{Fig:universal_quantum_emulator}, assuming that all controlled-reflections are ideal and that no post-selection is performed. 
Then,
%
\begin{equation}
    T \geq \frac{\mathrm{ln}(2\varepsilon_{\rm id}^{-2}\sqrt{d_{\rm eff}-1})}{1-\lambda_{\perp}} 
\end{equation}
is suffice to achieve the trace distance $\frac{1}{2}\| \mathcal{E}_{U}(\ketbra{\psi}) - U\ketbra{\psi}U^{\dagger} \|_{1}\leq \varepsilon_{\rm id}$ with $\varepsilon_{\rm id}\in(0,1]$ and $\ket{\psi}\in \mathcal{H}_{\rm in}$.
Here, $\lambda_{\perp}$ denotes the largest eigenvalue of the operator 
\begin{equation}
    \frac{1}{K-1} \sum_{k=2}^{K} A_k \otimes A_k^{\ast}, 
    \quad \quad A_k = P_{\perp} - 2P_{\perp} \ket{\phi_k^{\rm in}} \bra{\phi_k^{\rm in}}P_{\perp},
\end{equation}
where $P_{\perp}$ is the projector onto the subspace $\mathcal{H}_{\rm in}$ orthogonal to $\ket{\phi_1^{\rm in}}$. 
\end{lem}

The proof can be found in Appendix~A of Ref.~\cite{marvian2024universalquantumemulator}. As pointed out in Ref.~\cite{marvian2024universalquantumemulator}, $T$ depends on the convergence rate toward the reference state $\ket{\phi_{1}^{\rm in}}$, which is governed by the spectral parameter $\lambda_{\perp}$. 
Importantly, it is clear from the definition of $\lambda_{\perp}$ that the value is independent of the approximation error parameter $\varepsilon_{\rm id}$. 

The above analysis assumes ideal implementation of the controlled-reflection operators in Fig.~\ref{Fig:universal_quantum_emulator}. 
In practice, however, the UQE approximates these operators $\mathrm{c}R(\phi)$ using multiple copies of the state $\ket{\phi}$ via the DME technique~\cite{Lloyd2014-nn}\footnote{The ideal implementation of $\mathrm{c}R(\phi)$ requires implementing the unitary $U$ that generates $\ket{\phi}$, but in the present problem setting, we do not have direct access to $U$ in general.}. 
When approximation errors due to this simulation are taken into account, the total number of copies of input-output states and the total number of elementary gates required to implement the algorithm are given as follows:

\begin{lem}[UQE with approximate controlled-reflection~\cite{marvian2024universalquantumemulator}]
\label{lemma:uqe_with_simulated_reflection}
Let $\mathcal{E}_{U}^{\rm DME}$ denote the quantum channel in Fig.~\ref{Fig:universal_quantum_emulator} where the ideal controlled-reflection operations are replaced with approximated versions via the density matrix exponentiation~\cite{Lloyd2014-nn}. 
Then, for any $\ket{\psi}\in \mathcal{H}_{\rm in}$ and $d_{\rm eff}={\rm dim}\mathcal{H}_{\rm in}$,
$\frac{1}{2}\left\| \mathcal{E}^{\rm DME}_{U}(\ket{\psi}\bra{\psi}) - U\ket{\psi}\bra{\psi}U^{\dagger} \right\|_{1} \leq \varepsilon$ can be achieved using the following resources:
\begin{itemize}
  \item $T=\mathcal{O}\left((1-\lambda_{\perp})^{-1}\log(d_{\rm eff}/\varepsilon^2)\right)$, where $\lambda_{\perp}$ is defined as in Lemma~\ref{lemma:uqe_with_ideal_controlled_reflection},
  \item Total number of input-output state copies: $ \mathcal
  {O}(\varepsilon^{-1}T^2)$,
  \item Total number of one- and two-qubit elementary gates: 
  $\mathcal{O}(\varepsilon^{-1}T^2 \log d)$.
\end{itemize}
\end{lem}

\noindent
While the detailed proof is provided in Appendix~B in Ref.~\cite{marvian2024universalquantumemulator}, we here give a proof sketch below for completeness.

\begin{proof}
As shown in~\cite{Lloyd2014-nn} or Section~\ref{supple_sec:non_physical_ope}, the controlled-reflection operator $\mathrm{c}R(\phi)$ can be approximated by a quantum channel $\mathcal{E}_{\rm DME}(\phi)$ simulated via DME, such that 
\begin{equation}
\| \mathcal{E}_{\rm DME}(\phi) - \mathrm{c}R(\phi) \|_{\diamond}
    \leq\varepsilon_{\rm DME}.
\end{equation}
This implementation requires $\mathcal{O}(\varepsilon_{\rm DME}^{-1})$ copies of $\ket{\phi}$ and has a gate complexity of $\mathcal{O}(\varepsilon^{-1}_{\rm DME}\log d)$.

Let $\mathcal{E}_{U}^{\rm DME}$ denote the quantum channel obtained by replacing all $4T+1$ controlled reflections in $\mathcal{E}_{U}$ with their simulated versions via DME; then, triangle inequality for the diamond norm implies
\begin{eqnarray}
    \frac{1}{2}\| \mathcal{E}_{U}^{\rm DME} (\ket{\psi}\bra{\psi}) - U\ket{\psi}\bra{\psi}U^{\dagger}\|_{1} 
    &\leq& \frac{1}{2}\| \mathcal{E}_{U}^{\rm DME}(\ket{\psi}\bra{\psi}) - \mathcal{E}_{U}(\ket{\psi}\bra{\psi}) \|_{1} \nonumber\\
       && + \frac{1}{2}\| \mathcal{E}_{U}(\ket{\psi}\bra{\psi}) - U\ket{\psi}\bra{\psi}U^{\dagger} \|_{1} \nonumber\\
    &\leq& (4T+1) \varepsilon_{\rm DME} + \varepsilon_{\rm id}.
\end{eqnarray}
In order to ensure that the total error is bounded by $\varepsilon$, it suffices to set $\varepsilon_{\rm DME}=\mathcal{O}(\varepsilon/T)$ and $\varepsilon_{\rm id}=\mathcal{O}(\varepsilon)$. Thus, $\mathcal{E}_{U}^{\rm DME}$ which satisfies $\frac{1}{2}\left\| \mathcal{E}^{\rm DME}_{U}(\ketbra{\psi}) - U\ketbra{\psi}U^{\dagger} \right\|_{1} \leq \varepsilon$ can be implemented using 
\begin{equation}
    N_{\rm tot}=
    (4T+1) \times \mathcal
    {O}\left(\frac{1}{\varepsilon_{\rm DME}}\right) = \mathcal{O}\left(\frac{T^2}{\varepsilon}\right)
\end{equation}
copies of the sample states, and the total gate complexity is
\begin{equation}
    N_{\rm tot}\times \mathcal
    {O}(\log{d})=\mathcal{O}\left(\frac{T^2 \log{d}}{\varepsilon}\right).
\end{equation}
\end{proof}




\paragraph{Proof of Theorem~\ref{thm:universal_quantum_emulator}}

As shown in Lemma~\ref{lemma:uqe_with_simulated_reflection}, with respect to $\varepsilon$, the existing UQE algorithm requires $N_{\rm tot}=\mathcal{O}(\varepsilon^{-1})$ copies of the input-output pairs and $t_{\rm tot}=\mathcal{O}(\varepsilon^{-1})$ elementary gates per emulation circuit, resulting in a substantial overhead in practical implementations, particularly for small $\varepsilon$. Remarkably, we show that if the task is to estimate properties of the emulated state $U\ket{\psi}$, both $N_{\rm tot}$ and $t_{\rm tot}$ can be exponentially reduced to $\mathcal{O}(\log({1/\varepsilon}))$ without affecting the scaling with respect to other parameters, $d_{\rm eff}$ and $d$. 

\begin{proof}[Proof of Theorem~\ref{thm:universal_quantum_emulator}]
To prove the theorem, we explicitly provide an algorithm that satisfies the stated conditions. 
Before delving into the details, we define the following quantum channel, corresponding to a part of the circuit shown in Fig.~\ref{Fig:universal_quantum_emulator} (which excludes state preparation and measurement of Fig.~\ref{Fig:universal_quantum_emulator}) for a fixed set of indices $\bm{i}:=(i_1,...,i_{T})$:
\begin{align}
    \mathcal{E}_{\rm circ}^{[\bm{i}]}
    := \underbrace{\mathcal{U}[W^{\dagger}(\phi_{i_{1}}^{\rm out})] \circ \cdots \circ \mathcal{U}[W^{\dagger}(\phi_{i_{T}}^{\rm out})]}_{\text{Step 4}} \circ \underbrace{\mathcal{U}[S]}_{\text{Step 3}} \circ \underbrace{\mathcal{U}[H_{b_1}] \circ \mathcal{U}[\mathrm{c}R(\phi_{1}^{\rm in})]}_{\text{Step 2}} \circ \underbrace{\mathcal{U}[W(\phi_{i_{T}}^{\rm in})] \circ \cdots \circ \mathcal{U}[W(\phi_{i_{1}}^{\rm in})]}_{\text{Step 1}}.\nonumber
\end{align}
Here, $S$ in Step 3 denotes the SWAP operator between the systems $\rm m$ and $\mathrm{b}_2$.
Note that the operators in Step~1 and Step~4 consist of controlled-reflection operations and Hadamard gates:
\begin{align}
    \mathcal{U}[W(\phi_{i_{k}}^{\rm in})] &= \mathcal{U}[\mathrm{c}R(\phi_{i_{k}}^{\rm in})] \circ \mathcal{U}[H_{\mathrm{a}_k}] \circ \mathcal{U}[\mathrm{c}R(\phi_{1}^{\rm in})], \\
    \mathcal{U}[W^{\dagger}(\phi_{i_{k}}^{\rm out})] &= \mathcal{U}[\mathrm{c}R(\phi_{1}^{\rm out})] \circ \mathcal{U}[H_{\mathrm{a}_k}] \circ \mathcal{U}[\mathrm{c}R(\phi_{i_k}^{\rm out})].
\end{align}
In addition, we define the following function for $\bm{i}:=(i_1,...,i_{T})$ and $M:=4T+1$:
\begin{align}
     &f_{M}^{[\bm{i}]}\left(\left\{\ketbra{\phi_{i}^{\rm in}},\ketbra{\phi_{i}^{\rm out}}\right\}\right)\notag\\
     &:= \mathrm{tr}\left[ \left(O_{\rm m} \otimes \bm{1}_{\mathrm{\bm{a,b}}}\right) \mathcal{E}_{\rm circ}^{[\bm{i}]}\left( \ket{-}\bra{-}^{\otimes T}_{\mathrm{\bm{a}}} \otimes \ket{-}\bra{-}_{\mathrm{b}_1} \otimes \ket{\phi_{1}^{\rm out}} \bra{\phi^{\rm out}_{1}}_{\mathrm{b}_2} \otimes \rho \right) \right].
\end{align}
Note that the unitary circuit $\mathcal{E}^{[\bm{i}]}_{\rm circ}$ implicitly depends on $\ket{\phi_i^{\rm in}},\ket{\phi_i^{\rm out}}$. 
To achieve the condition in the theorem statement, we perform the following procedure: for each shot $N$, we randomly choose the index $\bm{i}$, and apply the general algorithm described in Section~\ref{sec:nearlyopt_qalg} to the corresponding function $f_{M}^{[\bm{i}]}$. 
Specifically, we execute the following algorithm.

\vspace{0.2cm}

\noindent For $j=1,\ldots,N$:
\begin{enumerate}
    \item Sample $\bm{i}=(i_1,...,i_{T})$ uniformly at random from $\{2,...,K\}^{T}$.
    \item Replace all $M$ controlled-reflection operators in $\mathcal{E}_{\rm circ}^{[\bm{i}]}$ with $\hat{\Phi}_{\mathrm{pure},r}$, as defined in Theorem~\ref{thm:DME_pure_state}. Denote the resulting channel by $ \mathcal{E}_{\hat{\Phi}}^{[\bm{i}]}$.
    \item Execute the circuit corresponding to 
    \begin{equation}\label{eq:quantum_circuit_random}
        \mathrm{tr}\left[ \left(O_{\rm m} \otimes \bm{1}_{\mathrm{\bm{a,b}}}\right) \mathcal{E}_{\hat{\Phi}}^{[\bm{i}]}\left( \ket{-}\bra{-}^{\otimes T}_{\mathrm{\bm{a}}} \otimes \ket{-}\bra{-}_{\mathrm{b}_1} \otimes \ket{\phi_{1}^{\rm out}} \bra{\phi^{\rm out}_{1}}_{\mathrm{b}_2} \otimes \rho \right) \right],
    \end{equation}
    and denote the measurement outcome by $\hat{\mu}_{j}$. Multiply this outcome by the weight factor $(C^{2r}_{\rm pure})^{M}$.
\end{enumerate}
After $N$ iterations, we compute the final estimator as $\hat\mu=\frac{1}{N}\sum_{j=1}^{N}\hat{\mu}_{j}C_{\rm pure}^{2rM}$.

We now verify that the above procedure satisfies the theorem's requirement.
By Chebyshev's inequality, setting $N=\mathcal{O}(\bar\varepsilon^{-2}{C_{\rm pure}^{4rM}}\|O\|^2)$ guarantees that $| \hat{\mu}-\mathbb{E}[\hat{\mu}] |\leq \bar\varepsilon$ with high probability. Moreover, since
\begin{align}
    \mathbb{E}[\hat{\mu}] 
    &= \sum_{\bm{i}} p(\bm{i}) C_{\rm pure}^{2rM} \mathbb{E}\left[ \mathrm{tr}\left[ \left(O_{\rm m} \otimes \bm{1}_{\mathrm{\bm{a,b}}}\right) \mathcal{E}_{\hat{\Phi}}^{[\bm{i}]}\left( \ket{-}\bra{-}^{\otimes T}_{\mathrm{\bm{a}}} \otimes \ket{-}\bra{-}_{\mathrm{b}_1} \otimes \ket{\phi_{1}^{\rm out}} \bra{\phi^{\rm out}_{1}}_{\mathrm{b}_2} \otimes \rho \right) \right] \right]\\
    &= \sum_{\bm{i}} p(\bm{i}) \mathrm{tr}\left[ \left(O_{\rm m} \otimes \bm{1}_{\mathrm{\bm{a,b}}}\right) \mathcal{E}_{\rm circ}^{[\bm{i}]}\left( \ket{-}\bra{-}^{\otimes T}_{\mathrm{\bm{a}}} \otimes \ket{-}\bra{-}_{\mathrm{b}_1} \otimes \ket{\phi_{1}^{\rm out}} \bra{\phi^{\rm out}_{1}}_{\mathrm{b}_2} \otimes \rho \right) \right] \\
    &= \sum_{\bm{i}} p(\bm{i}) f_{M}^{[\bm{i}]} = \mathrm{tr}[O \mathcal{E}_{U}(\rho)]
\end{align}
where $p(\bm{i})$ is the probability distribution of $\bm{i}$ and the second equality uses Theorem~\ref{thm:DME_pure_state}, we obtain $| \hat{\mu}-\mathrm{tr}[O \mathcal{E}_{U}(\rho)] |\leq \bar\varepsilon$ with high probability.

Finally, Lemma~\ref{lemma:uqe_with_ideal_controlled_reflection} guarantees that $|\mathrm{tr}[O \mathcal{E}_{U}(\rho)-OU\rho U^{\dagger}] |\leq 2\varepsilon_{\rm id}\|O\|$ provided that $T=\mathcal{O}((1-\lambda_{\perp})^{-1}\log( d_{\rm eff}/{\varepsilon}_{\rm id}^2))$. By choosing $\bar{\varepsilon} + 2\|O\|\varepsilon_{\rm id} \leq \varepsilon$, we arrive at the final error bound 
\begin{equation}
| \hat{\mu}-\mathrm{tr}[OU\rho U^{\dagger}] |\leq \varepsilon
\end{equation}
Lastly, to ensure that the scaling factor $C^{4rM}_{\rm pure}\leq e^{\mathcal{O}(M/r)}$ remains constant, it is sufficient to take $r=\mathcal{O}(M)$, and thus each random circuit in Eq.~\eqref{eq:quantum_circuit_random} uses at most: 
\begin{itemize}
    \item $\mathcal{O}(M^2)=\mathcal{O}(T^2)$ copies of input-output sample states: 
    \item $\mathcal{O}(M^2 \log d)=\mathcal{O}(T^2\log d)$ of one- and two-qubit elementary gates.
\end{itemize}
\end{proof}

\subsection{Proof of Theorem~\ref{thm:entropy-estimation} and Corollary~\ref{cor:relative-entropy-estimation}}
\label{sec:proof-entropy}

\begin{proof}
[Proof of Theorem~\ref{thm:entropy-estimation} and Corollary~\ref{cor:relative-entropy-estimation}]
    From Ref.~\cite[Corollary 71]{gilyen2019-qsvt}, given access to the controlled time-evolution $e^{i\rho/2}$, we can construct a $(\pi, 2, \varepsilon_{\rho})$-block-encoding of $\rho$, where $\varepsilon_{\rho} \in (0, 1/2]$ is a block encoding error. 
    This block encoding requires $\mathcal{O}(\log(1/\varepsilon_{\rho}))$ uses of the controlled-$e^{i\rho/2}$ and its inverse, and $\mathcal{O}(\log(1/\varepsilon_{\rho}))$ uses of one- or two-qubit gates.
    
    Next, from Ref.~\cite[Lemma 11]{gilyen2019-world}, 
    there exists an even real polynomial $P(x)$ of degree
    $n_P = \mathcal{O}\left(\delta^{-1} \log\left(1/\varepsilon_P\right)\right)$,
    which provides an approximation of a target function $f(x)\propto \ln(1/x)$ such that
    \begin{gather}\label{eq:log-polynomial}
        \left| P(x) - f(x) \right|  = \left| P(x) - \frac{\ln(1/x)}{2\ln (2/\delta)} \right| \leq \varepsilon_P~~~\forall x\in {[\delta, 1]}\\
        \left|P(x)\right|\le 1~~~\forall x\in[-1,1],
    \end{gather}
    where $\delta, \varepsilon_P >0$.
    Then, applying the eigenvalue transformation determined by $P(x)$ to the block encoding of $\rho$, we can construct $(1, 3,  4n_P \sqrt{\varepsilon_\rho})$-block encoding of $P(\rho / \pi) \approx f(\rho/\pi)$; see Ref.~\cite[Lemma 22]{gilyen2019-qsvt}.
    We can see that $\mathrm{tr}[\rho P(\rho/\pi)]$ is used to approximately calculate $S(\rho)$ as follows;
    \begin{equation}
        S(\rho) = -\mathrm{tr}(\rho \ln \rho)
        = 2\ln \left(\frac{2}{\delta}\right) \mathrm{tr}\left[\rho f\left(\frac{\rho}{\pi}\right)\right]  - \ln(\pi)
        \approx 2\ln \left(\frac{2}{\delta}\right) \mathrm{tr}\left[\rho P\left(\frac{\rho}{\pi}\right)\right]  - \ln(\pi).
    \end{equation}
    We now evaluate the approximation error. 
    For this purpose, first note that $\rho \ln \rho$ can be represented as $\sum_{i: p_i > 0} p_i \ln(p_i) \ketbra{i}$ by the spectral decomposition of $\rho$ as $\rho = \sum_i p_i \ketbra{i}$, where $p_i$ are the eigenvalues of $\rho$. 
    Then, the polynomial approximation error of $f(x)$ is described as
    \begin{align}\label{eq:error-bound-of-xlog}
        & \mathrm{tr}\left[ \left|\rho f\left(\frac{\rho}{\pi}\right)  - \rho P\left(\frac{\rho}{\pi}\right) \right|\right]
         = \sum_{i:p_i > 0} \left|p_i f \left(\frac{p_i}{\pi }\right) - p_i P\left(\frac{p_i}{\pi }\right) \right|\notag \\
         = &\sum_{0<p_i/\pi  < \delta} \left|p_i f \left(\frac{p_i}{\pi}\right) - p_i P\left(\frac{p_i}{\pi}\right) \right|
        + \sum_{p_i/\pi \ge \delta} \left|p_i f \left(\frac{p_i}{\pi}\right) - p_i P\left(\frac{p_i}{\pi}\right) \right|.
    \end{align}
    From the assumption that $1/\kappa$ is a known lower bound on the smallest non-zero eigenvalue of $\rho$,
    we can null the first term in Eq.~\eqref{eq:error-bound-of-xlog} by setting
    $\delta = \frac{1}{\pi \kappa}$.
    Then, using the condition~\eqref{eq:log-polynomial} of $P$, the error is bounded by
    \begin{align}
        \left|\mathrm{tr}\left[\rho f\left(\frac{\rho}{\pi}\right) \right]- \mathrm{tr}\left[\rho P\left(\frac{\rho}{\pi}\right)\right] \right|
        \le \sum_{p_i/\pi \ge \delta} \left|p_i f \left(\frac{p_i}{\pi}\right) - p_i P\left(\frac{p_i}{\pi}\right) \right|
        \le \varepsilon_P.
    \end{align}
    The second error source is the block encoding of $\rho$.
    Let $W_P$ be a unitary for the $(1, 3,  4n_P \sqrt{\varepsilon_\rho})$-block encoding of $P(\rho/\pi)$. 
    Then we can bound the block encoding error as follows:
    \begin{align}
        \left|\mathrm{tr}\left[\rho P\left(\frac{\rho}{\pi}\right)\right]
        - \mathrm{tr}\left[(\ketbra{0}^{\otimes 3} \otimes \rho ) W_P\right]
        \right|&=\left|\mathrm{tr}\left[\rho P\left(\frac{\rho}{\pi}\right)\right]
        - \mathrm{tr}\left[\rho (\bra{0}^{\otimes 3}\otimes \bm{1}) W_P(\ket{0}^{\otimes 3}\otimes \bm{1})\right]
        \right|\notag\\
        &\leq \|\rho\|_1\left\|P\left(\frac{\rho}{\pi}\right)-(\bra{0}^{\otimes 3}\otimes \bm{1}) W_P(\ket{0}^{\otimes 3}\otimes \bm{1})\right\|\notag\\
        & \le 4 n_P \sqrt{\varepsilon_{\rho}}.
    \end{align}

    Now, we have a quantum circuit $W_{P}$ including $\mathcal{O}(n_P\log(1/\varepsilon_\rho)) =\mathcal{O}(\kappa\log(1/\varepsilon_P)\log(1/\varepsilon_\rho))$ uses of $\ketbra{0}\otimes \bm{1} +\ketbra{1}\otimes e^{-i\rho/2}$ (and its inverse) and $\mathcal{O}(n_P\log(1/\varepsilon_\rho))$ single- or two-qubit gates.
    Thus, using the general quantum algorithm in Section~\ref{sec:nearlyopt_qalg} for the Hadamard-test circuit Fig.~\ref{fig:entropy-estimation-circuit}, we can obtain samples of the estimator $\hat{\mu}$ for the following quantities:
    \begin{equation}
        \mathrm{tr} \left[(\ketbra{0}^{\otimes 3} \otimes \rho) W_P\right] 
    \end{equation}
    within additive error $\varepsilon_\Phi$.
    To bound the total error by $\varepsilon$ with a high probability, i.e., 
    \begin{equation}
        \mathrm{Pr}\left[\left|2\ln (2\pi \kappa )\hat{\mu}  - \ln(\pi) - S(\rho) \right| \le \varepsilon\right] \ge 2/3,
    \end{equation}
    it is sufficient to set
    $\varepsilon_P = 4n_P \sqrt{\varepsilon_\rho} = \varepsilon_\Phi = \mathcal{O}(\varepsilon/\log (\kappa))$
     and choose the number of measurements as $\mathcal{O}(1/\varepsilon_{\Phi}^2)$ to construct a single sample of $\hat{\mu}$.
    
    
    Based on the above analysis, we evaluate the complexity of the algorithm.
    The number $M$ of controlled-$e^{i\rho/2}$ in the circuit is obtained as
    \begin{equation}
        M = \mathcal{O}\left(n_P \log\left(\frac{1}{\varepsilon_\rho}\right)\right)
        = \mathcal{O}\left(\kappa \log\left(\frac{1}{\varepsilon_P}\right) \log\left(\frac{1}{\varepsilon_\rho}\right)\right)
        = \mathcal{O}\left(\kappa \log\left(\frac{1}{\varepsilon}\right) \log\left(\frac{\kappa}{\varepsilon}\right)\right),
    \end{equation}
    where we ignore $\log \log$ factors.
    The required number of identical copies of $\rho$ in the circuit is thus 
    \begin{equation}
    \mathcal{O}\left(M^2\frac{\log(M/\varepsilon_\Phi)}{\log\log(M/\varepsilon_\Phi)}\right)
    = \mathcal{O}\left(\kappa^2 \log^3\left(\frac{\kappa}{\varepsilon}\right)\log^2\left(\frac{1}{\varepsilon}\right)\right).
    \end{equation}
    Also, the number of measurements required for generating each $\hat{\mu}$ is given by
    \begin{equation}
        \mathcal{O}\left(\frac{1}{\varepsilon_{\Phi}^2}\right)=\mathcal{O}\left(\frac{\log^2(\kappa)}{\varepsilon^2}\right).
    \end{equation}
    The number of one- and two-qubit elementary gates are $\mathcal{O}(M^2 \log(M/\varepsilon)\log(d))$ for the virtual DME and $\mathcal{O}(M)$ for the eigenvalue transformation.
    Thus, the gate complexity is $\mathcal{O}(M^2 \log(M/\varepsilon)\log(d) + M) = \mathcal{O}(M^2 \log(M/\varepsilon)\log(d))$.
    
    The proof for the quantum relative entropy mostly follows that for the von Neumann entropy. 
    Let $\sigma=\sum_{i} p_i^{(\sigma)} \ketbra{i^{(\sigma)}}$ be the spectral decomposition of $\sigma$ and define $\sigma^0=\sum_{p^{(\sigma)}_i > 0} \ketbra{i^{(\sigma)}}$ be the projector on $\mathrm{supp}(\sigma)$. 
    Then we have $\sigma^0 \ln(\sigma) = \sum_{p^{(\sigma)}_i > 0} \ln(p_i^{(\sigma)}) \ketbra{i^{(\sigma)}}$. 
    Using $\sigma^0 \rho = \rho$ by the assumption that $\mathrm{supp}(\rho)\subseteq \mathrm{supp}(\sigma)$, we can ensure the polynomial approximation error similarly to Eq.~\eqref{eq:error-bound-of-xlog} as
    \begin{align}
         & \mathrm{tr}\left[ \left|\rho \sigma^0\left( f\left(\frac{\sigma}{\pi}\right) - P\left(\frac{\sigma}{\pi}\right) \right) \right|\right]
         =  \mathrm{tr}\left[ \left|\rho \sum_{p_i^{(\sigma)} > 0} \ketbra{i^{(\sigma)}} \left( f\left(\frac{p_i^{(\sigma)}}{\pi}\right) - P\left(\frac{p_i^{(\sigma)}}{\pi}\right) \right) \right|\right] \notag \\
         &\leq \left\| \sum_{p_i^{(\sigma)} > 0} \ketbra{i^{(\sigma)}} \left( f\left(\frac{p_i^{(\sigma)}}{\pi}\right) - P\left(\frac{p_i^{(\sigma)}}{\pi}\right) \right) \right\|=\max_{p_i^{(\sigma)}>0}\left| f\left(\frac{p_i^{(\sigma)}}{\pi}\right) - P\left(\frac{p_i^{(\sigma)}}{\pi}\right)\right|\leq \varepsilon_P.
    \end{align}
    Here, the final inequality holds by choosing $\delta=\frac{1}{\pi\kappa^{(\sigma)}}$.
    The remaining discussion is almost the same as that for the case of entropy. 
\end{proof}

\subsection{Proof of Theorem~\ref{thm:application_linearsyssolver}}

\begin{proof}[Proof of Theorem~\ref{thm:application_linearsyssolver}]
\label{sec:thm_proof_QLS}
    To explicitly construct $f_M$, we mainly focus on an efficient quantum algorithm developed in Ref.~\cite{lin2020optimal}.
    The core idea of the algorithm is sequentially applying projectors in order that a sequence of projected quantum states follows an adiabatic quantum evolution toward the solution vector $\ket{x}$.
    Specifically, we define the following time-dependent Hamiltonian on the target system and a single-qubit system
    \begin{equation}
        H(f(s)) :=(1-f(s))H_0+f(s)H_1,~~~s\in[0,1],
    \end{equation}
    where the initial and final Hamiltonians are given by
    \begin{equation}
        H_0:=X\otimes \left(\bm{1}-\ket{b}\bra{b}\right),~~~
        H_1:=\ket{0}\bra{1}\otimes A(\bm{1}-\ket{b}\bra{b})+\ket{1}\bra{0}\otimes (\bm{1}-\ket{b}\bra{b})A.
    \end{equation}
    Here, $f(s)$ is a scheduling function specified later.
    The Hamiltonian $H(f)$ always has the two-dimensional null space spanned by $\ket{1}\ket{b}$ and $\ket{0}\ket{x(f)}$ for a quantum state $\ket{x(f)}$ such that $ ((1 - f)\bm{1} + fA)\ket{x(f)} \propto \ket{b}$, and the other eigenvalues are separated from $0$ by a gap at least $1-f+f/\kappa$~\cite{lin2020optimal}.
    If one choose the initial state $\ket{0}\ket{x(0)}=\ket{0}\ket{b}$, then $\ket{1}\ket{b}$ does not appear in the adiabatic path due to the definition of $H(f)$.
    So, if we can follow the eigenstates $\ket{0}\ket{x(f)}$ from $\ket{0}\ket{b}$, we finally obtain the solution vector $\ket{0}\ket{x(1)}=\ket{0}\ket{x}$.

    Instead of implementing the time-dependent Hamiltonian dynamics, the previous method uses the projection for $\ket{x(f)}$ at each discretized time step.
    We first define the time segment $s_j=j/M'$ ($j=0,1,2,...,M'$) and $f_j:=f(s_j)$.
    To follow the adiabatic path at each $s_j$, we use a quantum gate $B(f_j,\Delta_j)$ that is an $(1,a',\Delta_j)$
    block encoding of the projector $\ket{x(f_j)}\bra{x(f_j)}$.
    We describe the construction of this gate with reflection gates $e^{-i\pi \ket{b}\bra{b}}$ later.
    By using these gates, the initial state $\ket{b}$ approximately propagates along the adiabatic path by the circuit in Fig.~\ref{fig_circ:linearsystemsolver1}.
    \begin{figure}[tb]
\centering
\begin{tabular}{c}
\\
~~~~~\Qcircuit @C=0.55em @R=1.5em {
  \lstick{\ket{0}^{\otimes a'}}&{/}\qw &\multigate{1}{B(f_1,\Delta_1)}&\meter& \push{\ket{0}^{\otimes a'}} &{/}\qw &\multigate{1}{B(f_2,\Delta_2)}&\meter&&\cdots&& \push{\ket{0}^{\otimes a'}} &{/}\qw &\multigate{1}{B(f_{M'},\Delta_{M'})}&\qw&\meter\\
  \lstick{\ket{b}}&{/}\qw&\ghost{B(f_1,\Delta_1)}&\qw&\qw&\qw &\ghost{B(f_2,\Delta_2)}&\qw&&\cdots&&&\qw&\ghost{B(f_{M'},\Delta_{M'})} &\qw\\}
\\
\\
\end{tabular}
\caption{Quantum circuit for the linear system solver in Ref.~\cite{lin2020optimal}. 
$B(f,\Delta)$ is a $\Delta$-precise block encoding of the projector $\ket{x(f)}\bra{x(f)}$.}
\label{fig_circ:linearsystemsolver1}
\end{figure}
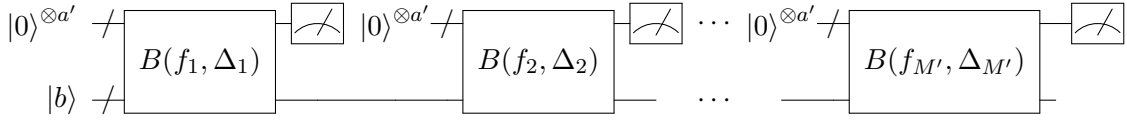
    The final unnormalized state of the circuit becomes 
    \begin{align}
        &\left[\langle 0|^{\otimes a'}B(f_{M'},\Delta_{M'})\ket{0}^{\otimes a'}\cdots \langle 0|^{\otimes a'}B(f_1,\Delta_1)\ket{0}^{\otimes a'}\right]\ket{b},
    \end{align}
    when the mid-circuit measurement outcomes are all zero. We write this unnormalized state as $c_p\ket{\tilde{x}}$ ($\|\ket{\tilde{x}}\|=1$).
    To determine the overall computational resources, we need to specify the scheduling function $f(s)$, $M'$, and $\Delta_j$.
    The authors of Ref.~\cite{lin2020optimal} suggest that 
    \begin{equation}
        f(s):=\frac{1-\kappa^{-s}}{1-\kappa^{-1}}
    \end{equation} 
    and $M'=\lceil 4\log^2(\kappa)/(1-\kappa^{-1})^2\rceil$, $\Delta_j=1/162(M')^2$ (if $j<M'$) and $\Delta_{M'}=\varepsilon^2/4$.
    Under these choices, the projection probability $|c_p|^2$ is lower bounded by a constant value, and the overlap between the final state $\ket{\tilde{x}}$ and the solution $\ket{x}$ satisfies~\cite{lin2020optimal}
    \begin{equation}
        |\bra{x}\tilde{x}\rangle|>1-\varepsilon^2.
    \end{equation}

    To complete the description of the entire circuit (i.e., $U_M$ in $f_M$ for our algorithm), we here construct the unitary gates $B$ as follows.
    First, we construct block encodings $W_j$ for $H_j$ ($j=0,1$) using the reflection gate with respect to $\ket{b}$ in Fig.~\ref{fig_circ:linearsystemsolver2}.
\begin{figure}[tbp]
  \centering
  \begin{minipage}{0.4\textwidth}
    \centering
    \begin{tabular}{c}
      \Qcircuit @C=0.55em @R=1em {
        \lstick{\ket{0}} & \gate{H} & \ctrl{1} & \gate{H} \\
                         & {/}\qw   & \gate{e^{-i\pi \ket{b}\bra{b}}} & \qw \\
                         & \qw      & \gate{X} & \qw \\
      }
    \end{tabular}
  \end{minipage}
  \begin{minipage}{0.4\textwidth}
    \centering
    \begin{tabular}{c}
      \Qcircuit @C=0.55em @R=1em {
        \lstick{\ket{0}} & \gate{H} & \qw & \qw&\ctrl{2}&\gate{H}&\qw\\
        \lstick{\ket{0}} & \gate{H} & \ctrl{1}&\qw&\qw&\gate{H}&\qw\\
        &\qw&\ctrl{1} &\gate{X}&\ctrl{1}&\qw&\qw\\
        &\qw&\gate{e^{-i\pi\ket{b}\bra{b}}}&\multigate{1}{U_A}&\gate{e^{-i\pi\ket{b}\bra{b}}}&\qw&\qw\\
        \lstick{\ket{{0}}^{\otimes a}}&\qw&\qw&\ghost{U_A}&\qw&\qw&\qw
      }
      \\
      \\
    \end{tabular}
  \end{minipage}
  \caption{Quantum circuits for block encoding of $H_j$ ($j=0,1$). Right (Left) circuit is a block encoding of $H_1$ ($H_0$).}
\label{fig_circ:linearsystemsolver2}
\end{figure}
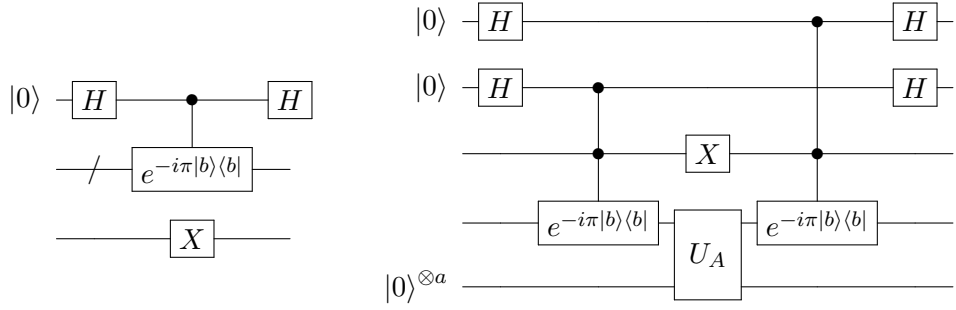
    Then, combining the controlled version of these unitaries with the linear combination of unitaries (LCU) method~\cite{childs2012hamiltonian,berry2015hamiltonian}, we obtain a quantum circuit for a block encoding of $H(f(s))$.
    Furthermore, using the quantum eigenvalue transformation of $H(f(s))$ with an $\mathcal{O}((1-f+f/\kappa)^{-1}\log(1/\Delta))$-degree polynomial, we obtain the unitary gate $B$ (see Theorem~3 in Ref.~\cite{lin2020optimal}), and this uses $\mathcal{O}((1-f+f/\kappa)^{-1}\log(1/\Delta))$ controlled reflection gates $e^{-i\pi\ket{b}\bra{b}}$.
    We note that in the circuit $B$, the number of one-qubit or two-qubit elementary gates is $\mathcal{O}(a(1-f+f/\kappa)^{-1}\log(1/\Delta))$ (except for the reflection gates of $\ket{b}$) and the number of queries to $U_A$ and its inverse is $\mathcal{O}((1-f+f/\kappa)^{-1}\log(1/\Delta))$.

    From the above discussion, we observe that there exists a function $f_M$ such that
    \begin{equation}
        f_M(\ketbra{1}\otimes \ketbra{b})=|c_p|^2\bra{\tilde{x}}O\ket{\tilde{x}}.
    \end{equation}
    where the number $M$ of density matrix exponentials for $\ket{1}\ket{b}$ is calculated as
    \begin{equation}
        M=\sum_{j=1}^{M'-1} \mathcal{O}\left(\log(1/\Delta_j)\frac{1}{1-f_j+f_j/\kappa}\right)+\mathcal{O}(\kappa\log(1/\varepsilon))=\mathcal{O}(\kappa(\log(\kappa)\log\log(\kappa)+\log(1/\varepsilon))).
    \end{equation}
    Here, we use $\ket{{1}}\ket{b}$ as input to $f_M$ because we need to use the controlled reflection gates.
    The quantity $f_M(\ket{1,b}\bra{1,b})$ can be obtained within $\mathcal{O}(\varepsilon)$ error with a high probability by our general algorithm in Section~\ref{sec:nearlyopt_qalg}.
    Similarly, we can obtain $|c_p|^2$ by simply setting $O=\bm{1}$ in the function $f_M$.
    So, by taking the ratio of these estimates, we obtain an estimate $\bra{\tilde{x}}O\ket{\tilde{x}}$ within $\mathcal{O}(\varepsilon)$ error with a high probability.
    This quantity is also $\mathcal{O}(\varepsilon \|O\|)$ close to the target value $\bra{x}O\ket{x}$ due to $|\langle x\ket{\tilde{x}}|>1-\varepsilon^2$. 
    We note that the variance of the final ratio estimator is amplified by only a constant value in the final division, so the measurement cost in our general algorithm is still $\mathcal{O}(1/\varepsilon^2)$.
    The gate complexity directly follows from the discussion in Section~\ref{sec:nearlyopt_qalg}.
\end{proof}

\subsection{Proof of Theorem~\ref{thm:error-mitigation}}
\label{sec:proof-error-mitigation}

    \begin{proof}[Proof of Theorem~\ref{thm:error-mitigation}]
    We first provide an overview of our algorithm. 
    The algorithm utilizes the circuit in Fig.~\ref{fig:mitigation-circuit}.
    $W_P$ is a block encoding of the projector onto the eigenspace of $\rho$ with eigenvalue $1-\lambda$. 
    $W_P$ is constructed from the time-evolution $e^{i\rho/2}$ for an unknown target state $\rho$ and the well-established eigenvalue transformation technique.
    In this proof, we first show the construction of $W_P$ and then analyze the algorithmic error and its complexity.

    To efficiently construct the projector via time-evolution, we employ the technique of quantum eigenvalue transformation of unitary matrices with real polynomials (QETU)~\cite{Dong2022-qetu}. 
    QETU can directly and efficiently construct the projector from the time-evolution $e^{i\rho/2}$, without using the intermediate step of the block encoding $\rho$ shown in the proof of Theorem~\ref{thm:entropy-estimation} in Section~\ref{sec:proof-entropy}.
    Let $P(x)$ be an even real polynomial of degree $n_P$ satisfying $|P(x)| \le 1$ for $x \in [-1,1]$. We can construct the block encoding of $P(\cos(\frac{\rho}{2}))$ with $n_P$ uses of $\ketbra{0}\otimes e^{i\rho/2}+\ketbra{1}\otimes e^{-i\rho/2}$~\cite[Corollary 17]{Dong2022-qetu}.
    By leveraging the result of polynomial approximating the rectangular function
    ~\cite[Corollary 7]{low2017-spectral-amplification} or~\cite[Lemma 29]{gilyen2019-qsvt},
    we can readily obtain the even real polynomial $P(x)$ of degree $n_P=\mathcal{O}(\delta^{-1}\log(1/\varepsilon_P))$,
    which approximates $\Theta(\Delta-x)
    = \frac{1}{2}(-\mathrm{sgn}(x-\Delta) + 1)$ in $x>0$ and satisfies the following:
    \begin{gather}
        |P(x) - \Theta(\Delta-x)| \leq \varepsilon_P; ~~  \forall x \in \left[0, \Delta \right] \cup \left[\Delta + {\delta}, 1\right], \\
        |P(x) | \leq 1; ~~ \forall  x \in [-1,  1].
    \end{gather}
    By setting $\delta = \cos(\frac{\eta}{2}) - \cos(\frac{1-\eta}{2}) = \Omega(1- 2\eta)$  and  $\Delta =\cos(\frac{1-\eta}{2})$,
    $P(\cos(\frac{\rho}{2}))$ approximates the desired projector $\ketbra{\psi}$.
    For simplicity, we define $P := P(\cos(\rho/2))$ and $\Braket{O} := \bra{\psi} O\ket{\psi}$.
    Therefore, QETU provides a $(1,1,0)$-block encoding $W_P$ and $\mathrm{tr}_a[(\ketbra{0}_a) W_P(\ketbra{0}_a\otimes\rho) W_P^\dagger] = P\rho P$.

    Next, we see that $\frac{\mathrm{tr}\left[OP \rho P\right]}{\mathrm{tr}\left[P \rho P\right]}$ provides a good estimate of $\Braket{O}$. 
    Using
    \begin{equation*}
       P=P(\cos(\rho/2))
        =P\left( \cos\left(\frac{1-\lambda}{2} \ketbra{\psi} + \frac{\lambda}{2}\rho_{\rm err} \right)\right)
       =P\left( \cos\left(\frac{1-\lambda}{2} \right)\right) \ketbra{\psi} 
          + P\left( \cos\left(\frac{\lambda}{2}\rho_{\rm err} \right)\right)
    \end{equation*}
    and assuming that $ \varepsilon_P \le 1/2$, we can bound the polynomial error as
    \begin{align}
        \left|\frac{\mathrm{tr}\left[OP \rho P\right]}{\mathrm{tr}\left[P \rho P\right]} - \Braket{O} \right| 
        & = \left|\frac{(1-\lambda) (1 +\varepsilon_1)^2 \Braket{O} + \lambda \mathrm{tr}\left[OP \rho_{\mathrm{err}} P\right]}{(1-\lambda) (1 +\varepsilon_1)^2 + \lambda \mathrm{tr}\left[P \rho_{\mathrm{err}} P\right]} - \Braket{O} \right| \notag \\
        & = \frac{ \lambda \left| \mathrm{tr}\left[OP \rho_{\mathrm{err}} P\right] - \Braket{O} \mathrm{tr}\left[P \rho_{\mathrm{err}} P\right]\right| }{(1-\lambda) (1 +\varepsilon_1)^2 + \lambda \mathrm{tr}\left[P \rho_{\mathrm{err}} P\right]}  \notag \\
        & \le  \frac{\lambda}{(1-\lambda)(1 - \varepsilon_P)^2} \left| \mathrm{tr}\left[OP \rho_{\mathrm{err}} P\right] - \Braket{O} \mathrm{tr}\left[P \rho_{\mathrm{err}} P\right]\right| \notag \\
        & \le  \frac{\lambda}{1-\lambda} \frac{1}{(1-\frac{1}{2})^2} 2 \|O\| \varepsilon^2_P \le  8 \|O\| \varepsilon^2_P,
    \end{align}
    where $\varepsilon_1 =  P(\cos((1-\lambda)/2)) - 1$.

    Now, we have a quantum circuit $W_{P}$ including $n_P=\mathcal{O}((1-2\eta)^{-1}\log(1/\varepsilon_P))$ uses of $\ketbra{0}\otimes e^{i\rho/2}+\ketbra{1}\otimes e^{-i\rho/2}$ and $\mathcal{O}(n_P)$ single-qubit gates.
    Thus, using the general quantum algorithm in Section~\ref{sec:nearlyopt_qalg}, we can obtain samples of estimators $\hat{\mu}_A$ and $\hat{\mu}_B$ for the following quantities:
    \begin{equation}
        {\rm tr}[(\ketbra{0}\otimes O)W_P(\ketbra0\otimes\rho)W_P^\dagger]={\rm tr}[OP\rho P],~~~{\rm tr}[(\ketbra{0}\otimes \bm{1})W_P(\ketbra0\otimes\rho)W_P^\dagger]={\rm tr}[P\rho P],
    \end{equation}
    within the additive error $\varepsilon_A,\varepsilon_B$ with a high probability, respectively.
    Then, we calculate $\hat{\mu}_A/\hat{\mu}_B$ to serve it as the final estimate for $\Braket{O}$.
    By taking 
    $\varepsilon_A=\mathcal{O}(\mathbb{E}[\hat{\mu}_B]\varepsilon), \varepsilon_B=\mathcal{O}(\mathbb{E}[\hat{\mu}_B]\varepsilon/\|O\|))$ and appropriately choosing the number of measurements to construct a single sample of $\hat{\mu}_A$ and $\hat{\mu}_B$ (more specifically, the number of samples of $\hat{g}'$; see Section~\ref{sec:nearlyopt_qalg}), we can say that 
    \begin{equation}
        {\rm Pr}\left[\left|\frac{\hat{\mu}_A}{\hat{\mu}_B}-\frac{{\rm tr}[OP\rho P]}{{\rm tr}[P\rho P]}\right|\leq\varepsilon\right]\geq 2/3
    \end{equation}
    holds. 
    It is sufficient to take the number of measurements $\mathcal{O}(\|O\|^2/(\mathbb{E}[\hat{\mu}_B]\varepsilon)^2)$ to obtain such estimates $\mu_A,\mu_B$.
    Therefore, by setting $\varepsilon^2_P=\mathcal{O}(\varepsilon/\|O\|)$, we can conclude that the estimator $\hat{\mu}_A/\hat{\mu}_B$ yields an estimate that is $\mathcal{O}(\varepsilon)$-close to $\Braket{O}$ with a high probability.
    
    From the above parameters, we obtain the number of calls of controlled $e^{i \rho/2}$ and its inverse:
    \begin{equation}
        M =  \mathcal{O}\left(n_P \right) = \mathcal{O}\left( \frac{1}{1 - 2\eta}\log\left(\frac{\|O\|}{\varepsilon}\right)\right), 
    \end{equation}
    and the number of copies of $\rho$ per circuit as
    \begin{equation}
        \mathcal{O}\left(M^2\frac{\log(M/\varepsilon)}{\log\log(M/\varepsilon)}\right)
        = \mathcal{O}\left(\frac{1}{(1 - 2\eta)^2}\log^2\left(\frac{\|O\|}{\varepsilon}\right)\log\left(\frac{1}{(1-2 \eta) \varepsilon}\right) \right),
    \end{equation}
    where we have ignored the $\log \log$ factors.
    The number of measurements required for generating $\hat{\mu}_{A}$ or $\hat{\mu}_{B}$ is given by 
    \begin{equation}
        \mathcal{O} \left(\frac{\|O\|^2}{(\mathbb{E}[\hat{\mu}_B] \varepsilon)^2}\right) = 
        \mathcal{O} \left(\frac{\|O\|^2}{(\mathrm{tr}\left[P \rho P\right] \varepsilon)^2}\right)
        = 
        \mathcal{O} \left(\frac{\|O\|^2}{(1-\eta)^2\varepsilon^2}\right).
    \end{equation}
    The number of one- and two-qubit elementary gates are $\mathcal{O}(M^2 \log(M/\varepsilon)\log(d))$ for the virtual DME and $\mathcal{O}(n_P)$ for the eigenvalue transformation.
    Thus, the gate complexity is $\mathcal{O}(M^2 \log(M/\varepsilon)\log(d) + n_P) = \mathcal{O}(M^2 \log(M/\varepsilon)\log(d))$.

    \end{proof}

\bibliographystyle{unsrt}
\bibliography{ref}

\end{document}